%% file: main.tex
\documentclass[11pt,letterpaper]{article}
\usepackage{ramp}

\def\authornotes{1}

\include{macros}

\begin{document}

\title{Semidefinite programs simulate approximate message passing robustly}
\author{
Misha Ivkov\thanks{Stanford University. \texttt{mishai@stanford.edu}} \and Tselil Schramm\thanks{Stanford University.  \texttt{tselil@stanford.edu}.}
}
\date{\today}
\maketitle

\input{abstract-arxiv}

\setcounter{tocdepth}{2}
\setcounter{page}{-1}
\thispagestyle{empty}
\newpage
\tableofcontents
\thispagestyle{empty}
\thispagestyle{empty}
\newpage
\setcounter{page}{1}

\input{intro-arxiv}

\input{overview-arxiv}

\input{prelims-arxiv}

\input{robust-arxiv}

\section*{Acknowledgments}
We would like to thank David Steurer, Andrea Montanari, Kangjie Zhou, Sam Hopkins, Sidhanth Mohanty, and Yuchen Wu for helpful conversations.
This work was supported by T.S.'s NSF CAREER award \# 2143246 and M.I.'s NSF Graduate Research Fellowship.
We thank the Simons Institute for their hospitality during the Fall 2021 program on the ``computational complexity of statistical inference,'' where part of this work took place.

\bibliographystyle{alpha}
\bibliography{main}

\appendix

\addcontentsline{toc}{section}{Appendices}

\input{appendix-arxiv}

\input{appendix-approx-arxiv}

\end{document}

%% file: macros.tex
\ifnum\authornotes=1
    \newcommand{\tselil}[1]{\footnote{\color{ForestGreen}Tselil: {#1}}}
    \newcommand{\misha}[1]{\footnote{\color{Orange}Misha: {#1}}}

    \newcommand{\Tnote}[1]{{\color{ForestGreen}[Tselil: #1]}}
    \newcommand{\mnote}[1]{{\color{Orange}[Misha: #1]}}
	\newcommand{\todo}[1]{{\color{Red} (TODO: #1)}}
\else
    \newcommand{\tselil}[1]{}
    \newcommand{\misha}[1]{}

    \newcommand{\Tnote}[1]{}
    \newcommand{\mnote}[1]{}

	\newcommand{\todo}[1]{}
\fi

\newcommand{\parencite}[1]{\cite{#1}}

\newcommand{\iprod}[1]{\left\langle #1 \right\rangle}
\newcommand{\Iprod}[1]{\langle #1 \rangle}

\usepackage{scalerel}

\newcommand{\LSH}{\text{LStH}\xspace}
\newcommand{\AMP}{AMP\xspace}
\newcommand{\sos}{SoS\xspace}

\newcommand{\lsh}{\mathrm{\LSH}}
\newcommand{\amp}{\mathrm{\AMP}}
\newcommand{\defeq}{:=}

\newcommand{\vSOS}{v_{\LSH}}
\newcommand{\vAMP}{v_{\mathrm{AMP}}}
\newcommand{\vamp}{v_{\mathrm{AMP}}}

\newcommand{\ampopt}{\mathrm{OPT}_\amp}
\newcommand{\optamp}{\mathrm{OPT}_\amp}

\newcommand{\rob}{\mathrm{robust}}
\newcommand{\sla}{\mathrm{slack}}
\newcommand{\Xprox}{{\hat{X}}}

\newcommand{\tX}{\Xprox}

\DeclareMathOperator{\op}{\mathsf{op}}
\DeclareMathOperator{\plim}{\operatornamewithlimits{p-lim}}

\newcommand{\trunk}{\mathsf{k}}

%% file: abstract-arxiv.tex
\begin{abstract}
    Approximate message passing (AMP) is a family of iterative algorithms that generalize matrix power iteration.
    AMP algorithms are known to optimally solve many average-case optimization problems.
    In this paper, we show that a large class of AMP algorithms can be simulated in polynomial time by \emph{local statistics hierarchy} semidefinite programs (SDPs), even when an unknown principal minor of measure $1/\polylog(\mathrm{dimension})$ is adversarially corrupted.
    Ours are the first robust guarantees for many of these problems.
    Further, our results offer an interesting counterpoint to strong lower bounds against less constrained SDP relaxations for average-case max-cut-gain (a.k.a. ``optimizing the Sherrington-Kirkpatrick Hamiltonian'') and other problems.
    
    \end{abstract}

%% file: intro-arxiv.tex
\section{Introduction}

Approximate Message Passing (AMP) is a family of algorithms which generalize matrix power iteration.
AMP is so named because it is a dense variant of the ``Belief propagation'' message-passing algorithm, with origins in statistical physics \parencite{Bolt14,DMM09,BM11}.
Since its early use in the context of compressed sensing \parencite{DMM09}, AMP has become widely studied in the high-dimensional statistics community and has found extensive applications, including sparse principal components analysis (PCA) \parencite{DM14}, linear regression \parencite{DMM09,BM11,KMSSZ12}, non-negative PCA \parencite{MR15}, and the recent breakthrough algorithm for finding the ground state of the Sherrington-Kirkpatrick Hamiltonian (an average-case version of Max-Cut-Gain) \parencite{Mon21}.
The surveys \parencite{Mon12,FVRR22} contain a wealth of additional examples.

To describe the AMP algorithm, we consider the illustrative example of the {\em non-negative principal components analysis} problem (non-negative PCA or nnPCA). 
In non-negative PCA, we observe a symmetric $n \times n$ matrix $X$, and our goal is to optimize the objective $\max \, \{ v^\top X v \mid v \in \R^n, v \ge 0, \|v\| \le 1\}$.
AMP produces a sequence of iterates $v_0,\ldots,v_t \in \R^n$ with the goal that $v_t$ has large objective value, where each iteration combines {\em matrix-vector multiplication} to amplify correlation with $X$ and the application of {\em denoiser functions} that allow us to enforce constraints on our output: that is,  $v_{s+1} \propto f(X v_s)$ for some $f: \R^n \to \R^n$.
The choice of $f$ falls to the algorithm designer, but a good choice in the context of nnPCA is entry-wise thresholding at zero, $f(u)_i = \max(u_i,0)$, which enforces that $v_{s+1} \ge 0$.

AMP enjoys a number of strengths:
It is simple to implement and extremely efficient (provided that the chosen denoisers $f$ can be applied efficiently).
For many spiked matrix models, AMP is known to perform well, achieving the information-theoretic minimum mean squared error (i.e. AMP is Bayes-optimal, see e.g. the discussion in \parencite{FVRR22}).
In fact, AMP has such a prominent place in the high-dimensional statistics community that a lower bound against AMP is considered to provide evidence for computational intractability \parencite{CM22,CMW20}. 
The major problem with AMP is that it is brittle to the model specification; AMP is known to have poor robustness.
Formal guarantees for AMP algorithms are known for a variety of problems, but typically one requires that the input $X$ have the form $X = Y + Z$ for $Y$ a simple planted structure (potentially $Y = 0$) and $Z$ a random matrix with {\em independently sampled} subgaussian entries.
The independence assumption on the noise is crucial for the success of the algorithm, and minor perturbations can cause the algorithm to behave unstably \parencite{CZK14,RSFS19}. 
The first aim of this paper is to study the following question:
\begin{center}
{\em
Can AMP be supplanted by a polynomial-time algorithm which is robust to adversarial noise?}
\end{center}

In particular, we wish to understand whether sum-of-squares semidefinite programs (SDPs) and related algorithms can simulate AMP, with the additional benefit of robustness.
SDPs are a natural choice for two reasons.
Firstly, SDPs represent the most powerful polynomial-time algorithms we know in many contexts, such as worst-case approximation algorithms (e.g. \parencite{Rag08,ARV09}) and several problems in statistical estimation (e.g \parencite{BM16}, see also \parencite{RSS18} for a survey).
Additionally, SDP algorithms are often quite robust, enabling a recent renaissance in algorithmic robust statistics (see \parencite{HL18,KSS18}, and the aftermath).

However, in the context of random optimization problems, the supremacy of SDPs is yet uncertain.
There are several random optimization problems where AMP is known to succeed while the natural semidefinite programming relaxation is known to fail, such as non-negative PCA \parencite{BKW22} and optimizing the Sherrington-Kirkpatrick (SK) Hamiltonian.
This latter example is particularly dramatic: the SK problem asks us to find $\argmax\{ v^\top X v \mid v \in \{\pm \tfrac{1}{\sqrt{n}}\}^n\}$ for $X$ an $n \times n$ matrix with independent $\calN(0,\frac{1}{n})$ entries.
The celebrated works of Parisi and Talagrand \cite{Par80,Tal06} show that with high probability, the true value is $\approx 1.52$, and in a recent breakthrough Montanari showed that (modulo a widely believed conjecture) an AMP algorithm achieves value $1.52-\eps$ in time $C(\eps) \cdot n^2$, for $C(\eps)$ a function depending only on $\eps$ \parencite{Mon21}.
For the basic semidefinite programming relaxation of this problem, there is an integrality gap of value $2$ \parencite{MS16}, and it was recently shown that this gap persists even after $n^{\Omega(1)}$ rounds of the sum-of-squares hierarchy \parencite{GJJPR20}. 
This is surprising, given that results from hardness of approximation for worst-case CSPs suggests that SDPs may be optimal among polynomial-time algorithms \cite{Rag08}.
The second aim of this paper is to understand:
\begin{center}
{\em
Are SDPs really worse than other algorithms for average-case optimization problems?
}
\end{center}

One way to reconcile the apparent weakness of SDPs in this setting is as follows: the natural SDP relaxations are {\em certifying } an upper bound on the value of the maximization problem in question, which may be a computationally harder problem than merely finding a $(1-\eps)$-optimal solution (see e.g. \parencite{BKW22,BBKMW21}).
If one artificially plants a larger-valued solution in the SK model or nnPCA, AMP is not guaranteed to find it; on the other hand, because the natural SDP is a convex relaxation of the original optimization problem, the SDP value has to reflect the presence the planted solution.

To address this problem, Banks, Mohanty, and Raghavendra introduce a family of semidefinite programming relaxations, called the {\em Local Statistics Hierarchy} (\LSH) \parencite{BMR21}.
The \LSH is based on the sum-of-squares (\sos) hierarchy, but rather than relaxing an optimization problem, it is a feasibility program which incorporates prior information about the joint distribution over $X,v^*$ for $v^*$ the desired solution.\footnote{\parencite{BMR21} do not study optimization; rather, there is some ``planted'' solution $v^*$ which they are trying to recover.}
\cite{BMR21} study \LSH in the context of the stochastic block model; they are unable to prove that \LSH can estimate $v^*$, but they do prove that polynomial-time \LSH can {\em hypothesis test} between $X$ drawn from the stochastic block model and $X$ drawn from a null distribution without community structure.\footnote{In some sense, this is the natural ``decision'' variant of the problem, where estimating $v^*$ is the natural ``search'' variant.}
Their work leaves open the intriguing question of whether \LSH can redeem SDPs more broadly as a best-in-class algorithm for random optimization and estimation problems.

\medskip

In this paper, we will show that in the context of random optimization problems, under mild conditions the \LSH can simulate AMP in polynomial time, estimating $v^*$ as well as the corresponding AMP algorithm while also being {\em robust} to adversarial perturbations.
Namely, if any $n/\polylog n$-sized principal minor of $X$ is adversarially corrupted, then the \LSH SDP run on $X$ can be rounded to a solution $v_{\lsh}$ which approximates $v_{\amp}$ in $\ell_2$-norm, for $v_{\amp}$ the solution output by AMP on the uncorrupted instance.
Our result captures the contexts in which AMP is known to succeed while the standard SDP relaxation fails: non-negative PCA and the SK model.
This redeems SDPs as at least as powerful as other algorithms when it comes to random optimization and estimation problems; further, it is the first demonstration of the power of \LSH for robust estimation. 

\subsection{Our results}

Our main result is a meta-theorem stating that if we have an AMP algorithm for a random quadratic optimization problem satisfying certain conditions, then the degree-$O(1)$ \LSH can simulate this algorithm, and robustly. 
In order to make our statement precise, we must first give some definitions.

We begin with the class of quadratic optimization problems to which our results apply.
\begin{definition}[Random quadratic optimization problem]\label{def:rqo}
Let $n \in \Z_+$, let $\calD$ be a symmetric subgaussian distribution over $\R$ with variance $\frac 1n$ and $\E_{X \sim \calD}[|X|^\ell] \le O(\frac{\sqrt{\ell}}{n})^{\ell}$ for each even $\ell \in \Z_+$, and $\calK \subset \R^n$.
We define the {\em random quadratic optimization problem} $\calP_n(\calD,\calK)$ as follows. 
A problem instance from $\calP_n(\calD,\calK)$ is sampled by choosing a matrix $X \in \R^{n\times n}$ with $X_{ij} = X_{ji} \sim \calD$ independently for all $i,j \in [n]$
(we write $X \sim \calP_n(\calD,\calK)$ for short).
The goal is to find $v^* = \argmax\{v^\top X v \mid v \in \calK\}$, or a $\delta$-approximate solution $u \in \R^n$ satisfying $|\iprod{u,v^*}| \ge (1-\delta)\|u\|\|v^*\|$.\footnote{
A $\delta$-approximate solution $u$ automatically implies an additive approximation guarantee for the objective, since if $u = \pm(1-\delta) v^* + u^{\perp}$  (with $\|u\|= \|v^*\|=1$), then $u^\top X u \ge (1-\delta)^2(v^*)^\top X v^* - O(\sqrt{\delta})\|X\|_{op}$. 
Typically $\|X\|_{op} = \Theta(\max_{v} v^\top X v)$.}
\end{definition}

For example, the above-mentioned Sherrington-Kirkpatrick problem is (up to rescaling) the distribution $\calP_n(\calN(0,\frac{1}{n}),\{\pm \frac{1}{\sqrt{n}}\}^n)$, and non-negative PCA is $\calP_n(\calN(0,\frac{1}{n}),\bbS^{n-1}_{\ge 0})$, where $\bbS^{n-1}_{\ge 0}$ denotes the set of non-negative unit vectors.
The entries of $X$ are normalized to have magnitude $\Theta(\frac{1}{\sqrt n})$ because at this scale the objective value tends to be $\Theta(1)$.

Note that we are interested in the {\em search/estimation} version of the problem rather than the problem of determining the objective value---this is because the objective values of these optimization problems concentrate extremely well, so the deterministic expected optimum gives a near-optimal estimate of the objective value with high probability.

Now, we formally define AMP algorithms.
\begin{definition}[AMP algorithm]\label{def:amp-intro}
An {\em AMP algorithm} $\calA$ is an algorithm specified by a sequence of deterministic {\em denoising} functions $\calF = f_1,f_2\ldots$, with $f_s: \R^{s+1} \to \R$ for all $s \ge 1$.
Given an input matrix $X \in \R^{n \times n}$ and a number of iterations $t \in \Z_+$, $\calA$ outputs a sequence of iterates $v_1,\ldots,v_t$ according to the rule $v_0 = \vec{1}$, and
\[
v_{s+1} = f_{s+1}(X v_s, v_s, v_{s-1},\ldots,v_0) - \Delta_s(v_{s-1},\ldots,v_0).
\]
where $f_{s+1}$ is applied coordinate-wise; that is, for $u_{s+1},\ldots,u_0 \in \R^n$, $f_{s+1}(u_{s+1},\ldots,u_0) \in \R^n$ and for each $i \in [n]$, $(f_{s+1}(u_{s+1},\ldots,u_0))_i = f_{s+1}(u_{s+1}(i),\ldots,u_0(i))$.
The function $\Delta_s$ is the so-called Onsager correction term, which is determined by $\calF$ and is included so as to decrease correlation between the iterates (see \pref{def:amp-instance} for details).
\end{definition}
This definition of an AMP algorithm is more restrictive than the broadest definition one sees in the AMP literature; in particular, we have required that $X$ be $n\times n$, that $f_s$ be univariate functions applied entrywise, and that $v_0 = \vec{1}$.
We will require the latter two conditions in our proofs, but we do not find these to be too restrictive because almost every theoretical guarantee for an AMP algorithm stipulates the same restrictions.\footnote{In the literature $v_0$ is usually an arbitrary starting point independent of $X$; restricting to $v_0=\vec{1}$ is almost without loss of generality, since the $X$ in the literature are typically Gaussian and therefore rotationally invariant.}
The symmetry of $X$ is a condition which could almost certainly be removed, but making this assumption makes the proofs more convenient.

Finally, we introduce the Local Statistics Hierarchy (\LSH) of \cite{BMR21}.
The idea of \LSH is as follows: suppose we know the joint distribution $\calP_{X,v}$ over matrix-vector pairs $(X,v)$, and observing $X$, we wish to estimate the marginal over $v$.
The \LSH combines a \sos program searching for feasible {\em pseudoexpectations} over $v$ with constraints consistent with prior information about $\calP$.

\begin{definition}[Local Statistics Hierarchy, Non-Robust Version]\label{def:LSH}
Let $\calP$ be a distribution over $(X,v) \in \R^{n \times n} \times \R^n$, and let $v^{\le d}$ denote the set of all monomials in $v$ of degree at most $d$.
Call a polynomial $q(X,v)$ {\em $\calS_n$-symmetric} if $q(X,v) = q(\Pi X \Pi^\top, \Pi v)$ for any $n\times n$ permutation matrix $\Pi$.

The {\em degree-$(d_X,d_v)$ Local Statistics Hierarchy} with input $X$ is a semidefinite program which returns a linear operator $\pE: v^{\le d_v} \to \R$ which satisfies the following constraints:
\setlist{nolistsep}
\begin{enumerate}[noitemsep]
\item Scaling: $\pE 1 = 1$.
\item Positivity: for any polynomial $p$ of degree at most $d_v/2$ in $v$, $\pE p(v)^2 \ge 0$
\item Prior matching: for any $\calS_n$-symmetric polynomial $q$ of degree at most $d_X$ in $X$ and $d_v$ in $v$, 
\[
\pE q(X,v) = \E_{(X',v') \sim \calP} q(X',v') \pm C \sqrt{\Var_{(X',v') \sim \calP} q(X',v')},
\]
for $C>0$ chosen so the constraint is satisfied with high probability when $(X^*,v^*) \sim \calP$ and on the left one plugs in $X = X^*, v = v^*$.
\end{enumerate}
\end{definition}
One can implement the degree-$(a,b)$ \LSH with an SDP with $O(n^a)$ variables and $O(n^a) + (a+b)^{O(a+b)}$ linear constraints by placing constraints on a finite basis of $\calS_n$ symmetric polynomials (see \pref{sec:alg} for more precise details).
So as long as $a,b = O(1)$, the \LSH algorithm is polynomial-time.
For our robust result, we will use a version of this program with more variables and constraints; we give a high-level description in \pref{sec:overview} and definition in \pref{sec:alg}.

Lastly, we define a class of adversarial perturbations against which our algorithms are robust:
\begin{definition}[adversarial $\eps$-principal minor perturbation]
We say $Y \in \R^{n \times n}$ is an {\em adversarial $\eps$-principal minor perturbation of the matrix $X$} if $X-Y$ is supported on some $\eps n \times \eps n$ principal minor.
\end{definition}
In \pref{sec:robust-model} we comment more on this notion of robustness: we give (i) an example which demonstrates that AMP is not robust to this type of corruption, and (ii) an information-theoretic obstacle to achieving the guarantees of AMP in the so-called ``strong contamination model.''

We are finally ready to state our theorem.

\begin{theorem}[Main theorem, informal]\label{thm:main-intro}
Suppose $X \sim \calP_n(\calD,\calK)$ is an instance of a random quadratic optimization problem, $\calA$ is an AMP algorithm with degree-$k$ polynomial denoiser functions, and $v_{\amp}$ is the $t$'th iterate of $\calA$ on input $X$.
Then there exists an integer $d = O({2k}^{2t})$ such that for any $\eta = \omega(1/\sqrt{n})$, the degree-$(d,2)$ Local Statistics Hierarchy on $X$ can be rounded to a vector $v_{\lsh}$ which satisfies $|\iprod{v_{\lsh},v_{\amp}}| \ge (1-\eta)\|v_{\lsh}\|\|v_{\amp}\|$ with probability $1-o(1)$.

Further for any $\eps \ge 0$ (allowing $\eps \to 0$ as $n \to \infty$), even when given as input an adversarial $\eps$-principal minor perturbation $Y$ of $X$, the degree-$(d,d)$ {\em robust} local statistics hierarchy can be rounded to a vector $v_{\lsh}(Y)$ satisfying $|\iprod{v_{\lsh}(Y), v_{\amp}(X)}| \ge (1-\eps^{1/2}\cdot (d\log n)^{O(d)}-\eta)\|v_{\lsh}(Y)\|\|v_{\amp}(X)\|$ with probability $1-o(1)$, for $v_{\amp}(X)$ the output of AMP on the {\em uncorrupted} $X$. 
\end{theorem}	

It is folklore that ``nice'' AMP denoisers are well-approximated by polynomials of bounded degree. 
A variety of formalizations appear in the literature (see e.g. \parencite{MW22}), but we were unable to find one that handles the setting of the Sherrington-Kirkpatrick algorithm, where the denoiser function $f_t$ depends on all previous iterates.
Hence, we prove our own polynomial approximation result (see \pref{app:amp-poly}), which gives us the following corollary of \pref{thm:main-intro}:
\begin{corollary}\label{cor:general}
Suppose $\calA$ is a $t$-step AMP algorithm with $\calF$ consisting of functions that are (1) $L$-Lipschitz, (2) have either pseudolipschitz or indicator-function derivatives, and (3) are well-conditioned (in a sense that is made precise in \pref{lem:approx-poly}).

Then for any $\delta =\Omega(1)$, there exists $d = (\frac{L}{\delta})^{O(4^t)}$ so that the degree-$(d,2)$ \LSH approximates $\calA$ (in the sense of \pref{thm:main-intro}) with error $\delta$ on $X \sim \calP_n(\calD,\calK)$.
Further for any $\eps \ge 0$ (allowing $\eps \to 0$ as $n \to \infty$), even when given as input an adversarial $\eps$-principal minor perturbation of $X$, the degree-$(d,d)$ \LSH approximates the output of $\calA$ on $X$ (in the sense of \pref{thm:main-intro}) with error $\le \eps^{1/2}(d \log n)^{O(d)} + \delta$.
\end{corollary}

As an application of \pref{cor:general}, we have robust polynomial-time approximation schemes for non-negative PCA and the Sherrington-Kirkpatrick problem.
\begin{corollary}[robust Sherrington-Kirkpatrick]
\label{cor:robust-sk}
Define the Sherrington-Kirkpatrick problem: 
\[
\argmax\left\{v^\top X v \mid v \in \{\pm \tfrac{1}{\sqrt{n}}\}^n\right\},
 \qquad \text{with } X_{ij} = X_{ji} \sim \calN(0,\tfrac{1}{n}) \text{ independently}.
\]
For any $\delta =\Omega(1)$, $\eps \ge 0$, there exists $d = d(\delta)$ depending only on $\delta$ so that if given an adversarial $\eps$-principal minor corruption $Y$ of $X$, the degree-$(d,d)$ \LSH can be rounded to a vector $v_{\LSH}$ which with probability $1 - o(1)$ satisfies $v_{\LSH}^\top X v_{\LSH} \ge (1-\eps^{1/4} (d\log n)^{O(d)} - \delta)\cdot \ampopt$, where $\ampopt$ is the objective value achieved by AMP as the number of iterations approaches infinity. 
 Modulo a popular conjecture \parencite{Mon21}, $\ampopt\approx 1.52\ldots$, the global optimal value.
\end{corollary}

\begin{corollary}[robust non-negative PCA]
\label{cor:robust-nnpca}
Define the the non-negative PCA optimization problem: 
\[
\argmax\{v^\top X v \mid v \ge 0,\,\, \|v\|=1\},
 \qquad \text{with } X_{ij} = X_{ji} \sim \calN(0,\tfrac{1}{n}) \text{ independently}.
\]
For any $\delta =\Omega(1)$, $\eps \ge 0$ there exists $d = d(\delta)$ depending only on $\delta$ so that when given an adversarial $\eps$-principal minor corruption $Y$ of $X$, the degree-$(d,d)$ \LSH can be rounded to a vector $v_{\LSH}$ which with probability $1-o(1)$ satisfies $v_{\LSH}^\top X v_{\LSH} \ge (1-\eps^{1/4} (d\log n)^{O(d)} -\delta)\cdot\sqrt{2}$, where $\sqrt{2}$ is the global optimum value.
\end{corollary}

These corollaries offer a counterpoint to the lower bounds against the non-\LSH SDPs for non-negative PCA \parencite{BKW22} and optimizing the Sherrington Kirkpatrick Hamiltonian \parencite{MRX20,KB21,GJJPR20}.
In particular, we see that if an SDP incorporates prior information on the solutions of our random optimization problem in the same way that AMP does, it can overcome these lower bounds in polynomial time, even robustly. 

\begin{remark}
\pref{cor:general} requires the denoising functions to be Lipschitz and well-conditioned.
Further, in both \pref{thm:main-intro} and \pref{cor:general} the entries of $X$ must be drawn independently from a subgaussian distribution, and the \LSH degree $d$ has an exponential or doubly-exponential dependence on the number of iterations $t$, so that the result is only meaningful when $t$ is fixed as a function of $n$, and even then it is laughably impractical.
Most of these drawbacks are shared by theoretical analyses of AMP algorithms.
Typically, AMP analyses rely on careful Gaussian approximation and approximate independence of the ``noise'' portion of the iterates, and it is assumed that $t $ does not grow with $n$ and that the denoisers $f$ are Lipschitz and not too poorly behaved, as otherwise the Gaussian approximations become inaccurate.
There are a few notable exceptions where these conditions have been relaxed, see \parencite{MV21,LW22}.
Of course, AMP algorithms are still fast to implement, whereas our algorithms run in time $n^{O(d)}$.
\end{remark}

\subsection{Relationship to prior work}

\paragraph{Relating models of computation in statistical settings.}
Our work is a part of the effort to understand the relative power of different models of computation in statistical settings (e.g. \parencite{HKPRSS17,BBHLS21}).
This is useful from both the algorithms and complexity standpoints: since a leading approach in average-case complexity is to prove lower bounds against restricted models of computation, establishing a hierarchy among models of computation amplifies the usefulness of such lower bounds.

Our work establishes a new result of this form: polynomial-time SDPs can robustly simulate AMP algorithms for a broad class of random optimization problems.
The works \parencite{BKW22} and \parencite{MRX20,KB21,GJJPR20} already mentioned above had suggested that SDPs may be worse than AMP in this context; the work \parencite{BKW20} gave some evidence that this may be because the optimization formulation of SDPs also solves the harder task of certifying an upper bound on the objective value.
Here, we elucidate this phenomenon further and confirm that when not forced to solve the certification problem, SDPs can be made to perform no worse than AMP.

Perhaps closest to our work is the recent \parencite{MW22}, which is a result of this form concerning AMP in the context of spiked matrix models.
In their setting, they observe $X = vv^\top + G$ for $G$ Gaussian noise, and the goal is to estimate $v$. 
They show that for such models, AMP and bounded-degree polynomials are equivalent in power: on the one hand AMP is well-approximated by low-degree polynomials, and on the other hand the AMP polynomials achieve the optimal estimation error among all bounded-degree polynomials.
Their work complements ours, and the techniques and technical challenges are almost completely distinct.
\LSH is thought (but not formally known) to be at least as powerful as low-degree polynomials for estimation; in light of \parencite{MW22}, our results are consistent with this hypothesis (though the results are a bit incomparable, because we do not consider the spiked setting).
Their work does not have algorithmic consequences for robustness.

\paragraph{The Local Statistics Hierarchy and SDPs for random estimation problems.}
The \LSH was proposed by Banks, Raghavendra, and Mohanty in \parencite{BMR21}.
They proposed the general framework described in \pref{def:LSH},
but studied it only in the specific context of community detection in the stochastic block model (SBM). 
In the SBM, $X$ is the adjacency matrix of a random graph of average degree $O(1)$ with a {\em planted} sparse $k$-partition, and the goal is to recover the planted partition.
When the signal-to-noise ratio is small, the global balanced minimum $k$-partition is not especially correlated with the planted partition, and so the natural optimization SDP relaxation for the problem should fail (because the exact integer solution to minimum $k$-partition fails).

In light of this, \cite{BMR21} suggest \LSH to give the SDP access to prior information about the joint distribution over the planted partition and the observed graph.
This approach was inspired in part by \cite{HS17}, in which SDPs given access to the appropriate moment tensors in $X$ were used to recover planted partitions.
\cite{BMR21} were unable to show that \LSH estimates the planted partition, but they do show that degree-$(D,2)$ \LSH for $D$ a large enough constant can distinguish, or hypothesis test, between stochastic block model graphs and \erdos-\renyi graphs, even when the graph is adversarially corrupted (in \pref{sec:robust-model} we discuss their noise model).
Building on their work, \cite{DORS21} show that with a different SDP relaxation (related to the \LSH hierarchy, but not exactly the same) can estimate the planted partition robustly. 
Taking a slightly different approach, \cite{LM22} use a different SDP in combination with a boosting procedure to obtain minimax-optimal robust recovery for the sparse stochastic block model, albeit in a slightly weaker adversarial corruption model.

Our paper is the first to use the \LSH hierarchy to perform robust estimation.
One difference between our use of \LSH and that of \parencite{BMR21} is that rather than using the moments of the joint distribution over matrices $X$ and optimizers $v$ of $v^\top X v$, we are specifically using the joint distribution over $X$ and $v_{\amp}$, the solution returned by AMP on $X$. 
This is {\em crucial} for our success; in fact, it is not clear that the task of sampling from the posterior over near-optimal $v$ conditioned on $X$ is a computationally tractable task. 
In contrast to these prior works, our paper is primarily relevant in the ``dense'' setting when a constant fraction of the entries of $X$ are nonzero, whereas the prior works are concerned with the setting where $X$ is the adjacency matrix of a sparse graph.

\paragraph{Approximate Message Passing.}
As we have already mentioned above, AMP is a popular algorithm that has found a wealth of applications in high-dimensional algorithmic statistics, see e.g. the surveys \parencite{Mon12,FVRR22}.
AMP is known to achieve Bayes-optimal error rates (i.e. minimize the mean squared error) for a number of spiked matrix problems and beyond \parencite{FVRR22}.
Because of the success of AMP in these settings, people often prove lower bounds against AMP as a restricted model of computation in order to better understand information-computation gaps \parencite{CMW20,CM22}.

In this context, it is interesting to understand when AMP is or is not more powerful than semidefinite programs, and whether lower bounds for one model can rule out the success of the other.
SDPs are known to outperform AMP in some contexts; for example, in the tensor version of PCA, SDPs dramatically dominate AMP (even ignoring issues of robustness), succeeding at asymptotically smaller signal-to-noise ratios \parencite{RM14,HSS15}.\footnote{Recently, \cite{WAM19} showed that a spectral algorithm inspired by message-passing algorithms (corresponding to the Kikuchi free energy in graphical models) match the performance of SDPs for tensor PCA, thus ``partially redeeming'' message passing algorithms in this context.} 
On the other hand, there are contexts where AMP algorithms are known to outperform the natural optimization SDPs, such as non-negative PCA and the SK problem \parencite{BKW22,MR15,Mon21,GJJPR20}.
Our work shows that polynomial-time SDPs dominate AMP algorithms when the correct SDP is used (albeit with a much slower running time).

\subsection{Discussion and open problems}

Our results show that a broad class of AMP algorithms can be simulated by \LSH semidefinite programs, even in the presence of adversarial corruptions of principal minors of measure $1/\polylog(n)$.
The corresponding SDP is polynomial-time, but we are only able to guarantee the success of a {\em very} large \LSH relaxation, rendering our algorithms dramatically slower than the corresponding AMP algorithm.
This immediately raises the question: can AMP be robustly simulated by more practical semidefinite programs? 
One reason to hope for an affirmative answer is that this is known to hold in the related context of Belief Propagation in the stochastic block model \parencite{LM22}.

Our paper is the first to use the \LSH in the context of estimation/optimization.
We calibrate the \LSH to the joint distribution $(X,v_{\amp}) \sim \calP$, and the analysis of our algorithm is predicated on the fact that the marginal distribution $\calP\mid_X$ on $v_{\amp}$ given $X$ is effectively supported on a point mass; that is, the solution output by AMP, $v_{\amp}$ is a deterministic function of the input.
This is the setting of most AMP algorithms, but one might hope that \LSH (and maybe also AMP) algorithms are useful in more complex situations, when $\calP|_X$ has more entropy, and perhaps not only for estimation but also for sampling.
For example, in the Sherrington-Kirkpatrick optimization problem there are exponentially many solutions of objective value $\eps$-close to the optimum; is it possible to use \LSH to sample from these?
It would also be interesting to combine the optimization capabilities of SDPs with the \LSH constraints, perhaps surpassing the theoretical guarantees of AMP.

In the non-robust context, our \LSH algorithm is a bit silly. 
Since we are calibrating \LSH to the joint distribution $(X,v_{\amp})$, the linear constraints of the SDP are essentially running AMP, without the benefit of the highly-efficient iterative implementation. 
But even here, we find it quite remarkable that {\em \LSH simulates AMP in a black-box fashion}. 
When we program the local statistics hierarchy SDP, we do not need to know the AMP algorithm.
We don't need to know what the denoising functions are; we just need to know some low-order statistics the joint distribution over inputs and AMP outputs, which could be handed to us by some oracle.
We think this emphasizes the power and ``universality'' of the \LSH as formulated by \cite{BMR21}.
We wonder whether SDP-based approaches which do take the structure of the AMP denoisers more directly into account could perhaps yield robustness with more practical running times.

Lastly, there is a question of the optimality of our algorithms in the robust context (in terms of the approximation error). 
Our guarantees are only meaningful if the measure of the principal minor corrupted is at most $\frac{1}{\polylog n}$.
Specifically,
\pref{cor:general} guarantees that when our observed matrix $X$ has an $\eps$-fraction of corruptions, we can recover a solution which is $(1-\sqrt{\eps}\polylog n)$-correlated with the AMP solution; the exponent of the logarithm depends doubly exponentially on the number of iterations of AMP (this can be improved to an exponential dependence if the denoisers are degree-$k$ polynomials). 
It would be interesting to understand if these rates are tight.
Could a constant fraction of errors be tolerated?
Could a $1/\polylog(n)$ fraction of errors be tolerated by a faster algorithm?

\subsubsection{Notions of robustness for AMP}\label{sec:robust-model}
We verify that AMP is not robust to the adversarial $\eps$-principal minor contamination model unless $\eps \ll \frac{1}{\sqrt{n}}$, as witnessed by the following example.
Our theorems tolerate $\eps = O(1/\polylog(n))$.
\begin{example}[AMP is not robust to principal minor corruptions]
Consider matrix power iteration, which is AMP with the denoiser $f(x) = x$.
Suppose that $X_{ij} \sim \calN(0,\frac{1}{n})$ iid, but instead AMP is given $Y = X + \frac{1}{\sqrt{n}} 1_{S} 1_S^\top$ for $1_S$ the restriction of the all-1 vector to $S \subset [n]$, $|S| = \eps n$.\footnote{The coefficient $\frac{1}{\sqrt{n}}$ chosen so that the typical entry of the perturbation will not be obviously larger than a typical entry of $X$; a larger coefficient would have been technically fine.
}
With high probability $\|X\|_{op} = \Theta(1)$, whereas $Y$'s top eigenvector is close to $1_S$, with eigenvalue $\eps \sqrt{n}$.
Hence, AMP incorrectly converges, eventually, to $\approx 1_S$ as long as $\eps \gg \frac{1}{\sqrt{n}}$.
Further, recall that we start at $v_0 = \vec{1}$, a ``warm start'' with $\Iprod{\frac{1}{\|v_0\|}v_0,\frac{1}{\|{1_S}\|}1_S} \ge \sqrt{\eps}\gg \frac{1}{\sqrt{n}}$; a calculation then shows that even the $t$'th iterate for $t = O(1)$ is equal to $1_S$ up to low-order noise, so long as $\sqrt{\eps}(\eps \sqrt{n})^t \gg 1$.
\end{example}

A priori one might have hoped for an algorithm which robustly simulates AMP even in the $\eps$-strong contamination model: when the corrupted $Y$ satisfies only that the support of $X-Y$ has size at most $\eps n^2$ (not necessarily taking the form of a principal minor).
The following observation shows that this is {\em information theoretically impossible} in the context of AMP.

\begin{observation}[Impossibility of robustness to strong contamination]
Suppose $X$ is a symmetric matrix with entries chosen iid and uniform from $\{\pm 1\}$.
Consider the single-step AMP algorithm which computes $v_\amp = \frac{1}{\sqrt{n}}X\vec{1}$. 
Note that $v_\amp(i)$ is proportional to the $i$th row sum of $X$.

We design an $\eps$-strong contamination $Y$ in which all row sums are zero, so that it is information-theoretically impossible to determine the signs of the row sums of $X$ with accuracy much better than a random guess.
At first we let $Y^{(0)} = X$.
For each $i \in [n]$: (1) copy over $Y^{(i)} = Y^{(i-1)}$, then (2) letting $b_i = (Y^{(i-1)}\vec{1})_i$, sample $|b_i|$ entries with sign $\mathrm{sign}(b_i)$ uniformly at random of the $i$th row of $Y^{(i-1)}$, and set them to zero in $Y^{(i)}$, after which (3) zero out the corresponding entries of column $i$ of $Y^{(i)}$ to maintain symmetry.
Finally, take $Y = Y^{(n)}$.

It is not difficult to show that the total number of entries corrupted in this process is $\sum_{i=1}^n 2|b_i| = O(n^{3/2})$ with high probability, which is $\le \eps n^2$ so long as $\eps = \Omega(n^{-1/2})$.
Furthermore, since the row sums of $X$ are weakly correlated sums of iid signs, it is information-theoretically impossible to infer the signs of each of the row sums of $X$ from $Y$ with accuracy much better than a random guess.\footnote{Formalizing this would go as follows: each of the $|b_i|$ changed entries are almost equally likely to have been positive or negative, even given the signs of $b_1,\ldots,b_{i-1}$, because almost all the $|b_j| = O(\sqrt{n\log n})$ for $j < i$ and because the zeroed entries are chosen at random, with very high probability there will be at most $O(\log n)$ zeroed out entries in row $i$ resulting from previously processed rows $j$.
This is not enough to appreciably influence the sign of the $i$th row sum at step $i$ (compared with its previous sum).
} 
Hence given only access to $Y$, one cannot approximate $v_\amp$ with any appreciable accuracy.
\end{observation}

Given this, we find our corruption model to be quite strong; further, it is consistent with the corruptions considered in work on robust algorithms for community detection in the stochastic block model (SBM).
In the sparse SBM, \cite{BMR21,DORS21} can recover a $\poly(\eps)$-approximate solution when any $O(\eps n) = O(\eps \|X\|_F^2)$ entries are adversarially corrupted---but clearly, these entries must be contained in an $2\eps n \times 2\eps n$ principal minor.
We can also handle $O(\eps \|X\|_F^2)$ corrupted entries, provided they occur in a principal minor of measure $\poly(\eps)$ (though we require $\eps = O(1/\polylog n)$).

\paragraph{Organization}
\pref{sec:overview} is a technical overview.
\pref{sec:prelims} is dedicated to preliminaries, including background on AMP, setup for the sum-of-squares proofs, and definitions of a useful basis of $\calS_n$-symmetric polynomials.
In \pref{sec:robust} we give our robust algorithm and prove \pref{thm:main-intro}.
We give some proofs of concentration in \pref{app:hyp} and of polynomial approximation in \pref{app:amp-poly}.

%% file: overview-arxiv.tex
\section{Proof overview}
\label{sec:overview}

For simplicity of exposition, suppose that $\calA$ is an AMP algorithm as in \pref{def:amp-intro} which applies the same denoiser at each of its $t$ iterations, and that every iteration only depends on the previous iteration, so that there exists an $f:\R \to \R$ so that $\calF = f_1,f_2,\ldots$ satisfies 
\[
f_s(Xv_{s-1}, v_{s-1},v_{s-2},\ldots,v_0) = f(Xv_{s-1}) \quad \text{for all }s \le t.
\]
This simpler case captures the main ideas.
We will also ignore the Onsager correction (this will greatly reduce bookkeeping and slightly simplify the proof).

\paragraph{Non-robust simulation with polynomial denoisers.}
First, consider the non-robust setting, in which we observe $X \sim \calP_n(\calD,\calK)$ and the \AMP algorithm maximizes the objective $v^\top X v$ subject to $v \in \calK$.
Suppose first that the denoiser $f$ is a polynomial of degree at most $k$.

Letting $d = k^{O(t)}$, we assume that we have access to the degree-$(d,2)$ moments of the joint distribution over pairs $(X',v_{t}(X'))$, where $X' \sim \calP_n(\calD,\calK)$ and $v_{t}(X')$ is the $t$th iterate of the AMP algorithm $\calA$ on the input $X'$.
If these moments are not known to us in advance, we can estimate the moments of this distribution up to arbitrary accuracy in polynomial time by sampling a sequence of $X_1,X_2,\ldots \sim \calP_n(\calD,\calK)$ and running the AMP algorithm $\calA$ on each of them. 

We then set up the local statistics hierarchy (\pref{def:LSH}) of degree-$(d,2)$ so that it will return a linear operator $\pE: v^{\le d} \to \R$ satisfying the linear constraints
\[
\pE q(X,v)  = \E_{X'\sim \calP_n} q(X',v_{t}(X')) \pm C \sqrt{\Var_{X'\sim \calP_n} q(X',v_{t}(X'))},
\]
for any $\calS_n$-symmetric polynomial $q(X,v)$ and $C$ a large enough constant.\footnote{In \pref{sec:alg} we add only the subset of these constraints that will be useful in our analysis.}
For the polynomials $q$ we care about, the standard deviation term will be of lower order, so that this is effectively enforcing 
\[
\pE q(X,v) = (1\pm \delta)\E_{X' \sim \calP_n} q(X',v_t(X'))
\]
for $\delta$ an arbitrarily small constant.
We make use of a $d^{O(d^2)}$-sized basis of polynomials to enforce this with only $d^{O(d^2)}$ linear SDP constraints.
Because of the concentration of low-degree polynomials in $X$, as long as $C$ is a large enough constant the SDP will be feasible with high probability (as witnessed by the ``integral solution'' pseudoexpectation given by $v = v_t(X)$).

Since we have assumed the denoiser $f$ is a degree-$k$ polynomial, and since AMP iteratively defines
\[
v_t = f(X v_{t-1}) = f(X f(X v_{t-2})) = \cdots = f(Xf(Xf(X \cdots f(Xv_0) \cdots ))),
\]
we can express $v_t$ as a vector-valued polynomial in $X$, $v_t = v_t(X)$.
It is also the case that $v_t$ is $\calS_n$-symmetric: since $v_0 \propto \vec{1}$, permuting the rows and columns of $X$ and of $v_t$ in a consistent way fixes the vector-valued function $v_t$.
Because $v_t(X)$ is the $t$-fold composition of the function $f(X\cdot)$ applied to $v_0$, the degree of $v_t$ in $X$ is at most $k(\deg(v_{t-1} +1))$. 
Solving this recurrence gives $\deg(v_t) = O(k)^t$.

Because of the symmetry of $v_t$, the polynomial $\iprod{v_t(X),v}^2$ is $\calS_n$ symmetric in $v$ and $X$ of degree at most $(O(k)^t,2)$. 
Hence our \LSH relaxation automatically includes the constraint
\begin{equation}
\pE \left[\Iprod{v_t(X),v}^2\right] 
\ge (1-\delta) \E_{X'}[\iprod{v_t(X'),v_t(X')}^2]
 = (1-\delta)\E_{X'}[\|v_t(X')\|^4]
\ge (1-\delta)\E_{X'}[\|v_t(X')\|^2]^2
\label{eq:main}
\end{equation}
where $\delta$ comes from the slack in our \LSH constraints.
This implies (by re-writing the left-hand side)
\begin{equation}
v_t(X)^\top\left(\pE[vv^\top]\right) v_t(X) 
\ge (1-\delta)\E_{X'}[\|v_t(X')\|^2]^2
\end{equation}
Further, we also have the \LSH constraint
\[
\tr(\pE[vv^\top]) = \pE[\|v\|^2] \le (1+\delta) \E_{X'}[\|v_t(X')\|^2],
\] 
and by the concentration of sums of low-degree polynomials in the subgaussian distribution $\calD$, with high probability $\|v_t(X)\|^2_2 = (1\pm o(1)) \E_{X'} \|v_t(X')\|^2$.
Putting these together,
\[
\tfrac{v_t(X)}{\|v_t(X)\|}^\top (\pE[vv^\top]) \tfrac{v_t(X)}{\|v_t(X)\|} \ge (1-o(1))\tfrac{1-\delta}{1+\delta} \tr(\pE[vv^\top]).
\]
Hence \pref{eq:main} and the positive-semidefiniteness of $\pE[vv^\top]$ implies that $\pE[vv^\top]$ is well-enough correlated with a rank-$1$ matrix that the top eigenvector $v^*$ of $\pE[vv^\top]$ must be proportional to $v_t(X)$, by the following easy-to-prove claim (see the proof of \pref{thm:ramp-recovery-1}):
\begin{claim}
If $A$ is a positive-semidefinite matrix and $u$ is a unit vector satisfying $u^\top A u \ge (1-\eta)\tr(A)$, the top eigenvector $v$ of $A$ satisfies $\iprod{u,v}^2 \ge 1-2\eta$.
\end{claim}
Taking $A = \pE[vv^\top]$ and $\eta = O(\delta)$ we may conclude that $\Iprod{v^*,\frac{v_t(X)}{\|v_t(X)\|}}^2 \ge 1-O(\delta)$, so eigenvector rounding will yield a solution which approximates the AMP solution.
This proves that when $f$ is a polynomial of degree $k$, degree-$(O(k)^t,2)$ \LSH can simulate AMP.

\paragraph{Robust simulation with polynomial denoisers.}
Suppose that instead of being given a clean sample, our input matrix $Y$ is a corrupted observation of a matrix $X \sim \calD$, where our only guarantee is that $X-Y$ is supported on a principal minor of dimension $\le \eps n$.
We will modify the \LSH in a manner inspired by prior works in the \sos robust-statistics literature, by adding SDP variables $\Xprox_{ij}$ and $W_{i}$ for each $i,j \in [n]$, where $\Xprox_{ij}$ is a proxy variable which represents our best guess for $X_{ij}$, and $W_{i}$ is a variable that represents the indicator that $\Xprox_{i,\cdot} = Y_{i,\cdot}$, i.e. that row $i$ is uncorrupted.
To capture this, we add the polynomial constraints $W_i^2 = W_i$ (Booleanity), $W_{i}(\Xprox_{ij} - Y_{ij}) = 0$ (encoding $\Xprox_{ij} = Y_{ij}$ if row $i$ is uncorrupted) and $\sum_{i} W_{i} = (1-\eps) n$ (the total fraction of uncorrupted rows is $\ge (1-\eps)$).
The $Y_{ij}$ variables are now our ``observed'' variables, and the $v,W,\Xprox$ are our ``program'' variables.

We also add, in addition to the typical local statistics constraints, the operator norm constraint
\[
\Xprox^2 \psdle (1+\gamma)\E_{X' \sim \calP_n}[\|X'\|_{\op}]^2 \cdot \Id,
\]
and some ``fixed-coordinate'' local statistics constraints 
\[
|q(\Xprox)| \le C_n \cdot \sqrt{\E_{X' \sim \calP_n} q(X')^2}, \quad \forall q \in \calQ_i,\quad \forall i \in [n],
\] 
where $\calQ_i$ is the set of all polynomials in $X$ which are fixed by the action of $\calS_{n-1}$ on coordinates $[n]\setminus \{i\}$.
Both of these constraints can be encoded in the \sos SDP; again, using a convenient basis, the latter only requires $n \cdot d^{O(d^2)}$ linear constraints.\footnote{In \pref{sec:alg} we will only the subset of these constraints that we make use of in our proof.}
We choose $\gamma$ and $C_n$ large enough so that this is satisfied by $\Xprox = X$ with high probability over $X$.
We can choose $\gamma$ to be a very small constant because the operator norm of $X$ concentrates very well.
As the fixed-coordinate constraints have to hold simultaneously for all $i \in [n]$, it is necessary and sufficient to take $C_n = \Theta((\log n)^{\deg(q)/2})$.

The AMP solution $v_t$ is a polynomial of degree at most $d = O(k)^t$, but it also has special structure: since it results from the iterated application of a polynomial function, $v_t$ belongs to a special class of vector-valued polynomials which we call ``forest polynomials'' (because the computation graph for $v_t$ looks like a weighted sum of trees). 
Letting $D =O(d)$ (now $D$ represents our degree in both $v$ and the program variables $\Xprox,W$), we prove in the degree-$(d,D)$ sum-of-squares proof system that any vector-valued forest polynomial $u(X)$ of degree at most $d$ satisfies
\begin{equation}\label{eq:close}
\pE[\|u(X) -u(\Xprox)\|^2] \le \alpha\E[\|u(X)\|^2], \quad\text{ for }\alpha = O(\sqrt{\eps}) (d\log n)^{O(d)}.
\end{equation}
We'll explain how \pref{eq:close} is proven below.
Applying \pref{eq:close} with $u=v_t$ shows that $\E_X\pE[\|v_t(X)- v_t(\Xprox)\|^2 ] \le \alpha \pE[\|v_t(X)\|^2] \approx \alpha\pE[\|v\|^2]$.
Now the argument is easily finished by appealing to the reasoning in the previous section, which tells us that $\pE[\|v - v_t(\Xprox)\|^2] \ll \pE[\|v\|^2]$.
Combining these using the (\sos) triangle inequality and concentration of $\|v_t(X)\|$ gives that 
\[
\pE[\|v - v_t(X)\|^2] \le O(\alpha)\pE[\|v\|^2] \implies (1-O(\alpha))\pE[\|v\|^2]\le \pE[\iprod{v,v_t(X)}].
\]
Since the non-negativity of the variance has a sum-of-squares proof,
$\pE[\iprod{v,v_t(X)}^2] \ge (1-O(\alpha)) \pE[\|v\|^2]^2$, and we can follow the same chain of reasoning as previously to conclude that the top eigenvector of $\pE[vv^\top]$ is $1-O(\alpha)$ correlated with the output of AMP on the uncorrupted matrix, $v_t(X)$.

\medskip
We now sketch the ideas in the proof of \pref{eq:close}:
the AMP solution $v_t(X)$ is formed by iterative application of (i) multiplication by $X$ and (ii) entry-wise applications of a polynomial denoiser (potentially a different denoiser at each step, $f_s$ at step $s$), starting from $v_0 = \vec{1}$.
The proof of \pref{eq:close} more-or-less reduces to the case where each denoiser $f_s$ is homogeneous, $\{f_s(x) = x^{\ell_s}\}_{s=1}^t$, as $v_t$ is roughly a linear combination of such terms.\footnote{To express $v_t$ we have to allow trees with more irregular degree structure, in which nodes at the same depth may have different degrees. The proof for such trees is the same.}
We call each such term a ``tree'' and a linear combination of such terms a ``forest.''\footnote{The Onsager correction term requires a bit of additional work.}

The proof then proceeds by induction on the polynomial degree of the vector: the base case is the degree-$0$ tree polynomial, $u(X) = \vec{1}$, which clearly satisfies \pref{eq:close} since it does not depend on its input. 

The induction step will have two cases.
In the first case, suppose we have $u(X) = X w(X)$ for $w(X)$ any degree-$(d-1)$ tree polynomial.
This is the ``heart'' of the argument, as this is the step at which we utilize the robustness constraints.
We can write
\begin{align}
\|u(X) - u(\Xprox)\|_2^2 
&= \|X w(X) - \Xprox w(\Xprox)\|_2^2\nonumber \\
&= \|\tfrac{1}{2}(X - \Xprox)(w(X) + w(\Xprox)) + \tfrac{1}{2}(X+\Xprox)(w(X)- w(\Xprox))\|_2^2\nonumber\\
&\le \|(X - \Xprox)(w(X) + w(\Xprox))\|^2 + \|(X+\Xprox)(w(X)- w(\Xprox))\|_2^2\label{eq:two-term}
\end{align}
To bound the first term, define $U_i = \Ind[\text{row/col }i\text{ uncorrupted in }Y]$.
Define the \sos variable $V_i = 1-W_iU_i$ for $W_i$ the \sos robustness variable defined above; $V_i$ represents the indicator that either $Y_{i,\cdot} \neq X_{i,\cdot}$ or $\Xprox_{i,\cdot} \neq Y_{i,\cdot}$.
Finally, let $D_V = \diag(V)$.
Our robustness constraints (and the symmetry of $X,\Xprox$) imply that 
\[
(X_{ij} - \Xprox_{ij})V_j = (X_{ij} - \Xprox_{ij}) \implies (X-\Xprox)D_V = X-\Xprox
\]
as can be proven by a short, dull \sos proof which expresses in algebra the idea that $X_{ij} = \Xprox_{ij}$, unless either $X_{\cdot,j} \neq Y_{\cdot,j}$ or $Y_{\cdot,j} \neq \Xprox_{\cdot,j}$.
Hence, by the ``triangle inequality'' $(A+B)^2 \le 2A^2 + 2B^2$,
\[
\|(X-\Xprox)(w(X) + w(\Xprox)\|^2 
= \|(X-\Xprox)D_V(w(X) + w(\Xprox)\|^2 
\le 2\left(\|X\|_{\op}^2 + \Xprox\|_{\op}^2\right)\|D_V(w(X) + w(\Xprox))\|^2.
\]
Since $\|X\|_{\op} = O(1)$ with high probability, and we have enforced an upper bound on the operator norm of $\Xprox$ as well, the operator norms are $O(1)$.
We now make use of the fact that $D_V$ is $O(\eps)$-sparse, and the vector $w(X) + w(\hat X)$ is delocalized:
\[
\|D_V(w(X) + w(\Xprox))\|^2 = \sum_i V_i^2 (w(X) + w(\Xprox))_i^2 \le \sqrt{\|V\|_4^4 \cdot \|w(X) + w(\Xprox)\|_4^4} \le \sqrt{\|V\|_4^4\cdot 4\left(\|w(X)\|_4^4 +\|w(\Xprox)\|_4^4\right)}
\]
where the inequalities are Cauchy-Schwarz and $(A+B)^4 \le 4(A^4 + B^4)$.
The Booleanity constraints imply that $V_i^4 = V_i$, and the fact that the fraction of corrupted rows is at most $\eps$ imply that
\[
\|V\|_4^4 = \sum_i V_i = \sum_{i} 1-W_iU_i  \le \sum_i (2- W_i - U_i) \le 2\eps n.
\]
At the same time, $w(X)$ is a $\calS_n$-symmetric, degree-$(d-1)$ vector-valued polynomial in iid subgaussian random variables.
Thus $w$ is delocalized with high probability, in the sense that most entries are on the same order of magnitude and $\|w(X)\|_4^4 =O\left( \frac{1}{n}\right) \cdot  \|w(X)\|_2^4$.
The local statistics hierarchy constraints then enforce that $\|w(\Xprox)\|_4^4 \approx \|w(X)\|_4^4 = O(\frac{1}{n}) \cdot \|w(X)\|_2^4$.
Putting these conclusions together, we have that the first term of \pref{eq:two-term} can be bounded with high probability by
\[
\|(X-\Xprox)(w(X)+w(\Xprox)\|^2 \le \sqrt{O(\eps n) \cdot O\left(\tfrac{1}{n}\right) \cdot \|w(X)\|_2^4} = O(\sqrt{\eps}) \|w(X)\|_2^2.
\]
To handle the second term of \pref{eq:two-term}, we apply the induction hypothesis to $\|w(X) - w(\Xprox)\|^2$ (using as before that the operator norms of $X$ and $\Xprox$ are bounded). 
This case in the induction is then complete by noting that $\|w(X)\|^2$ and $\|u(X)\|^2$ are of the same order (via concentration of low-degree polynomials).

It remains to handle tree polynomials produced by the powering step of AMP; taking the $\ell$th power of a vector $w(X)$ entrywise produces a vector of the form  $w(X)^{\circ \ell} = w(X)\circ w(X)^{\circ \ell-1}$ (here $\circ$ is the ``entry-wise'' or ``Hadamard'' product).
So in the second case of induction, we let $u(X)$ be a tree polynomial of the form $u(X) = w_1(X) \circ w_2(X)$, for $w_1$ and $w_2$ tree polynomials of degree $d_1,d_2 > 0$ satisfying $d_1 + d_2 = d$.
Here, we will have to make use of the coordinate-wise local statistics constraints, and incur logarithmic factors.
Using the decomposition $2(a\circ a'-b\circ b') = (a+b)\circ(a'-b') + (a-b) \circ (a'+b')$,
\begin{align*}
\|u(X) - u(\Xprox)\|^2 
&= \|\tfrac{1}{2}(w_1(X) - w_1(\Xprox))\circ(w_2(X) + w_2(\Xprox)) + \tfrac{1}{2}(w_1(X) + w_1(\Xprox))\circ(w_2(X) - w_2(\Xprox))\|^2\\
&\le \|(w_1(X) - w_1(\Xprox))\circ(w_2(X) + w_2(\Xprox))\|^2 + \|(w_1(X) + w_1(\Xprox))\circ(w_2(X) - w_2(\Xprox))\|^2.
\end{align*}
The argument for bounding these two terms is identical, so we explain the argument for just the latter.
We apply the bound,
\begin{align*}
\|(w_1(X) + w_1(\Xprox))\circ(w_2(X) - w_2(\Xprox))\|^2
&= \sum_i (w_1(X)_i + w_1(\Xprox)_i)^2 (w_2(X)_i - w_2(\Xprox)_i)^2\\
&\le \left(\max_{i \in [n]}(w_1(X)_i + w_1(\Xprox)_i)^2 \right) \cdot \|w_2(X) - w_2(\Xprox)\|_2^2\\
&\le\left(\max_{i \in [n]}(w_1(X)_i + w_1(\Xprox)_i)^2 \right) \cdot \sqrt{\eps} (d_2\log n)^{O(d_2)} \|w_2(X)\|^2,
\end{align*}
where in the final line we have applied the induction hypothesis to the tree polynomial $w_2$.
Since we have included the coordinate-wise local statistics constraints and because of concentration and symmetry, we have that $\left(\max_{i \in [n]}(w_1(X) + w_1(\Xprox))^2_i \right) \le (d_1 \log n)^{O(d_1)}\frac{1}{n}\E[\|w_1(X)\|^2]$.
Since $w_1$ and $w_2$ are low-degree polynomials in iid samples from $\calD$, they concentrate such that $\frac{1}{n}\E[\|w_1(X)\|^2]\|w_2(X)\|^2 = \Theta(1)\cdot \|u(X)\|^2$ with high probability.
Combining these,
\begin{align*}
\|(w_1(X) + w_1(\Xprox))\circ(w_2(X) - w_2(\Xprox))\|^2 
&\le \left((d_1 \log n)^{O(d_1)}\cdot \frac{1}{n}\E[\|w_1(X)\|^2]\right) \cdot \left(\sqrt{\eps} (d_2 \log n)^{O(d_2)} \cdot \|u(X)\|^2\right)\\
& \le \sqrt{\eps} (d\log n)^{O(d)} \|u(X)\|^2.
\end{align*}
Repeating the argument for the other term, we see that the inductive hypothesis holds, completing the proof of \pref{eq:close}.
The proof then easily carries over to the degree-$(O(d),O(d))$ \sos proof system.

It is not clear to us if our for the second case of the inductive argument is tight.
If we had very fine-grained control on the tail of the ``empirical distribution'' of entries of $w_1(\Xprox)$ and $w_2(\Xprox)$, we could potentially avoid paying the logarithmic factors associated with the bound on the maximum coordinate.
We do not know how to obtain such tight control in constant-degree \sos.

\paragraph{Lipschitz, well-conditioned denoisers.}
If the denoisers $\calF$ are not polynomials but instead are $L$-Lipschitz, well-conditioned functions with reasonable derivatives, we can show that each $v_t(X)$ can be approximated up to error $\delta$ by a polynomial $h_t(X)$ of degree $d\le (\frac{L}{\delta})^{O(4^t)}$; we do this using a combination of standard techniques in polynomial approximation theory, and results from the theory of AMP.
The reasoning above then applies as before, except that we now require much larger degree, and we incur an approximation error $\delta$ as well as the robustness error of order $\sqrt{\eps} (d\log n)^{O(d)}$.

%% file: prelims-arxiv.tex
\section{Preliminaries}\label{sec:prelims}
\subsection{Notation}
We will use $\circ$ to denote the entry-wise (or Hadamard) product: for $a,b$ of the same dimension, $(a \circ b)_i = a_i b_i$. 
We will also use the entry-wise $k$th power notation $(a^{\circ k})_i = a_i^k$, and entry-wise products over a set $(\bigcirc_{a \in A} a)_i = \prod_{a \in A} a_i$.

We use standard big-$O$ notation.
We will sometimes write $O_k(1)$ to denote a term which is constant as $k\to\infty$ (making similar use of $o_k(1)$, etc.).
If we write $O(1)$, etc. with no subscript, it is understood to mean $O_n(1)$ for $n$ the dimension of the matrix input to the algorithm.

\begin{definition}\label{def:subg}
We say that a scalar random variable $Z$ is $\sigma$-subgaussian if $\E[|Z|^k] \le \sigma^{k}\cdot k^{k/2}$ for each integer $k \ge 1$.
\end{definition}

\subsection{Approximate Message Passing (AMP)}

AMP is a family of iterative alternating-projection algorithms, generalizing matrix power iteration.

\begin{definition}[AMP algorithm]
    \label{def:amp-instance}
    An {\em AMP algorithm} is defined by a collection of functions $\calF = \{f^t:t\in \mathbb N\cup \{-1\}\}$ called the {\em denoiser functions}, with each $f^t: \R^{t + 1}\rightarrow \R$. 
These may naturally be extended to separable functions $f^t: \R^{n\times t}\rightarrow \R^n$. 
We let $f^{-1}(\cdot) = 0$.

        The AMP algorithm on input $X$ produces the ordered list of iterates $\calI(X) = \{x^t:t\in \mathbb N\cup \{-1\}\}$ where $x^{-1} = \vec 0$, $x^0 = \vec{1}$, and each subsequent iterate is defined by the recursion \[x^{t+1} = Xf^t(x^t,x^{t-1},\ldots,x^0) - \sum_{j=1}^{t}b_{t,j}f^{j-1}(x^{j-1},\ldots,x^0).\]
        Here, $b_{t,j}$ is defined as 
        \[b_{t,j} = \frac 1n\sum_{i=1}^{n}\left.\frac{\partial f^t(x_i^{t},\ldots,u_i^{j},\ldots,,x_i^{0})}{\partial u_i^j}\right|_{u^j\rightarrow x^j},\]
        that is, the normalized divergence of $f^t$ with respect to $x^j$. This latter term is commonly called the \emph{Onsager correction}.
    
        An AMP algorithm may employ, as a post-processing step, a Lipschitz ``rounding'' function which maps iterates $x^t$ to a set $K$ (e.g. the set of unit vectors), $g_n: \R^n\rightarrow K$.
\end{definition}

In the above definition, the main takeaway is the application $x^{t+1} = Xf^t(x^t)$ (the $b_{t,j}$ is a correction term to give a precise high-dimensional characterization of the solution). 

\begin{example}[Non-negative PCA]\label{ex:nnpca}
For the \emph{nonnegative PCA} (nnPCA) problem, 
\[
\argmax\{v^\top X v \mid x \ge 0, \|x\|=1\},
\] 
A natural choice of denoiser is $f^t(x) = x_+ \defeq \max(x,0)$, and a natural choice of rounding function is $g(x) = \frac{x_+}{||x_+||_2}$. 
In this case, $b_{t,t} = ||x_+||_0$ and $b_{t,j} = 0$ for $j < t$.
It is shown in \cite{MR15} that when this AMP algorithm is applied to $X \sim \text{GOE}(n)$, with probability $1$ it finds a vector $\hat{v}$ of the optimal value for nnPCA, satisfying  $\lim\limits_{t\rightarrow\infty}\lim\limits_{n\rightarrow\infty}\hat{v}^\top X \hat{v} = \sqrt 2$ (see Theorem 2 of~\parencite{MR15}).
\end{example}

\subsection{Optimization in AMP}

With the AMP algorithm definition and nnPCA example in mind, we define what it means for AMP to converge to an optimal solution.

\begin{definition}[AMP-amenable random quadratic optimization problem]
\label{def:amp-amenable}
A {\em random quadratic optimization problem} $\calP_n(\calD,\calK)$ as defined in \pref{def:rqo} is called {\em AMP-amenable} if when $X \sim \calP_n(\calD,\calK)$, then with probability $1-o_n(1)$,
\[
\lim_{t\rightarrow\infty} \,\, (\hat{x}^{t})^\top X\,\hat{x}^{t} 
=(1\pm o_n(1))\cdot \E_{X'\sim \calP_n}\left[\max_{v\in \calK_n} v^\top X' v
\right],
\]
and $\hat{x}^{t}\in \calK$ for all $n$.	
\end{definition}

Using our example of nnPCA from above, here we would take $\calK$ the all-positive orthant of $\bbS^{n - 1}$, and $\calD = \mathrm{GOE}(n)$ (suppressing the dependence on $n$ in the notation for $\calD$ and $\calK$ for the sake of brevity). 
As mentioned above, $\lim_{n\rightarrow\infty}\E[\max_{v\in \calK}\, v^\top X v] = \sqrt{2}$ and the AMP algorithm from \pref{ex:nnpca} achieves this objective value.

An important tool in the analysis of AMP (which will also be useful for us here) is {\em state evolution}: as $n \to \infty$, any nice coordinate-wise function $\psi$ behaves similarly whether it is applied to the AMP iterates or applied to a specific scalar Gaussian process:

\begin{theorem}[State Evolution]
\label{thm:state-evolution}
	Consider the AMP algorithm defined by $\calF$. 
	Suppose further that each $f^t$ is either Lipschitz or polynomial. 
	Then, for any pseudo-Lipschitz function $\psi: \R^{t+1}\rightarrow \R$,
	\[\operatorname*{p-lim}\limits_{n\rightarrow\infty}\frac 1n\sum_{i=1}^{n}\psi(x^0,x^1,\ldots,x^t) = \E_{U}[\psi(U^0,U^1,\ldots,U^t)]\]
	where $U$ is a centered Gaussian process with covariance matrix $Q$ defined as, for $i \ge j$,
	\[Q_{ij} = \E_{U^0,U^1,\ldots,U^{i-1}}[f^{i-1}(U^0,\ldots,U^{i-1})f^{j-1}(U^0,U^1,\ldots,U^{j-1})].\]
\end{theorem}
The notation $\operatorname{p-lim}\limits_{n\rightarrow\infty}$ means that this statement holds with probability 1 as $n\rightarrow\infty$.
For a proof of the above, we refer the reader to the AMP literature; see for example~\cite[Proposition 2.1]{Mon21}.

\subsection{Robust Optimization}
Our algorithms are robust to adversarial perturbations to a principal minor of bounded size; see also the discussion in \pref{sec:robust-model}.

\begin{definition}[Adversarial $\eps$-principal minor corruption]
    \label{def:corruption-model}
    Let $X\in \R^{n\times n}$ be a symmetric matrix. 
A symmetric matrix $Y \in \R^{n \times n}$ is said to be a {\em $\varepsilon$-principal-minor corruption of $X$} if it differs from $X$ only on an $\eps n \times \eps n$ principal minor.
\end{definition}

\subsection{Tree Polynomials}

Throughout, we will be interested in the vector-valued polynomial $v$ which results from applying the AMP algorithm for $t$ steps as a function of the input $X$.  Here we develop some terminology which will be helpful in analyzing these special polynomials.

\begin{definition}[Interaction graph]
    \label{def:interaction-graph}
    An \emph{interaction graph} is a tuple $G = (V, E, m)$ where $V = \{v_1,\ldots,v_\ell\}$, $E\subseteq V\times V$ is the set of undirected edges (with possible self loops), and multiplicities $m = \{m_e: e\in E\}$ with each $m_e\in \N^+$.
    The total multiplicity of this graph is defined as $M(G) = \sum_{e\in E}m_e$.
\end{definition}

The interaction graph representation is the standard way to set up the \LSH.
The polynomials we are interested in, e.g. $v^\top X v$, are $\calS_n$-symmetric (see~\pref{def:LSH}) and thus may be represented as linear combinations of interaction graphs on vertex set $[n]$. 
However, the structure of AMP allows us to simplify these polynomials one step further to be combinations of tree-like structures; we sacrifice the generality of our SDP to work with these simpler tree-polynomial constraints.

\begin{definition}[Rooted Trees and Tree Polynomials]
A {\em rooted tree} and the corresponding vector-valued polynomial in $\R^n$ are defined recursively according to the following operations:
\begin{itemize}
    \item The ``empty'' tree is a single node, which has corresponding vector valued polynomial $\vec{1}$. 

    \item (Rerooting) Given a rooted tree $T'$, one may produce a new tree $T$ by re-rooting $T'$, extending an edge out of the root of $T'$ to a new root node.
Algebraically, this corresponds to creating the vector-valued polynomial $T(X)$ via matrix multiplication, $T(X) = X T'(X)$.

\tikzset{
itria/.style={
  draw,shape border uses incircle,
  isosceles triangle,shape border rotate=90,yshift=-1.3cm},
rtria/.style={
  draw,shape border uses incircle,
  isosceles triangle,isosceles triangle apex angle=90,
  shape border rotate=-45,yshift=0.3cm,xshift=0.5cm},
ltria/.style={
  draw,shape border uses incircle,
  isosceles triangle,isosceles triangle apex angle=90,
  shape border rotate=-135,yshift=-1.3cm,xshift=0.5cm}
}
\begin{center}
\begin{tikzpicture}[sibling distance=2cm]
\node[circle, fill = black] {}
    child { node[circle,draw, fill=black] {} 
             { node[itria] {$T'$} }                         
        };
\end{tikzpicture}
\end{center}

\item (Grafting) Given two {\bf non-empty} trees $T_1$ and $T_2$, one may produce a new tree $T$ by grafting the root of $T_1$ and $T_2$ together, so that they branch out from the same root. 
Algebraically, this corresponds to creating the vector-valued polynomial $T(X)$ by entry-wise multiplication: $T(X) = T_1(X) \circ T_2(X)$.
    
\tikzset{
itria/.style={
  draw,shape border uses incircle,
  isosceles triangle,shape border rotate=90,yshift=-1.3cm},
ritria/.style={
  draw,shape border uses incircle,
  isosceles triangle,isosceles triangle apex angle=120,
  shape border rotate=-45,yshift=0.0cm,xshift=0.15cm},
litria/.style={
  draw,shape border uses incircle,
  isosceles triangle,isosceles triangle apex angle=120,
  shape border rotate=-135,yshift=0.0cm,xshift=-0.1cm}
}
\begin{center}
\begin{tikzpicture}[sibling distance=2cm]
\node[circle, fill = black] {}
    child{ node[ritria] {$T_1$} }
    child{ node[litria] {$T_2$} };
\end{tikzpicture}
\end{center}
\end{itemize}
\end{definition}

In addition to trees, we must consider what happens upon flattening a tree: as we shall see shortly, this is exactly the form that the Onsager correction term takes.

\begin{definition}[Trunk and Lumber]
	Given a \emph{non-empty} tree $T$\footnote{the tree should be nonempty, otherwise the associated trunk is just $1$.} and associated vector-valued polynomial $T(X)$, define the \emph{trunk} of $T$ to be the {\em scalar} polynomial $\trunk_T(X) = \frac 1n \langle T(X), \vec{1}\rangle$.\footnote{The name comes from noticing that this essentially unroots the tree, thereby turning it into a fallen trunk.} 

We will say that the vector-valued $T(X)$ is {\em lumber} if it is the product of a vector-valued tree $T_1(X)$ and a finite collection of trunks $\trunk_{T_2}(X),\ldots,\trunk_{T_\ell}(X)$, $T(X) = \left(\prod_{i=2}^\ell\trunk_{T_i}(X)\right)\cdot T_1(X)$.
Because of the commutativity of scalar multiplication, the re-rooting and grafting operations act on lumber the same way that they do on trees.
\end{definition}

Lumber has a natural ``inductive'' structure: either a lumber is a simple tree, or it has an accompanying collection of trunks.

\begin{definition}[Forest]
	A vector-valued polynomial $P(X)$ is called a \emph{forest} if it is a weighted sum of lumber: that is, if there exists a collection of lumber $\calL$ such that
	\[P(X) = \sum_{T\in \calL} c_T \cdot T(X)\]
	for some constants $c_T$ which do not depend on $X$. 
We also require that $|c_T| = O_n(1)$ (where $n$ is the dimension of the associated vector) for all $T\in \calL$.
	
	A forest's \emph{degree} is its degree as a polynomial in $X$.
\end{definition}

\begin{fact}
\label{fact:weighted-forest-linear}
	If $P(X)$ and $Q(X)$ are weighted forests of degrees $d_1$ and $d_2$ respectively, then so are $P(X) + Q(X)$ (of degree $\max \{d_1, d_2\}$), $P(X) \circ Q(X)$ (of degree $d_1\cdot d_2$), and $XP(X)$ (of degree $1 + d_1$).
\end{fact}
The proofs of these facts are immediate from the definitions.

\begin{example}
Consider the AMP algorithm corresponding to denoiser functions \[f^t(x^t,x^{t-1},\ldots,x^0) = (x^t)^{\circ 2}.\] Then, we may compute \[b_{t,t} = \frac 2n \sum_{i=1}^{n} (x^t)_i = 2\cdot \frac 1n \langle x^t, \vec 1\rangle\]
and $b_{t,j} = 0$ for $j < t$.
Using this information, let us compute the first few iterates and their corresponding forests.

\begin{enumerate}
	\item Immediately, $x^1 = X\vec 1$: the corresponding forest just consists of the singular once-rerooted tree $T_1$, with corresponding weight $c_{T_1} = 1$.
	\item More interesting is $x^2 = X f^1(x^1) - b_{1,1} f^0(x^0) = X (X1 \circ X1) - 2\trunk_{T_1} \cdot \vec 1$: the first term is still just a tree by applying grafting and rerooting while the second term is a lumber and not just a tree. Let us denote the first tree here by $T_2$.
	\item $x^3$ introduces correlations between terms: in particular, $f^2(x^2) = (T_2(X) - 2\trunk_{T_1}\vec 1)^{\circ 2}$. By~\pref{fact:weighted-forest-linear}, this still corresponds to a weighted forest and so do all further operations.
\end{enumerate}

\end{example}

\subsection{Basic Sum-of-Squares Proofs}

Throughout, we will use several well-known SoS facts (see e.g. \cite{MSS16}): we tabulate these here.

\begin{fact}[SoS Almost-Triangle Inequality]
\label{fact:sos-almost-triangle}
	Let $f_1, f_2, \ldots, f_r$ be indeterminates and $t\in \N$. Then,
	\[\sststile{2t}{f_1,f_2,\ldots,f_r}\left\{\left(\sum_{i=1}^{r}f_i\right)^{2t} \le r^{2t-1}\left(\sum_{i=1}^{r}f_i^{2t}\right)\right\}.\]
\end{fact}

\begin{fact}[SoS Cauchy-Schwarz Inequality, ]
\label{fact:sos-cauchy-schwarz}
	Let $f_1,f_2,\ldots,f_r$ and $g_1,g_2,\ldots,g_r$ be indeterminates. Then,
	\[\sststile{2}{f_i, g_i}\left\{\left(\sum_{i=1}^{r}f_ig_i\right)^2 \le \left(\sum_{i=1}^{r}f_i^2\right)\left(\sum_{i=1}^{r}g_i^2\right)\right\}.\]
\end{fact}

%% file: robust-arxiv.tex
\section{Robust AMP Recovery with Local Statistics}\label{sec:robust}

In this section, we prove our main theorem: the \LSH can simulate AMP, even in the strong contamination model.

\begin{theorem}[Robust simulation of AMP]
	\label{thm:ramp-recovery-1}
	Suppose that $\calD$ is a $\frac{K}{\sqrt n}$-subgaussian distribution with $K = O(1)$, giving rise to the quadratic optimization problem $\calP_n(\calD,\calK)$.
	Let $\calF$ be an AMP algorithm consisting of polynomial denoiser functions of degree at most $k$.
	Let $\vAMP(X)$ denote the output of the $t$-step AMP algorithm on input $X$, normalized so that $\frac 1n \E[\|\vAMP\|_2^2] = 1$.

	Then there exists an integer $d = O((2k)^{2t})$ such that for any $\eta = \omega(1/\sqrt{n})$, $\eps \ge 0$, a robust version of the \LSH hierarchy of degree-$(d,d)$ can, given as input an $\eps$-principal minor corruption $Y$ of $X\sim \calP_n(\calD,\calK)$, be rounded to a vector $v_{\LSH}(Y)$ which satisfies
	\[\frac{\langle v_{\LSH}(Y), \vAMP(X)\rangle^2}{\left\|\vAMP(X)\right\|^2_2 \cdot \left\|v_{\LSH}(Y)\right\|_2^2}\ge 1 - \sqrt{\eps} \cdot O(\log n)^{(2k)^{t}} - \eta\]
	with probability $1-o_{n}(1)$.
	This algorithm runs in time $n^{O(d)}$.
\end{theorem}
\begin{remark}
	Note that the case $\eps = 0$ is covered by the theorem above; in fact, a subset of the SDP constraints and variables we use will suffice if we observe $Y = X$.
\end{remark}

\subsection{Description of the algorithm}
\label{sec:alg}

We are in the setting where there is some latent symmetric $X\in \R^{n\times n}\sim \calD_n$ and we have access to symmetric $Y\in \R^{n\times n}$ which is an adversarial $\varepsilon$-principal minor corruption of $X$.
Our algorithm will use a robust version of the degree-$(d,d)$ Local Statistics Hierarchy (\LSH).
The variables of our program will be $\{\tX_{ij}\}_{i\le j \in [n]}$ which are our proxies for $\{X_{ij}\}_{i\le j \in [n]}$, variables $\{W_{j}\}_{j \in [n]}$ where $W_{j}$ is our proxy for the indicator that the $j$th column is not corrupted, $\Ind[X_{ij} = Y_{ij}\, \forall i \in [n]]$, slack variables $\{B_{ij}\}_{i,j \in [n]}$, and variables $\{v_i\}_{i \in [n]}$ which are proxies for the entries of the vector optimizer.

We'll take our polynomial constraint set $\calC = \calC_{\rob} \cup \calC_{\LSH}$.
The first set of constraints, $\calC_{\rob}$, is designed to identify the ``clean'' rows of $Y$:
\begin{equation}
	\label{eqn:intersection-constraints}
	\calC_{\rob} = \left\{\begin{matrix}
		\forall j\in [n], & W_{j}^2 = W_{j} \quad \text{ and } \quad
		W_j\tX_{i,j} = W_jY_{i,j}                                            \\
		\vspace{0.1cm}
		                  & \sum\limits_{i}(1 - W_{i}) = \varepsilon \cdot n \\
		\vspace{0.2cm}
		                  & \tX^\top = \tX                                   \\
	\end{matrix}\right\}
\end{equation}
The constraints $\calC_{\rob}$ are enforced as polynomial constraints, so that for each constraint $f(\Xprox,W,v) = 0$ in $\calC_{\rob}$, we get that $\pE[f(\Xprox,W,v) \cdot q(\Xprox,W)] = 0$ for all polynomials $q$ so that $\deg(f\cdot q) \le d$.

The second set of constraints utilizes the \LSH hierarchy.
As explained in \pref{def:LSH}, although there are infinitely many polynomials of degree at most $d = O(k)^{t+1}$, we need only consider $\exp(O(d^2))$ many polynomials: those corresponding to lumber of degree at most $d$.
An alternative basis, such as an orthonormalized basis of graphical polynomials, would yield the same results, but is less convenient in our proofs.
\begin{claim}\label{claim:nG}
	There are at most $(2d+2)^{d(d+1)}$ lumber with degree at most $d$.
\end{claim}
\begin{proof}
	By Cayley's Theorem, there are $(d + 1)^{d-1}$ unrooted trees with $d + 1$ vertices, and thus loosely at most $(d + 1)^d$ rooted trees with exactly $d + 1$ vertices (not accounting for isomorphism).

	Notice that a lumber of degree at most $d$ can be written as a tuple of trees $T^\ast,T_1,\ldots,T_k$ such that all of the trees $T_1,\ldots,T_k$ appear as trunks, $T^\ast$ is the underlying tree, and the sum of the degrees of all of these trees is at most $d$.

	As a loose upper bound, there are $\binom{2d}{d + 1} \le (2d)^{d+1}$ ways to choose the degrees of these $k + 1$ trees (using stars and bars).
	Then, there are at most $\left((d + 1)^d\right)^{d+1}$ ways to choose the trees themselves, so the total number of such lumber is at most $(2d)^{d+1}\cdot (d+1)^{d(d+1)} \le (2d+2)^{d(d+1)}$ as desired.
\end{proof}

For convenience, we let $\calT$ denote the set of all trees $T$ with degree at most $d$ and $\calL$ the set of all lumber $T$ with degree at most $d$.
Let us also write $N_\calL \le (2d+2)^{d(d+1)}$ as the total number of lumber of degree at most $d$, as shown above.

Denote $\calP'_n$ as the joint distribution on $(X, \vAMP(X))$: that is, it is the joint distribution on matrices and the computed $t$-step AMP solution for them, with the normalization $\frac 1n\E[\|\vAMP(X)\|_2^2] = 1$.

We set slack parameters $c_\sla = \frac{\eta}{s(n) \cdot N_\calL^2}$ for $s(n)$ a slowly-growing function of $n$, and $C_K = O(K^2)$ (recall $K/\sqrt{n}$ is the subgaussian parameter of $X$).
Our \LSH constraints $\calC_{\LSH}$ ask that the joint statistics corresponding to $T(X)$ in $(\Xprox,v)$ match those of $(X,\vAMP(X))$ up to a slack which is satisfied by $(X,\vAMP(X)) \sim \calP'_n$ with high probability.
\begin{equation}
	\label{eqn:losh-constraints}
	\calC_{\LSH} = \begin{cases}
		 & \frac 1n\|v\|^2_2 = 1                                                                                                                                                   \\
		\forall T_1, T_2\in \calL,
		 & \pE\left[\frac 1n\langle T_1(\tX),T_2(\tX)\right] = \operatorname*{\E}\limits_{(Z, v^\ast)\sim \calP'_n}\left[\frac 1n\langle T_1(Z), T_2(Z)\rangle\right] \pm c_{\sla} \\
		\forall T\in \calL,
		 & \pE[\frac 1n \langle T(X), v\rangle] = \operatorname*{\E}\limits_{(Z, v^\ast)\sim \calP'_n}[\frac 1n\langle T(Z),v^\ast\rangle] \pm c_\sla                              \\
		\forall T\in \calT\text{ and } i\in [n],
		 & T(\tX)_i^4 \le (5C_K\deg(T)\log n)^{2\deg(T)}                                                                                                                             \\
		 & \Xprox^2 = 5\Id - BB^\top                                                                                                                                               \\
	\end{cases}
\end{equation}
\begin{remark}
	The ``full'' Local Statistics Hierarchy as defined in \cite{BMR21} would include constraints corresponding to all symmetric polynomials of degree at most $d$;
	here we include only the subset of those constraints used in our proofs.
	In order to deal with the adversarial corruptions we must also incorporate ``infinity norm'' constraints on $T(\Xprox)_i^4$ and an operator norm constraint on the matrix $\Xprox^2$; neither set of constraints is part of the \LSH as originally proposed, but both are enforceable with low-degree \sos.
	The ``infinity norm'' constraints would be implied by a ``coordinate-wise'' version of the \LSH, in which polynomials which are fixed by the action of the symmetric group on $[n]\setminus \{i\}$ for each $i\in[n]$ are also constrainted to be within the typical range of their expectations.
\end{remark}

\paragraph{Feasibility.}
We will set our slack parameters so that the constraint system $\calC$ is feasible if $\Xprox = X$.
Indeed, since we hope $\calC$ will force $\Xprox$ to be have like $X$, $\Xprox \approx X$ must be a feasible solution with high probability.
We define a ``reasonable sample'' to be an $X$ which is feasible for this program, and show that most $X$ are reasonable.

\begin{definition}[Reasonable sample]
	\label{def:reasonable-sample}
	Let $d\in \N$ be fixed, and set $c_{\sla} = \frac{\eta}{s(n)\cdot N_\calL^2}$ for any $s(n)$ so that $\omega(1/\sqrt{n}) = c_\sla = o_n(\eta)$ in the polynomial denoiser case, and in the Lipschitz denoiser case $s(n)$ is the rate of convergence implied by Theorem 1 of \cite{JM13}.
	A sample $X\sim \calD_n$ with $X\sim \R^{n\times n}$ is \emph{reasonable} at degree $d$ if the following two conditions are satisfied:
	\begin{itemize}[noitemsep]
		\item (Concentration of lumber) For all lumber $T_1, T_2 \in \calL$ of degree at most $d$ we have
		      \[\frac 1n \langle T_1(X), T_2(X)\rangle = \E_{(Z, v^\ast)}\left[\frac 1n \langle T_1(Z), T_2(Z)\rangle\right] \pm c_\sla\]
		      and
		      \[\frac 1n \langle T_1(X),\vAMP\rangle = \E_{(Z, v^\ast)}\left[\frac 1n \langle T_1(Z), v^\ast\rangle\right] \pm c_\sla.\]

		\item (Concentration of $\vAMP$) $\frac 1n \|\vAMP\|_2^2 = 1\pm \frac {\eta}{48}$.

		\item (Infinity norm of trees) For all trees $T\in \calT$ of degree $d' \le d$ we have
		      \[\frac 1n \|T(X)\|_4^4 \le \|T(X)\|_{\infty}^4 \le (5C_Kd'\log n)^{2d'}.\]

		\item (Bounded operator norm of $X$) $\|X\|^2_{\mathsf{op}} \le 5$.
	\end{itemize}
\end{definition}

\begin{lemma}[Most samples are reasonable]\torestate{
		\label{lem:most-samples-reasonable}
		Fix $d\in \N$.
Then
		\[\operatorname*{\mathbf{Pr}}\limits_{X\sim \calP_n}[X\text{ is reasonable at degree $d$}] \ge 1 - o_n\left(1\right).
		\]}
\end{lemma}
This requires us to verify the concentration of the constrained functions of $X$; we give the proof in \pref{app:hyp}.

\paragraph{The algorithm.}
With $\calC$ specified, we are ready to describe our algorithm.

\begin{algorithmSELF}[Robust AMP polynomial recovery]
	\label{algo:robust-perm-recovery}
	\textbf{Input: }
	A scalar $\eta > 0$ and a matrix $Y\in \R^{n\times n}$, given as a $\varepsilon$-corruption of the latent $X\sim \calP_n(\calD,\calK)$.

	\noindent \textbf{Operation: }
	\begin{itemize}[noitemsep]
		\item Compute a pseudodistribution $\xi$ satisfying $\calC = \calC_{\rob} \cup \calC_{\LSH}$.
		\item Return $\vSOS$, the top eigenvector of $\pE_\xi[vv^\top]$, normalized to be unit.
	\end{itemize}
	\noindent \textbf{Output: }
	Vector $\vSOS\in \R^n$ such that with high probability over the choice of $X$, \[\langle \vSOS, \vAMP(X)\rangle^2 \ge 1 - \sqrt{\eps} \cdot O(\log n)^{(2k)^{t}} - \eta.
	\]
\end{algorithmSELF}

Our guarantee on $\vSOS$ is symmetric up to sign.
Note that the above algorithm does not include a rounding procedure for ensuring $\vSOS\in \calK$; this reflects the fact that AMP algorithms typically produce a final iterate $\vAMP$ which has to be rounded to a vector in $\calK$.
\footnote{The rounding depends on the problem: for example, the standard rounding for nnPCA simply sets $\hat{v} = \frac{(\vSOS)_+}{\|(\vSOS)_+\|_2}$, while rounding for SK is more involved.}
Since $\vSOS$ is close to the true $\vAMP$, we can substitute it for $\vAMP$ in the AMP rounding procedures.

We are now ready to prove that \pref{algo:robust-perm-recovery} works, which implies \pref{thm:ramp-recovery-1}.
The proof strategy is as in the overview: we show that $\vAMP$ constructs a weighted forest, and thus decompose it into the constituent lumber.
For each lumber $T$, we then show by induction the closeness of $T(X)$ and $T(\tX)$ and thus complete the proof.
The conclusion is summarized by the following lemma, which we will prove in \pref{sec:cor-basis}.

\begin{lemma}[Reasonable samples guarantee correlation with AMP solution]
	\label{lem:xpp-vec-close}
	If $X$ is a reasonable sample at degree $(2k)^{2t}$, then a degree-$(2k)^{2t}$ pseudodistribution satisfying constraints $\calC$ has
	\[\frac 1{n}\pE\left[\left\langle v,\frac{\vAMP(X)}{\|\vAMP(X)\|}\right\rangle^2\right] \ge 1 - \sqrt{\eps} \cdot O(\log n)^{(2k)^{t}} - \frac \eta 2.
	\]
\end{lemma}

\subsection{Analysis of rounding scheme given correlation with AMP solution}
Before proving \pref{lem:xpp-vec-close}, we will use it to give a simple proof that the eigenvector rounding succeeds (establishing \pref{thm:ramp-recovery-1}).

\begin{proof}[Proof of Theorem~\ref{thm:ramp-recovery-1}]
	We first establish that the SDP is feasible with high probability, and then argue that the eigenvector rounding step succeeds.
	By \pref{lem:most-samples-reasonable}, $X$ is reasonable with high probability; we condition on the reasonableness of $X$ from here on.

	\paragraph{Feasibility.}
	Since $X$ is reasonable, the program defined by $\calC$ is feasible with $\Xprox = X$ for all $n\ge n_0$.
	Indeed, $\calC_{\rob}$ are satisfied by taking an arbitrary $\varepsilon$ fraction of rows and setting $W_{i} = 1$.
	Choosing $v = \vAMP(X)$, the \LSH constraints then follow by reasonableness.

	\paragraph{Correctness.}
	Define $R = \frac 1n\pE[vv^\top]$.
	We have that $\tr(R) = 1$ by $\calC_{\rob}$ and let $w = \frac{\vAMP}{\|\vAMP\|_2}$.
	Since $R\succeq 0$ and $\tr(R)=1$,
	\[1 - \kappa \le w^\top Rw\le 1\]
	for $\kappa = \sqrt{\eps} \cdot O(\log n)^{(2k)^{t}} + \frac \eta 2$ (the lower bound is from \pref{lem:xpp-vec-close}).
	We assume $\kappa < \frac 12$ (which can be done by decreasing $\eps$ by a factor of 4).
	Now, let $u$ be the top eigenvector of $R$ with corresponding eigenvalue $\lambda$.
	Since $u^\top Ru \ge w^\top R w$, it follows that $\lambda \ge 1 - \kappa$ and the sum of all other eigenvalues is at most $\kappa$.
	Hence, by looking at the eigendecomposition of $R = \sum_{i=1}^n \lambda_i v_i v_i^\top$ we have

	\[\langle w,u\rangle^2 = \frac 1\lambda\left(w^\top R w - \sum_{i=2}^{n}\lambda_i \langle w,v_i\rangle^2\right) \ge 1\left(w^\top R w - \kappa\right) \ge 1 - 2\kappa.
	\]

	Recalling that $w = \frac{\vAMP}{\|\vAMP\|_2}$ and $u=\vSOS$, we have established the claim of Theorem~\ref{thm:ramp-recovery-1}.
	\qedhere

\end{proof}

\subsection{Correlation of SDP with AMP output}
\label{sec:cor-basis}
The rest of this section is devoted to a proof of \pref{lem:xpp-vec-close}.
To begin, it will be vital to us that the result of running AMP is a weighted forest.

\begin{theorem}[AMP creates a forest]
	Suppose the denoisers $\calF = \{f^i\}$ define an AMP algorithm with polynomial $f^i$ having all coefficients independent of $n$ and degree bounded by $k$, and with starting iterate $x^0 = \vec{1}$.
	Then as a function of $X$, $x^t$ is a weighted forest of degree at most $(2k)^t$.
	Furthermore, $\vAMP = P(X)$ is also a weighted forest, with probability $1 - o_n(1)$.
\end{theorem}

\begin{proof}
	First, we claim that $\frac 1n\|x^t\|_2^2 = \Omega(1)$ with high probability.
	Indeed, this is a property satisfied by iterates produced by any reasonable AMP algorithm---by state evolution (\pref{thm:state-evolution}), notice that
	\[\plim_{n\rightarrow\infty}\frac 1n\|x^t\|_2^2 = \plim_{n\rightarrow\infty}\frac 1n\sum_{i=1}^{n}(x^t_i)^2 = \E[U_t^2]\]
	for $U_t$ a centered Gaussian with covariance only depending on the polynomials $f^i(\cdot)$, which (by definition of being separable and independent of $n$) have constant coefficients.
	This immediately implies that $\frac 1n\|x^t\|_2^2 = \Omega(1)$ with probability $1 - o_n(1)$.
	Therefore, if $x^t$ is a weighted forest then we can normalize $\vAMP = P(X) = \frac{x^t}{(\frac{1}{n}\E[\|x^t\|^2])^{1/2}}$ without introducing coefficients which depend on $n$.

	Next, let us prove that $x^t$ defines a weighted forest.
	Define $\deg(t) = (2k)^{t}$.
	The proof is by induction on $t$.
	Trivially, $x^0 = \vec 1$ is a tree, lumber, and weighted forest as well.

	For the inductive step, recall the AMP iteration
	\[x^{t+1} = Xf^t(x^t,x^{t-1},\ldots,x^0) - \sum_{j=1}^{t}b_{t,j}f^{j-1}(x^{j-1},\ldots,x^0).
	\]

	We first prove that $f^j(x^j,x^{j-1},\ldots,x^0)$ is a weighted forest whenever $x^j,x^{j-1},\ldots,x^0$ are, for any $j$.
	Indeed, expanding $f$ in the monomial basis,
	\[
		f^j(x^j,x^{j-1},\ldots,x^0) = \sum_{\substack{W = (w_0,\ldots,w_j) \in [k]^j\\ \sum_{i=0}^{j} w_i\le k}} c_W \cdot \operatornamewithlimits{\bigcirc}_{i=0}^{j}(x^i)^{\circ w_i}.
	\]

	Since each $x^i$ is a weighted forest, it follows by~\pref{fact:weighted-forest-linear} that $\operatornamewithlimits{\bigcirc}_{i=0}^{j}(x^i)^{\circ w_i}$ is also a weighted forest (this is a finite product of weighted forests).
	Then, again from~\pref{fact:weighted-forest-linear} it follows that the sum of all such terms is also a weighted forest and thus $f^j(x^j,x^{j-1},\ldots,x^0)$ is also a weighted forest, of degree at most $k\cdot \deg(j)$ (recall that the inner arguments are themselves polynomials of degree at most $\deg(j)$ in $X$).

	Next, we wish to show that $b_{t,j}f^{j-1}(x)$ is a weighted forest.
	Recall that
	\[b_{t,j} = \frac 1n\sum_{i=1}^{n}\frac{\partial f^t(x)_i}{x_i^j} = \frac 1n \langle \grad_j f^t, \vec{1}\rangle.
	\]
	Here, $\grad_j$ is the gradient relative to $f^t$'s $j$th argument, $x^j$.
	Since $f^t(x)$ is a weighted forest, it follows that $\grad_j f^t$ is a weighted forest.
	Hence $b_{t,j}$ is a weighted sum of trunks.
	Notice that coefficient increases by a factor of at most $k$ by the above expansion of $f^j$ and the definition of a derivative, so the coefficients remain $O(1)$ so long as $t = O(1)$.

	Therefore, since the product of a sum of trunks and a forest is a weighted forest, it follows that $b_{t,j} f^{j-1}$ is a weighted forest.
	The degree of this weighted forest is at most $\deg (b_{t,j}) + \deg(f^{j-1})$.
	As $\deg(b_{t,j}) = \deg(f^t) - 1 \le k\cdot \deg(t)$ and $\deg(f^{j-1}) \le k\cdot \deg(j - 1)$, it follows that $\deg(b_{t,j}) + \deg(f^{j-1}) \le 2k\deg(t)$.

	Therefore, we have that $x^{t+1}$ is a sum of weighted forests and thus is itself a weighted forest, of degree at most $\max(1 + k\deg(t), 2k\cdot \deg(t)) \le 2k\cdot \deg(t) = \deg(t + 1)$ and our induction is complete.
\end{proof}

Now, to prove \pref{lem:xpp-vec-close}, we will use the forest nature of AMP and appeal to the constraints of $\calC_{\LSH}$.
In particular, we'll use the following two lemmas to deduce \pref{lem:xpp-vec-close}.

\begin{lemma}
	\label{lem:forest-close}
	Suppose the denoiser functions $\calF =\{f^i\}$ are polynomials with all coefficients independent of $n$ and degree bounded by $k$.

	Suppose $P(X)$ and $P(\tX)$ are the $t$'th iterates of the AMP algorithm defined by $\calF$ when applied to $X$ and $\tX$ respectively, normalized so that $\frac 1n \E[\|P(X)\|_2^2] = 1$.
	Then, the following inequality holds in \sos:
	\[\calC\sststile{4\cdot (2k)^{t} + 8}{\tX, W}\left\{\|P(X) - P(\tX)\|_2^4 \le \eps \cdot O(\log n)^{2\cdot (2k)^{t}}\cdot \|P(X)\|_2^4\right\}.
	\]
\end{lemma}

\begin{lemma}[$v$ approximates $\vAMP(X)$]
	\label{lem:v-approx-robust}
	As a consequence of $\calC_\LSH$, the following statement holds:
	\[\frac 1n \pE_\zeta \left[\|v - P(\tX)\|_2^2\right] \le \frac \eta {16}\]
	where $\zeta$ is a degree $2\cdot (2k)^t$ pseudodistribution.
\end{lemma}

\begin{proof}[Proof of Lemma \ref{lem:xpp-vec-close}]
	By the Pseudo-Expectation Cauchy-Schwarz inequality, the duality of \sos proofs and Pseudo-distributions (as applied to~\pref{lem:forest-close}), and the reasonableness of $X$, it follows that
	\[\frac 1n\pEB{\|P(X) - P(\tX)\|_2^2}\le \frac 1n\sqrt{\pEB{\|P(X) - P(\tX)\|_2^4}} \le \sqrt{\eps} \cdot O(\log n)^{(2k)^{t}}\cdot \frac 1n\|P(X)\|_2^2 = \sqrt{\eps} \cdot O(\log n)^{(2k)^{t}}.\]
	Next, by the Pseudo-Expectation Almost-Triangle Inequality and recalling that $\vAMP = P(X)$, we have
	\[\frac 1n\pE[\|v - \vAMP\|_2^2] \le \frac 2n\left(\pEB{\|v - P(\tX)\|_2^2} + \pEB{\|P(\tX) - P(X)\|_2^2}\right) \le \frac \eta 8 + \sqrt{\eps} \cdot O(\log n)^{(2k)^{t}}.\]
	The former of these bounds is from~\pref{lem:v-approx-robust} and the latter from the previous equation.

	Finally, we may rewrite
	\[\frac 1n \pEB{\|v - \vAMP\|_2^2} = \frac 1n \pEB{\|v\|_2^2} + \frac 1n \pEB{\|\vAMP\|_2^2} - \frac 2n \pEB{\langle v, \vAMP\rangle} = 2\pm \frac \eta 8 - \frac 2n \pEB{\langle v, \vAMP\rangle}\]
	using $\calC_\LSH$ and~\pref{lem:most-samples-reasonable} which together with the above implies that for $\eta \le 1$,
	\[\frac 1n \pEB{\left\langle v, \frac{\vAMP}{\frac{1}{\sqrt{n}}\|\vAMP\|}\right\rangle}
		\ge\frac{2 \pm \frac{\eta}{8} - \frac{\eta}{8} - \sqrt{\eps}O(\log n)^{(2k)^t}}{2\sqrt{1\pm \frac{\eta}{48}}}
		\ge 1 - \sqrt{\eps} \cdot O(\log n)^{(2k)^{t}} - \frac \eta4.
	\]
	To achieve exactly the statement as written, we finish by using that
	\[
		\frac 1{n}\pEB{\left\langle v,\frac{\vAMP}{\|\vAMP\|}\right\rangle^2} \ge \frac 1{n^2}\pEB{\left\langle v,\frac{\vAMP}{\frac{1}{\sqrt n}\|\vAMP\|}\right\rangle}^2 \ge 1 - \sqrt{\eps} \cdot O(\log n)^{(2k)^{t}} - \frac \eta2.\qedhere\]

\end{proof}

Therefore, it remains to prove~\pref{lem:forest-close} and~\pref{lem:v-approx-robust}.
We begin with the latter, which comes down to analyzing the contributions of variances coming from $\calC_\LSH$.

\begin{proof}[Proof of Lemma~\ref{lem:v-approx-robust}]
	Begin by writing
	\[\frac 1n \pEB{\|v - P(\tX)\|_2^2} = \frac 1n \pEB{\|v\|_2^2} + \frac 1n \pEB{\|P(\tX)\|_2^2} - \frac 2n \pEB{\langle v, P(\tX)\rangle}.\]
	The first term here is $1$ by $\calC_\LSH$, so it suffices to bound the remaining terms.

	To this end, expand $P(\tX) = \sum_{T\in \calL}c_T T(\tX)$.
	Then,
	\begin{align*}\frac 1n \pEB{\|P(\tX)\|_2^2} & = \sum_{T_1, T_2\in \calL}c_{T_1} c_{T_2} \cdot \pEB{\frac 1n \langle T_1(\tX), T_2(\tX)\rangle}                                                   \\
        & = \sum_{T_1,T_2} c_{T_1}c_{T_2} \cdot \E_{(Z,v^\ast)}\left[\frac 1n\langle T_1(Z), T_2(Z)\rangle\right] \pm c_\sla \sum_{T_1,T_2} |c_{T_1}c_{T_2}| \\
        & = \frac 1n \E_{(Z,v^\ast)}\left[\|v^\ast\|_2^2\right] \pm c_\sla \cdot O(N_\calL^2)
              \,=\, 1 \pm \frac {\eta}{48}
	\end{align*}
	by taking large enough $n$ and applying $\calC_\LSH$ to each $\frac 1n \langle T_1(\tX), T_2(\tX)\rangle$ term and~\pref{claim:nG} to the number of trees.

	Similarly, we can write
	\begin{align*}
		\frac 1n \pEB{\langle v,P(\tX)\rangle} & = \sum_{T\in \calL} c_T \cdot \pEB{\frac 1n\langle v, T(\tX)\rangle}                                                          \\
		  & = \sum_{T} c_T \cdot \E_{(Z,v^\ast)}\left[\frac 1n \langle v^\ast, T(Z)\rangle\right] \pm c_\sla \cdot \sum_{T\in \calL}|c_T| \\
		  & = \E_{(Z,v^\ast)}\left[\frac 1n \|v^\ast \|_2^2\right] \pm \frac{\eta}{48}
		\, =\, 1 \pm \frac{\eta}{48}.
	\end{align*}

	Thus, we are left with
	\[\frac 1n \pEB{\|v - P(\tX)\|_2^2} = 1 + \left(1\pm \frac{\eta}{48}\right) - 2\left(1\pm \frac{\eta}{48}\right) \le \frac{\eta}{16}\]
	as desired.
\end{proof}

To prove~\pref{lem:forest-close}, we express weighted forests in terms of their constituent lumber and then study the closeness of lumber.
To do so, we must also define a notion of \emph{matching variables} which relate $\tX$ and $X$.

\begin{definition}[Matching Variables]
	\label{def:match-constraints}
	For each index $j\in [n]$, define $M_j = W_j\cdot \mathbf{1}[Y_{\cdot,j}=X_{\cdot,j}]$.
	Then, $M_j$ is the indicator that the $j$'th column of $\tX$ aligns with that of $X$: in other words, $M_j^2 = M_j$ and $M_j\Xprox_{ij} = M_j X_{ij}$.
	We also define $D_M = \diag(M)$, the diagonal matrix with $M$ on the diagonal.
\end{definition}

Though the indicators $\Ind[Y_{\cdot,j} = X_{\cdot,j}]$ are not known to us, we may still make use of the matching variables $M_i$be  in our \sos proofs:

\begin{lemma}\torestate{
		\label{lem:matching-2eps}
		The following two inequalities are implied by the constraints $\calC_\rob$:
		\begin{itemize}[noitemsep]
			\item $(X - \Xprox) (I - D_M) = X - \Xprox$
			\item For $M$ the vector in $\R^n$ with $i$th coordinate $M_i$, $\|\vec 1 - M\|_4^4 \le 2\eps n$
		\end{itemize}}
\end{lemma}

\begin{proof}
	We may rearrange the first of these inequalities to $(X - \tX)D_M = 0$: or, equivalently, that $M_j(X_{ij} - \tX_{ij}) = 0$ for all $i, j$.
	However, we may expand this and see that
	\[\calC_\rob \sststile{2}{M, \tX}\biggl\{M_j(X_{ij} - \tX_{ij}) = \mathbf{1}[Y_j=X_j]W_j(X_{ij} - \tX_{ij}) = W_j\mathbf{1}[Y_j=X_j](X_{ij} - Y_{ij}) = W_j\mathbf{1}[Y_j=X_j](X_{ij} - X_{ij}) = 0\biggr\}\]
	using that $W_j \tX_{ij} = W_j Y_{ij}$ and $\mathbf{1}[Y_j=X_j]X_{ij} = \mathbf{1}[Y_j=X_j]Y_{ij}$.
	Therefore, we have proven the first statement.

	For the second statement, notice first that Booleanity constraints are implied for $1-M_j$:
	\[\calC_\rob \sststile{2}{M}\biggl\{(1 - M_j)^2 = 1 - 2M_j + M_j^2 = 1 - M_j\biggr\}.
	\]
	Similarly, we claim that
	\begin{align}\calC_\rob \sststile{2}{M}\{1 - M_j \le (1 - W_j) + (1 - \mathbf{1}[Y_j=X_j])\}.\label{eq:mj}
	\end{align}
	To prove this, note that for any $x, y$
	\[1 - xy \le (1 - x) + (1 - y) \iff 1 - x - y + xy \ge 0 \iff (1 - x)(1 - y)\ge 0.
	\]
	Since $0\le W_j,\mathbf{1}[Y_j = X_j] \le 1$, this proves~\pref{eq:mj}.
	This can alternatively be seen as a consequence of the union bound.

	Now, may explicitly rewrite the second inequality as
	\begin{align*}
		\calC_\rob \sststile{8}{M}\Biggl\{\|\vec 1 - M\|_4^4 = \sum_{i=1}^{n}(1 - W_j \mathbf{1}[Y_j=X_j])^4 = \sum_{i=1}^{n}(1 - W_j \mathbf{1}[Y_j=X_j]) &          \\
		\le \sum_{i=1}^{n}(1 - W_j) + \sum_{i=1}^{n}(1 - \mathbf{1}[Y_j=X_j]) \le 2\eps n                                                                  & \Biggr\}
	\end{align*}

	as desired.\qedhere

\end{proof}

This leads us to the crux of the proof: $X$ and $\tX$ are close for any tree polynomial input.

\begin{lemma}
	\label{lem:tree-close}
	Suppose that $T$ is a tree with $d$ edges.
	Then,
	\[\calC\sststile{4d+8}{\tX,M}\left\{\frac 1{n^2}\left\|T(X) - T(\tX)\right\|_2^4 \le \eps \cdot (40C_Kd\log n)^{2d}\right\}\]
	where $C_K > 0$ is a constant only depending on $d$ and $K$, the subgaussian constant of $\sqrt n X_{ij}$.
\end{lemma}
\begin{proof}

	We prove a slightly stronger statement.
	In particular, suppose that exactly $q$ ``grafting'' operations were done during the construction of $T$ (this essentially corresponds to a slightly shifted sum of the degrees of vertices in $T$).
	Then, we prove that
	\[\calC\sststile{4d+8}{\tX,M}\left\{\frac 1{n^2}\left\|T(X) - T(\tX)\right\|_2^4 \le 8^q\cdot \eps \cdot (5C_Kd\log n)^{2d}\right\}.
	\]
	Note that as $q\le d\le 2d$, this implies the original claim.

	The proof is by structural induction on $T$.
	We take the base case to be when $T$ is an empty tree (a single node): then $T(X) = T(\tX) = \vec 1$ and thus $\frac 1{n^2}\left\|T(X) - T(\tX)\right\|_2^4 = 0$.

	By definition, a tree can be formed recursively either by re-rooting or by grafting.
	Consider first the case when $T$ is a re-rooted version of $T'$.

	We will split $T$ into $T'$ and the re-rooting edge, then apply the almost-triangle inequality $(a+b)^4 \le 8(a^4 + b^4)$:
	\begin{align}\sststile{4d}{\tX}\Biggl\{\frac 1{n^2}\left\|T(X) - T(\tX)\right\|_2^4 & = \frac 1{n^2}\left\|\frac 12(X - \tX)(T'(X) + T'(\tX)) + \frac 12(X + \tX)(T'(X) - T'(\tX))\right\|_2^4\nonumber                                  \\
       & \le \frac{1}{2n^2}\left[\left\|(X - \tX)(T'(X) + T'(\tX))\right\|_2^4 + \left\|(X + \tX)(T'(X) - T'(\tX))\right\|_2^4\right]\Biggr\}.\label{eq:p1}
	\end{align}
	For the second term, we will use that the operator norms of $X$ and $\tX$ are bounded (the latter from $\calC_{\LSH}$):
	\begin{claim}
		\label{claim:op-norm-sos}
		If $w$ is an SoS indeterminate, then for a reasonable sample $X$ the \sos constraints $\calC_{\LSH}$ imply that
		\[
			\calC_\LSH\sststile{4(\deg(w)+1)}{w,\tX}\biggl\{\|(X + \tX) w\|_2^4 \le 400\|w\|_2^4 \quad \text{ and } \quad \|(X - \tX) w\|_2^4 \le 400\|w\|_2^4\biggr\}.
		\]
	\end{claim}
	\begin{proof}
		The constraints $\calC_{\LSH}$ imply that
		\[
			\|\tX w\|_2^2
			= w^\top \tX^2 w
			= w^\top 5I w - w^\top B^\top Bw
			= 5\|w\|_2^2 - \|Bw\|_2^2
			\le 5\|w\|_2^2
		\]
		and thus $\|\tX w\|_2^4 \le 25\|w\|_2^4.
		$
		Since $X$ is a reasonable sample, $\|Xw\|_2^4 \le \|X\|_{\op}^4\|w\|_2^4 \le 25\|w\|_2^4$ as well.
		Hence, by the \sos Almost-Triangle Inequality (\pref{fact:sos-almost-triangle}), it follows that
		\[
			\|(X + \tX)w\|_2^4
			\le 8(25\|Xw\|_2^4 + 25\|\tX w\|_2^4)
			\le 400\|w\|_2^4
		\]
		as desired.
		The same proof gives the result for $X - \tX$.
	\end{proof}

	Using~\pref{claim:op-norm-sos}, we have that
	\[
		\calC \sststile{4d+4}{\tX}
		\Biggl\{
		\left\|(X+X')(T'(X) - T'(\tX))\right\|_2^4
		\le 400\left\|T'(X) - T'(\tX))\right\|_2^4
		\le 400\cdot 8^q\cdot \eps \cdot n^2\cdot (5C_Kd \log n)^{2(d-1)}
		\Biggr\},
	\]
	where we have applied the inductive step.
	The \sos degree comes from noting that the first inequality has degree $4(d - 1 + 1)= 4d$ by~\pref{claim:op-norm-sos} and the second inequality has degree $4(d-1)+8 = 4d + 4$ by induction.

	Now, to bound the first term, we use the matching variables $M_i$ as defined in~\pref{def:match-constraints}.
	Then we have that (using~\pref{lem:matching-2eps} for the first and last steps and~\pref{claim:op-norm-sos} for the second step)
	\begin{align*}
		\calC_\rob \sststile{4d+8	}{\tX,M}
		\Biggl \{
		\left\|(X - \tX)(T'(X) + T'(\tX))\right\|_2^4
		&= \left\|(X - \tX)(I - D_M)(T'(X) + T'(\tX))\right\|_2^4\\
		&\le 400\left\|(\vec 1 - M)\circ     (T'(X) + T'(\tX))\right\|_2^4\\
		&\le \|\vec 1 - M\|_4^4 \cdot \|T'(X) + T'(\tX)\|_4^4 \\
		&\le 800\eps n\cdot \|T'(X) + T'(\tX)\|_4^4
		\quad\Biggr\}.
	\end{align*}

	From the \sos Almost-Triangle Inequality we have \[
		\calC_{\LSH}\sststile{4d-4}{\tX}
		\left\{\left\|T'(X) + T'(\tX)\right\|_4^4 \le 8\left\|T'(X)\right\|_4^4 + 8\left\|T'(\tX)\right\|_4^4 \le 8n(5C_Kd\log n)^{2(d-1)}\right\},\] the final inequality by our bounds from \pref{lem:tree-inf-bd} and $\calC_\LSH$.
	Combining these bounds together,
	\begin{align*}
		\calC \sststile{4d+8}{\tX}\left\{\pref{eq:p1} \le \frac{1}{2}\left(800\eps \cdot 8(5C_Kd\log n)^{2(d-1)} +  400 \cdot 8^q\cdot \eps(5C_Kd\log n)^{2(d-1)} \right) \le8^q\cdot \eps  (5C_K d \log n)^{2d}\right\},
	\end{align*}
	for $n$ large enough.

	Now we consider the second case, in which $T$ was formed by a grafting of two trees: $T(X) = T_1(X) \circ T_2(X)$, where $T_1$ has $d_1$ total edges (and $q_1$ graftings) and $T_2$ has $d_2$ total edges (and $q_2$ graftings).
	Here we will use two facts about Hadamard products:
	\begin{enumerate}
		\item Firstly, note that for (any) vectors $a, b, c, d$ we have
		      $a\circ b - c\circ d = \frac 12 (a + c)\circ (b - d) + \frac 12 (b + d)\circ (a - c).$
		      This follows by distributivity and noticing that the Hadamard product is separable over coordinates.

		\item Secondly, we have that
		      \[\{\forall i, a_i^2 \le C^2\}\sststile{4(\deg(a) + \deg(b))}{a,b} \biggl\{\|a\circ b\|_2^4 \le C^4 \|b\|_2^4\biggr\}.\] for all indeterminates $a, b$ and constants $C$.
		      To preserve readability, we will use this instead as
		      $\sststile{}{} \biggl\{\|a\circ b\|_2^4 \le \|a\|_{\infty}^4 \|b\|_2^4\biggr\}$
		      where the infinity norm is a proxy for $C$.
	\end{enumerate}
	Similarly to the rerooting case, we begin by expanding
	\begin{align}
		\sststile{4d}{\tX}\Biggl\{\frac 1{n^2}\left\|T(X) - T(\tX)\right\|_2^4 & = \frac 1{n^2}\left\|\frac 12(T_2(X) + T_2(\tX))\circ(T_1(X) - T_1(\tX)) + \frac 12(T_1(X) + T_1(\tX))\circ(T_2(X) - T_2(\tX))\right\|_2^4\nonumber                                 \\
		 & \le \frac 1{2n^2}\left[\left\|(T_2(X) + T_2(\tX))\circ(T_1(X) - T_1(\tX))\right\|_2^4 + \left\|(T_1(X) + T_1(\tX))\circ(T_2(X) - T_2(\tX))\right\|_2^4\right]\Biggr\}.\label{eq:p2}
	\end{align}

	These two terms are essentially identical, so we will handle the first and claim the second by symmetry.
	By the second fact about Hadamard Products, it follows that
	\begin{align}\sststile{4d}{\tX}\left\{\left\|(T_2(X) + T_2(\tX))\circ(T_1(X) - T_1(\tX))\right\|_2^4 \le \left\|T_2(X) + T_2(\tX)\right\|_{\infty}^4\left\|T_1(X) - T_1(\tX))\right\|_2^4\right\}.\label{eq:p3}\end{align}
	Now, using the fact that for reasonable $X$ the maximum entry of $T_2(X)$ is bounded (see \pref{lem:tree-inf-bd}) and $\calC_\LSH$ to bound the infinity norm factor, and the induction hypothesis on the $T_1(X)-T_1(\Xprox)$ factor (since $T_2$ is a non-empty tree the induction hypothesis applies), we have that
	\[\calC \sststile{4d+8}{\tX} \left\{\pref{eq:p3} \le 8\cdot (5C_Kd_2\log n)^{2d_2}\cdot 8^{q_1}\eps\cdot n^2\cdot (5C_Kd_1\log n)^{2d_1}\le 8^{q_1}\cdot 8\eps\cdot n^2\cdot (5C_Kd\log n)^{2d}\right\}.
	\]
	The \sos degree comes from using a degree $4d_2$ proof to apply $\calC_\LSH$ and a $4d_1 + 8$ degree proof to apply induction.

	A similar statement follows for the second term in~\pref{eq:p2}: that is,
	\[\calC\sststile{4d+8}{\tX}\left\{\left\|(T_1(X) + T_1(\tX))\circ(T_2(X) - T_2(\tX))\right\|_2^4\le 8^{q_2}\cdot 8\eps\cdot n^2\cdot (5C_Kd\log n)^{2d}\right\}\]
	as well.
	Putting these past two equations together and noticing that $q_1, q_2 \le q$, it follows that
	\[\calC \sststile{4d+8}{\tX}\left\{\pref{eq:p2} \le \frac 1{2n^2}\left[8^{q_1}\cdot\eps\cdot n^2\cdot (5C_Kd\log n)^{2d} + 8^{q_2}\cdot\eps\cdot n^2\cdot (5C_Kd\log n)^{2d}\right] \le 8^{q}\cdot \eps \cdot (5C_Kd\log n)^{2d}\right\}\]
	which completes the proof.
\end{proof}

Now that we have shown closeness for trees, we can extend this to closeness for lumber, too.

\begin{lemma}
	\label{lem:lumber-close}
	Suppose that $T$ is a lumber with $d$ edges.
	Then, \[\calC \sststile{4d + 8}{\tX}\left\{\frac 1{n^2}\left\|T(X) - T(\tX)\right\|_2^4 \le \eps \cdot (160C_Kd\log n)^{2d}\right\}\]
	where $C_K > 0$ is the slack parameter chosen in $\calC_{\LSH}$ (as defined in~\pref{lem:tree-close}).
\end{lemma}

\begin{proof}
	Similarly to the tree case, we slightly strengthen the statement.
	In particular, suppose there are $\ell$ trunks on this lumber.
	Then, we prove that
	\[\calC \sststile{4d + 8}{\tX}\left\{\frac 1{n^2}\left\|T(X) - T(\tX)\right\|_2^4 \le 16^\ell\cdot \eps \cdot (40C_Kd\log n)^{2d}\right\}.
	\]
	Since $\ell\le d$, this implies the original statement.

	Any lumber is a tree, multiplied by a finite number of trunks.
	We proceed by structural induction on the number of trunks.
	In the base case, there are no trunks and the lumber is just a tree which is handled by~\pref{lem:tree-close}.

	Otherwise, suppose that $T(X) = \trunk_{T_1}(X)\cdot T_2(X)$, where $T_2$ is a lumber with degree $d_2$ and $\trunk(X) = \frac 1n\langle T_1(X),1\rangle$ is a trunk of degree $d_1$.

	Notice that, for a trunk, the following two inequalities hold by linearity and SoS Cauchy Schwarz (\pref{fact:sos-cauchy-schwarz}):
	\begin{align}\sststile{2d_1}{\tX} \biggl\{(\trunk(X) \pm \trunk(\tX))^2 & = \frac 1{n^2}\langle T_1(X) \pm T_1(\tX), 1\rangle^2 \le \frac 1n \|T_1(X) \pm T_1(\tX)\|_2^2\biggr\}.\label{eq:cs-trunk}
	\end{align}

	These inequalities allow us to transform a trunk into a tree, for which we have already proven closeness between $X$ and $\tX$ in~\pref{lem:tree-close}.

	Using this, we have
	\begin{align}
		\sststile{4d}{\tX}\Biggl\{\frac 1{n^2}\left\|T(X) - T(\tX)\right\|_2^4 = \frac 1{n^2}\left\|\frac 12(T_2(X) + T_2(\tX))(\trunk_{T_1}(X) - \trunk_{T_1}(\tX)) + \frac 12 (\trunk_{T_1}(X) + \trunk_{T_1}(\tX))(T_2(X) - T_2(\tX))\right\|_2^4 & \nonumber             \\
		\le \frac 1{2n^2}\left[(\trunk_{T_1}(X) - \trunk_{T_1}(\tX))^4\left\|T_2(X) + T_2(\tX)\right\|_{2}^4 + (\trunk_{T_1}(X) + \trunk_{T_1}(\tX))^4\left\|T_2(X) - T_2(\tX))\right\|_2^4\right]                                                   & \Biggr\}\label{eq:p4}
	\end{align}
	where we used that $\trunk_{T_1}(\cdot)$ is a scalar and thus can be safely factored out of the norm.

	To begin, the most ``mysterious'' term in this equation is $\left\|T_2(X) + T_2(\tX)\right\|_{2}^4$.
	However, using the SoS Almost-Triangle Inequality, $\calC_\LSH$, and~\pref{lem:tree-inf-bd} we have that
	\begin{align}\calC_\LSH \sststile{4d_2}{\tX}\Biggl\{\left\|T_2(X) + T_2(\tX)\right\|_{2}^4 & \le 8\|T_2(X)\|_2^4 + 8\|T_2(\tX)\|_2^4\nonumber                     \\
   & \le 8n^2\|T_2(X)\|_{\infty}^4 + 8n^2\|T_2(\tX)\|_{\infty}^4\nonumber \\
     &\le 16n^2\cdot (5C_Kd_2\log n)^{2d_2}\Biggr\}.\label{eq:42norm}
	\end{align}

	Now, let us handle each of the two terms in~\pref{eq:p4} separately.
	For the first, write
	\begin{align*}
		\calC\sststile{4d+8}{\tX}
		\Biggl\{
		(\trunk_{T_1}(X) - \trunk_{T_1}(\tX))^4\left\|T_2(X) + T_2(\tX)\right\|_{2}^4 & \le \frac 1{n^2}\|T_1(X) - T_1(\tX)\|_2^4\cdot 16n^2\cdot (5C_Kd_2\log n)^{2d_2} \\
&\le \eps\cdot (40C_Kd_1\log n)^{2d_1}\cdot 16n^2\cdot (5C_Kd_2\log n)^{2d_2}     \\
& \le 16n^2\cdot \eps\cdot (40C_Kd\log n)^{2d}
		\Biggr\}
	\end{align*}
	using~\pref{lem:tree-close},~\pref{eq:cs-trunk}, and~\pref{eq:42norm}.
	The \sos degree bound comes from needing $4d_1 + 8$ for the application of~\pref{lem:tree-close} and degree $4d_2$ to upper bound the $T_2$ term via $\calC_\LSH$.

	The second is nearly identical, but for different reasons:
	\begin{align*}
		\calC\sststile{4d+8}{\tX}
		\Biggl\{
		(\trunk_{T_1}(X) + \trunk_{T_1}(\tX))^4\left\|T_2(X) - T_2(\tX))\right\|_2^4 & \le \frac 1{n^2}\|T_1(X) + T_1(\tX)\|_2^4\cdot 16^{\ell - 1}\cdot \eps\cdot n^2\cdot (40C_Kd_2\log n)^{2d_2} \\
	& \le 16\cdot (5C_Kd_1\log n)^{2d_1}\cdot 16^{\ell - 1}\eps\cdot n^2\cdot (40C_Kd_2\log n)^{2d_2}                \\
	& \le 16^{\ell}n^2\cdot \eps\cdot (40C_Kd\log n)^{2d}
		\Biggr\}
	\end{align*}
	by~\pref{eq:42norm},~\pref{eq:cs-trunk}, and structural induction.

	Therefore, we have that
	\[\calC\sststile{4d+8}{\tX}\left\{\pref{eq:p4} \le \frac 1{2n^2}\left[16n^2\cdot \eps\cdot (40C_Kd\log n)^{2d} + 16^{\ell}n^2\cdot \eps\cdot (40C_Kd\log n)^{2d}\right] = 16^\ell\eps\cdot (40C_Kd\log n)^{2d}\right\}\]
	and the induction is complete.
\end{proof}

We are now ready to prove the lemma for AMP forest polynomials.

\begin{proof}[Proof of Lemma~\ref{lem:forest-close}]
	Decompose $P(X) = \sum_{T\in \calL} c_T T(X)$ in terms of its constituent lumber. Then,
	\begin{align*}
		\calC\sststile{4\cdot (2k)^t + 8}{\tX}\Biggl\{\frac 1{n^2}\|P(X) - P(\tX)\|_2^4 & \le N_\calL\sum_{T\in \calL}c_T^2 \frac 1{n^2}\|T(X) - T(\tX)\|_2^4           \\
& \le N_\calL\sum_{T\in \calL}c_T^2 \cdot \eps \cdot O(\log n)^{2\deg(T)}       \\
& \le  N_\calL\sum_{T\in \calL}c_T^2 \cdot \eps \cdot O(\log n)^{2\cdot (2k)^t} \\
& = \eps \cdot O(\log n)^{2\cdot (2k)^t}\Biggr\}
	\end{align*}
	using the fact that $N_\calL$ is independent of $n$ and each $c_T^2 = O(1)$ (by definition of a weighted forest).

	Finally, recall that $\frac 1{n^2}\|P(X)\|_2^4 = 1\pm \frac \eta {16}\ge \frac 12$ by~\pref{lem:most-samples-reasonable}, which completes the argument.
	\qedhere

\end{proof}

\subsection{Lipschitz denoiser functions}

As a corollary of our result with polynomials, we can also show that certain nicely-behaved Lipschitz AMP iterations can be simulated robustly.

We'll use the function $\deg(\delta, t) = O\biggl((tL(K + L)/\delta)^{4^{t}}\poly(2^t,\log\frac{1}{\delta},L,K)\biggr)$ throughout.

\begin{corollary}[Robust simulation of AMP for Lipschitz functions]
	\label{cor:apx-lipschitz}
	Fix $\eta = \Omega(1)$ and suppose the AMP denoisers $\calF = \{f^s: \R^{s+1}\rightarrow \R\}_{s\ge 1}$ satisfy
	\begin{itemize}[noitemsep]
		\item Each function is $L$-Lipschitz
		\item The partial derivatives $\frac{\partial f^s}{\partial x^j}$ are either pseudo-Lipschitz or indicators
		\item The state evolution (see~\pref{thm:state-evolution}) covariance matrix $Q^s$ corresponding to $f^s$ satisfies $Q\succeq I$ and $\max_{i,j}|Q_{ij}|\le 2$.
	\end{itemize}
	Then, if AMP is run for $t$ steps with the denoisers $\calF$ to produce $\vAMP$ from the $K / \sqrt{n}$-subgaussian distributed input $X$,
	the degree-$\deg(\eta,t)$ \LSH hierarchy defined above outputs a solution $\vSOS$ such that \[\frac{\|v_\LSH - \vAMP\|_2^2}{\|\vAMP\|_2^2} \le O\left(\sqrt \eps \cdot (\log n)^{2\deg(\eta,t)}\right) + 50\eta\]
	with probability $1 - o_n(1)$.
\end{corollary}

To prove this corollary, we require that AMP be approximable by a certain polynomial iteration instead.

\begin{proposition}[Nice AMP iterates are approximable by polynomials]\torestate{
		\label{prop:final-approx}
		Fix $\delta =\Omega(1)$ and suppose $f^t:\R^{t+1}\rightarrow \R$ have the properties noted in the above corollary, and let $\vAMP = \vAMP(X)$ denote the final ($t$'th) iterate, scaled so that $\frac 1n\E[\|\vAMP\|_2^2] = 1$.
		Then, there exists a weighted forest polynomial $P(X)$ of degree at most $\deg(\delta,t)$ such that with probability $1 - o_n(1)$ over the choice of $X$,
		\[\frac{\|\vAMP - P(X)\|_2}{\sqrt n} \le \delta,\]
and furthermore, there exists $n_0 \in \N$ such that for all $n \ge n_0$, $\frac 1n \E\left[\|\vAMP - P(X)\|_2^2\right] \le \delta^2.$}
\end{proposition}
This is proven in \pref{app:amp-poly}.

\begin{lemma}[Correlation to non-polynomial AMP]
	\label{lem:corr-non-poly-amp}
	Fix $\delta = \Omega(1)$.
	The guarantees of~\pref{lem:xpp-vec-close} hold approximately even if $\vAMP$ is not a polynomial: that is,
	\[\frac 1{n}\pE\left[\left\langle v,\frac{\vAMP(X)}{\|\vAMP(X)\|}\right\rangle^2\right] \ge 1 - \sqrt{\eps} \cdot O(\log n)^{\deg(\delta,t)} - \frac \eta 2 - 12 \delta\]
	whenever $\vAMP$ is constructed according to the guarantees of the above corollary.
\end{lemma}

\begin{proof}
	The proof is nearly identical to that of~\pref{lem:xpp-vec-close}, so we will highlight the main differences.

	Let us take $P(X)$ to be the polynomial of degree at most $\deg(\delta, t)$ approximating $\vAMP$, which must define a weighted forest.
	Then, we obtain
	\begin{align*}\frac 1n \pE[\|v - \vAMP\|_2^2] & \le \frac 3n\left(\pE[\|v - P(\tX)\|_2^2] + \pE[\|P(\tX) - P(X)\|_2^2] + \pE[\|P(X) - \vAMP\|_2^2]\right) \\
                                              & \le \frac 3n\pE[\|v - P(\tX)\|_2^2] + \sqrt{\eps}\cdot O(\log n)^{\deg(\delta,t)} + 3\delta^2.
	\end{align*}

	We can bound this first term by noticing that (similar to~\pref{lem:v-approx-robust}), for $n$ sufficiently large,
	\begin{align*}
\frac 1n\pE[\|v - P(\tX)\|_2^2] 
&= \frac 1n \pE\left[\|v\|_2^2+ \|P(\tX)\|_2^2 - 2\langle v, P(\tX)\rangle\right] \\
&= \frac{1}{n}\E[\|\vAMP-P(X)\|_2^2] \pm c_{\sla}\left(1 + \left(\sum_{T} |c_T|\right)^2 + \sum_T|c_T|\right)\\
&\le \delta^2 + o_n(1) \le 2\delta^2,
	\end{align*}
where the $c_T$ are the coefficients in the tree decomposition of $P$, and in the second line uses that each of the quantities above are constrained in $\calC_{\LSH}$.
In the final line we used \pref{prop:final-approx} and the fact that $c_{\sla} = o_n(1)$ while $\delta = \Omega(1)$.

	Thus, we obtain
	\[\frac 1n \pEB{\left\langle v, \frac{\vAMP}{\frac{1}{\sqrt{n}}\|\vAMP\|}\right\rangle}
		\ge\frac{2 \pm \frac{\eta}{8} - \frac{\eta}{8} - 6\delta^2 - \sqrt{\eps}\cdot O(\log n)^{\deg(\delta,t)} - 3\delta^2}{2\sqrt{1\pm \frac{\eta}{48}}}
		\ge 1 - \sqrt{\eps} \cdot O(\log n)^{\deg(\delta,t)} - 6\delta^2 - \frac{\eta}{4}\]
	which completes the argument similarly to~\pref{lem:xpp-vec-close}.
\end{proof}

\begin{proof}[Proof of Corollary \ref{cor:apx-lipschitz}]
	This follows immediately by taking $\delta = \sqrt{\eta}$ and applying the proof of~\pref{thm:ramp-recovery-1} to the result of~\pref{lem:corr-non-poly-amp}.\qedhere

\end{proof}

We finish by giving proof sketches for how to apply \pref{cor:apx-lipschitz} to the nnPCA and SK problems.

\begin{proof}[Proof Sketch of Corollary~\ref{cor:robust-nnpca}]
	We use the AMP iteration defined for~\cite[Lemma A.3]{MR15}: that is,

	\[x^{s+1} = X(x^s)_+ - \tfrac{\|(x^s)_+\|_0}{n}(x^{s-1})_+.\]
	These denoisers satisfy the conditions of \pref{cor:apx-lipschitz} almost immediately:
	\begin{itemize}
		\item The ReLu function $x_+$ is $1$-Lipschitz
		\item The partial derivatives of it are indicators of $x_i \ge 0$
		\item The covariance matrix $Q^t$ is \emph{diagonal}, and since each $f^t$ is the same function we may rescale to satisfy the requirements.
		      It may be of independent interest to note that in the case of $f^t$ depending only on the previous iteration (not everything that has occurred so far) the proof of \pref{prop:final-approx} can be significantly simplified.
	\end{itemize}
	Hence, we can apply \pref{cor:apx-lipschitz}.
	Let $t$ be the number of iterations that AMP requires to achieve objective value $\le \delta/100$.
We instantiate the \LSH at degree $d = (\frac{1}{\delta^2})^{O(4^t)}$ as guaranteed by \pref{cor:apx-lipschitz} to obtain error $\beta := \eps^{1/2}(d\log n)^{O(d)} + \delta^2/1000$.
We then have that our output vector $v$ achieves $\iprod{v,\vamp} \ge 1-\beta$.
From results of \cite{MR15}, it follows that 
\begin{align*}
v^\top X v 
&\ge (1-\beta)^2\vamp^\top X \vamp - 2\|X\|_{\op}(\|v-\vamp\|\|v\| + \|v-\vamp\|^2) \\
&\ge (1-2\beta)\optamp - \delta/100 - 5(\sqrt{2\beta} + 2\beta) \ge (1-\eps^{1/4}(d\log n)^{O(d)} - \delta)\optamp,
\end{align*}
as desired.
\end{proof}

\begin{proof}[Proof Sketch of Corollary~\ref{cor:robust-sk}]
	The SK AMP algorithm uses a slightly different framework of AMP known as \emph{Incremental AMP} (IAMP) which does consider all prior iterations.
	The interested reader can look at Section 2 of~\cite{Mon21}.

	Although this scenario is a bit different from nnPCA, it still satisfies our requirements:
	\begin{itemize}
		\item The $f^t$ functions used are combinations of linear functions and $\tanh(x)$ and have bounded Lipschitz constant.
		\item Using the above description, the partial derivatives are pseudo-Lipschitz.
		\item \cite[Lemma 2.2]{Mon21} still shows that $Q^t$ is diagonal, and we can rescale to change the diagonal.
	\end{itemize}

	The last snag is that the algorithm begins with $x^0 = N(0,I)$ instead of our fixed starting point $\vec{1}$.
	However, note that under our AMP iteration, $x^1 = N(0, I)$.
	Therefore, we may instead run AMP for $t + 1$ iterations to simulate a random starting point.
	Although now $x^1$ is correlated with $X$, this dependence is very mild and does not influence the algorithm. 

Now as in the proof of \pref{cor:robust-nnpca}, we choose $t$ to be is the number of iterations that AMP requires to achieve objective value $\le \delta/100$, and $d=(\frac{1}{\delta^2})^{O(4^t)}$ so that the error of the \LSH is $\le \delta^2/1000 + \sqrt{\eps}(d\log n)^{O(d)}$. 
The rest of the analysis is effectively identical.
\end{proof}

%% file: appendix-arxiv.tex
\section{Concentration of low-degree symmetric polynomials}\label{app:hyp}
This \pref{app:hyp} is dedicated to proving that various functions of tree polynomials concentrate, in order to ultimately prove \pref{lem:most-samples-reasonable}. 
Throughout, we will assume that $X \sim \calP_n(\calD,\calK)$, where $\calD$ is a symmetric $\frac{K}{\sqrt{n}}$-subgaussian distribution with $\E_{Z \sim \calD} Z^2 = \frac{1}{n}$ and $K= O_n(1)$.

\begin{lemma}[Expectations of tree polynomial coordinates]
\label{lem:expect-tree-coord}
If $T$ is a tree with $d$ edges for $d = O_n(1)$, then $0\le \E[T(X)_i] \le (Kd)^{d}$ for all $i$.
\end{lemma}

\begin{proof}
Given a $d$-edge tree $T$, let $v_1,\ldots,v_{d+1}$ be its vertices where $v_1$ is the root vertex, and let $E$ denote its edge set. 
Then, we may write the $i$th entry of the vector-valued polynomial
\[
T(X)_i = \sum_{\substack{v_1 = i\\v_2\ldots,v_{d+1}\in [n]}}\prod_{e\in E}X_{e}\]
where the sum is over all possible labelings of the vertices (including collisions).

Then, by linearity of expectation it follows that
\[\E[T(X)_i] = \sum_{\substack{v_1 = i\\v_2\ldots,v_{d+1}\in [n]}}\E\left[\prod_{e\in E}X_{e}\right].\]
Immediately we have that each of the component terms has expectation at least $0$ which gives the lower bound. For the upper bound, notice that the only terms with nonzero expectation are those where each edge $e$ occurs with even multiplicity.
Furthermore, if each edge occurs with even multiplicity and there are $q$ unique edges, the subgaussianity of $\calD$ implies that $\E[\prod_{e\in E}X_e] \le \frac{(K^2d)^{d/2}}{n^{q}}$.

It remains to count the number of labelings which result in an even-edge-multiplicity graph.
Each graph induced by the labeling must be connected, since $T$ is connected.
If the labeled graph has $q$ distinct edges (each of multiplicity at least two),
Hence, it follows that the vertices have at most $\frac q + 1$ unique labels (one of which is the special $v_1 = i$ vertex). 
Further, $q \le d/2$ always.
Letting $N_q$ be the number of even-edge-multiplicity multigraphs with $q$ distinct edges that can result from identifying vertices of $T$,
\[
\E[T(X)_i] 
= \sum_{q=1}^{d/2} N_q \cdot \binom{n}{q} \cdot \frac{(K^2d)^{d/2}}{n^q}
\le (K^2d)^{d/2}\cdot \sum_{q=1}^{d/2} N_q
\]
Since there are at most $d^{d/2}$ even-edge-multiplicity graphs that can result from labeling $T$ (since this is a bound on the number of matchings on $d$ edges), the conclusion follows.
\end{proof}

\begin{lemma}[Trunks are almost uncorrelated]
\label{lem:expect-trunk-prod}
	If $f(X) = \prod_{i=1}^{k}\trunk_{T^i}(X))$ is a product of non-empty trunks of degree $d$, then
\[
0\,\le\, \E[f(X)] - \prod_{i=1}^{k}\E[\trunk_{T^i}(X)] \,\le\, \frac{(K^2d)^{d/2}}{n}
\]
\end{lemma}

\begin{proof}
By definition, for a tree $T$ with vertices $v_1,\ldots,v_{d+1}$,
\[
\trunk_{T}(X) = \frac{1}{n}\iprod{T(X),\vec{1}} = \frac{1}{n}\sum_{v_1,\ldots,v_{d+1}\in[n]} \prod_{e \in E(T)} X_e.
\]
Hence the trunk associated with each tree $T^i$ corresponds to a sum over labelings of all vertices of $T^i$ (including the root), normalized by a factor of $\frac 1n$. 
By extension, a product of trunks can be associated to a collection of unrooted trees, of total degree $d$.
	
	For a collection of ordered multisets $V = [V^1,V^2,\ldots,V^k]$ with each $V^i$ taking elements from $[n]$ define $f_{V}(X)$ as the result of labeling the vertices in trunk $i$ with the the assignment $V^i$. 
Then, we can write 
	\[\E[f(X)] - \prod_{i=1}^{k}\E[\trunk_{T^i}(X)] = \sum_{V^1,V^2,\ldots,V^k} \E[f_V(X)] - \prod_{i=1}^{k}\E[\trunk_{T^i, V^i}(X)].\] 

Each term in the sum is non-negative: the symmetry of $\calD$ implies that if $\prod_{i=1}^k \E[\trunk_{T^i,V^i}(X)] > 0$ it is because each edge has even multiplicity, implying that $\E[f_V(X)] > 0$ as well; since the total multiplicity of each edge in $f_V(X)$ is the sum of its multiplicities in $\trunk_{T^i,V^i}(X)$, $\E[f_V(X)] \ge \prod_{i=1}^k \E[\trunk_{T^i,V^i}(X)]$ (by iterated application of H\"{o}lder's inequality).
For a collection $V$ to contribute positively to this sum, the following conditions must hold:
\begin{itemize}
	\item $\bigcup_{i=1}^{k}V^i$ must make every edge occur with even multiplicity. This condition is required because by~\pref{lem:expect-tree-coord}, each expectation is nonnegative.
	\item At least one edge must be shared between a pair of distinct $T^i, T^j$; otherwise, independence dictates that we can split $\E[f_V(X)]$ into a product over its constituent trunks.
 Therefore at least one vertex must be shared between some of the $V^i$: that is, $\left|\bigcup_{i=1}^{k}V^i\right| < \sum_{i=1}^{k}|V^i|$. 
\end{itemize}

Since there are a total of $d$ edges, the graph corresponding to $f_V(X)$ can have at most $\frac d2$ unique edges by the first condition. 
As the resulting graph has at most $k - 1$ components (as two components must be merged by the second condition), there can be at most $\frac d2 + k - 1$ vertices in $f_V(X)$.

Now, note that each positive term satisfies $\E[f_V(X)] = \frac 1{n^k}\cdot \frac {(K^2d)^{d/2}}{n^{d/2}} = \frac {(K^2d)^{d/2}}{n^{d/2+k}}$: the $\frac 1{n^k}$ comes from the $\frac{1}{n}$ normalization of each trunk, and the $((K^2d)/n)^{d/2}$ from the subgaussianity of $\calD$. 
Since there are at most $O(n^{d/2 + k - 1})$ positively contributing terms (there are $\binom {n}{d/2+k-1}$ ways to choose these vertices in the union of $V^i$ and at most $d!!$ ways to match the edges to form an even-edge-multiplicity graph), it follows that all positive terms together contribute $\frac{(K^2d)^{d/2}}{n}$, completing the proof.
\end{proof}

\begin{lemma}[Concentration of trunk polynomials]
\label{lem:conc-trunk-poly}
Let $T(X)$ be a tree polynomial of degree $d$ and let $\trunk_T(X) = \frac 1n \langle T(X), 1\rangle$ be the associated trunk.
Then for any $\kappa > 0$, 
\[
\Pr\left(\left|\trunk_T(X) - \E[\trunk_T(X)]\right| \le \kappa\right) \ge 1 - \frac {(K^2d)^{d/2}}{n\kappa^2}.
\]
\end{lemma}

\begin{proof}
By Chebyshev's Inequality, it follows that
\[\Pr[|\trunk_T(X) - \E[\trunk_T(X)]| \ge \kappa] \le \frac{\Var(\trunk_T(X))}{\kappa^2}.\]
Therefore, it suffices to show that $\Var(\trunk_T(X)) = \frac{(K^2d)^{d/2}}{n}$. 
Applying \pref{lem:expect-trunk-prod},
\[
\Var(\trunk_T(X)) = \E[\trunk_T(X)^2] - \E[\trunk_T(X)]^2 \le \frac{(K^2d)^{d/2}}{n},
\]
and the conclusion follows.
\end{proof}

We will need to show that the coordinates of $T(X)$ remain bounded by $\polylog(n)$ with high probability. 
For this, we will need the following theorem regarding the tail behavior of low-degree polynomials in subgaussian random variables:
\begin{theorem}[{{\cite[Theorem 1.2]{gotze2021concentration}}}: Polynomials of Subgaussian Random Variables Concentrate]
\label{thm:poly-sub-conc}
	Suppose that $Z_1, Z_2, \ldots, Z_m$ are independent $K=O(1)$-subgaussian random variables and $f: \R^m\rightarrow \R$ a polynomial of total degree $d\in \N$. Then, for all $t > 0$
	\[\prob{|f(Z) - \E[f(Z)]| \ge t} \le 2\exp\left(-\frac{1}{CK^2}\min_{1\le r\le d}\left(\frac{t^2}{\|\E[f^{(r)}(Z)]\|_F^2}\right)^{1/r}\right)\]
	where $C > 0$ is an absolute constant dependent only on $d$ (that is, not on $m$), $f^{(r)}$ is the tensor of order-$r$ partial derivatives of $f$, and $\|\cdot\|_F$ denotes the Hilbert-Schmidt norm.
\end{theorem}

To apply the theorem, we will need control of the partial derivative tensors of tree polynomials.
\begin{lemma}[Partial derivatives of $T$ have reasonable expectation]
\label{lem:partial-derivative-expectation}
Suppose $T$ is a tree polynomial with $d$ edges, and let $Z = \sqrt{n}\cdot X$.
For $r \le d$, define $T^{(r)}(Z)_i$ to be the tensor of order-$r$ partial derivatives of $T(Z)_i$.
Then
\[\left\|\E_{X}[T^{(r)}(Z)_i]\right\|_F^2 \le (K^2d+1)^{5d+2} n^d\]
\end{lemma}

\begin{proof}
As before, if $v_1,\ldots,v_{d+1}$ are the vertices of $T$, then we have $T(Z)_i = \sum_{v_2,\ldots,v_{d+1} \in [n]} \prod_{e \in T} Z_e$ with $v_1 = i$.
Each choice of labels for $v_2,\ldots,v_{d+1}$ in $[n]$ induces a labeled graph which is isomorphic to $T$ if and only if all labels are distinct. 
By linearity of the derivative we can split the sum according to the topology of the graphs: that is, write
\[\E[T^{(r)}(Z)_i] = \sum_{G}c_G \E[G^{(r)}(Z)]\]
where $c_G \in \N$ is the number of ways to produce a graph isomorphic to $G$ by identifying vertices of $T$. 
Note that $\sum_G c_G \le (d+1)^{d+1}$, since this is an upper bound on the number of partitions of $d+1$ vertices into at most $d+1$ sets.

Now, we analyze each $G^{(r)}$ term separately. 
Taking a partial derivative with respect to variables $Z_{i_1,j_1},\ldots,Z_{i_r,j_r}$ has the effect of removing the corresponding labeled edges from the graph.
The location of these edges within $G$ dependings on the labeling of $G$'s vertices.
We will further partition the sum over labelings of $G^{(r)}$ according to the subgraph $H$ defined by the terms with respect to which the derivative is being taken:
\[
\left\|\E[G^{(r)}(Z)]\right\|_F^2 \le \sum_{H\subseteq G} \sum_{V\rightsquigarrow H} d^{r}\E[(G\setminus H)_V(Z)]^2
\]
where $V \rightsquigarrow H$ denotes choosing labels $V \in [n]$ for each vertex of $H$, and $(G\setminus H)_V$ denotes the sum over products of edges given removing all edges in $H$ from $G$ and fixing the corresponding vertices of $V$ in the result. 
The $d^r$ accounts for taking $r$ derivatives, each yielding a factor of degree at most $d$.

Now, we bound this expectation as a function of properties of $H$ and $G$.
Suppose $H$ consists of $s$ new vertices (excluding the root, if it is present in $H$).
There are thus $n^s$ assignments $V\rightsquigarrow H$. 
Then, we claim that 
\begin{align}
\E[(G\setminus H)_V(Z)]^2 \le (K^2(d-s))^{d-s} \cdot n^{d - s}.\label{eq:bound-per}
\end{align}
Since there are at most $\binom dr\le d^r$ subgraphs $H\subseteq G$ (choose $r$ of the $d$ edges), this implies that 
\[\left\|\E[G^{(r)}(Z)]\right\|_F^2 \le d^{2r}\cdot n^s\cdot (K^2(d-s))^{d-s} n^{d-s}\]
and thus
\begin{align*}
\left\|\E[T^{(r)}(Z)_i]\right\|_F^2 
= \left\|\sum_{G} c_G \E[G^{(r)}(Z)]\right\|_F^2
&\le \left(\sum_G c_G \left\|\E[G^{(r)}(Z)]\right\|_F\right)^2\\
&\le \left(\sum_G c_G \cdot d^{r} \cdot O(d)^{d/2} n^{d/2}\right)^2
\le (d+1)^{2d+2} d^{2r} \cdot (K^2d)^d n^d.
\end{align*}
as desired.

So, it suffices to prove~\pref{eq:bound-per}. To do so, first note that if $G\setminus H$ has any edge of odd multiplicity, then this expectation is just $0$. 
So, we only have to look at the case when $G\setminus H$ has all edges of even multiplicity, of which there must be at most $\frac {d-r}2$ distinct edges. 
Let us define $k$ to be the number of components in $G\setminus H$ and $t$ to be the number of fixed vertices in $G\setminus H$.
 Then, the number of unfixed vertices (corresponding to the number of assignments summed in $\E[(G\setminus H)_V(Z)]$) is at most $\frac{d-r}{2} + k - t$. 
This implies that 
\[\E[(G\setminus H)_V(Z)]^2 \le O(d-r)^{d-r} \cdot n^{d-r + 2k - 2t}.\]

By design, $t = s + 1$: the root $i$ and the $s$ new introduced vertices. 
Furthermore, $k\le r + 1$ as each edge removed can introduce at most one new connected component. 
In fact, $k \le t=s+1$: each connected component contains at least one fixed vertex, either the root or the vertex incident on the edge that was removed to disconnect it; clearly, each fixed vertex belongs to only one component.
Taking these facts together,
\[
d - r + 2k - 2t 
= d + (k-r) + (k-t) - t 
\le d + 1 + 0 - (s + 1) 
= d - s
\]
which completes the proof.\end{proof}

\begin{lemma}\label{lem:tree-inf-bd}
Suppose that $T$ is a tree with $d$ edges. 
Then, there exists a constant $C_K > 0$ depending only on the subgaussian parameter $K=O(1)$ of the $\frac{K}{\sqrt{n}}$-subgaussian distribution $\calD$ such that with probability at least $1 - \frac{2}{n}$ over the choice of $X$,
\[
\frac 1n\|T(X)\|_{4}^4 \le \|T(X)\|_{\infty}^4 \le (C_K d^6\log n)^{2d}.\]
\end{lemma}

Note that the bound on the $4$-norm is loose; the 4-norm is actually $O(d^{2d})$ (for example, by instead applying Lemmas~\ref{lem:conc-trunk-poly} and~\ref{lem:expect-tree-coord} to the trunk of $T(X)^{\circ 4}$).
However, since we will be forced to suffer the logarithmic loss in any case, using the infinity norm bound is sufficient for our purposes.

\begin{proof}
We will apply \pref{thm:poly-sub-conc} to the entries of $T(X)$.
Define $Z = \sqrt n X$: that is, $Z_{ij} = \sqrt n \cdot X_{ij}$ and thus satisfies $\E[Z_{ij}^2] = 1$. 
Since $T$ is a homogeneous polynomial, $T(X)_i = \frac 1{n^{d/2}} T(Z)_i$. 
With this in mind, let us apply~\pref{thm:poly-sub-conc} to $f(Z) = \frac 1{n^{d/2}} T(Z)_i$ and $t = (2CK^2 (K^2d+1)^6\cdot \log n)^{d/2}$.
 As stated, $Z$ is $K = O(1)$-subgaussian and thus by~\pref{lem:partial-derivative-expectation} we have that
\[
\left\|\E_{Z}[f^{(r)}(Z)]\right\|_F^2 \le (K^2d + 1)^{5d+2} \le (K^2d + 1)^{6d},
\]
as the statement is trivial if $d = 1$.
So from \pref{thm:poly-sub-conc}, 
\begin{align*}
\prob{|f(Z) - \E[f(Z)]| \ge t} 
&\le 2\exp\left(-\frac{1}{CK^2}\min_{1\le r\le d}\left(\frac{t^2}{(K^2d+1)^{6d}}\right)^{1/r}\right) \\
&\le 2\exp\left(-\frac{1}{CK^2}\left(\frac{2CK^2 (K^2d+1)^6\cdot \log n}{(K^2d+2)^6}\right)\right) 
= \frac 2{n^2}.
\end{align*}
We may also bound $0 \le \E[f(Z)] \le (Kd)^{d}$ by~\pref{lem:expect-tree-coord}. Thus, it follows that
\[
\|T(X)_i\|^4 
\le ((K^2d)^{d} + (2CK^2 (K^2d+1)^6\cdot \log n)^{d/2})^4 
= (C_K \cdot d^6\cdot \log n)^{2d}
\]
with probability at least $1 - \frac 2{n^2}$. Union bounding over all $n$ coordinates implies then that $\|T(X)\|_{\infty}^4 \le (C_K\cdot d^6\cdot \log n)^{2d}$ with probability at least $1 - O(\frac 1n)$ as desired.
The fact that $\|T(X)\|_4^4 \le n\|T(X)\|_\infty^4$ finishes the proof.	
\end{proof}

\restatelemma{lem:most-samples-reasonable}
\begin{proof}
For $R$ a sufficiently large constant to be chosen later, set $\kappa = \frac{c_{\sla}}{R} = \omega(1/\sqrt{n})$.
Let us prove that each constraint in $\calC_{\LSH}$ holds when we substitute $X = \Xprox$ with with high probability.
\begin{itemize}
	\item First, we wish to show that $\frac 1n \langle T_1(X), T_2(X)\rangle = \E[\frac 1n \langle T_1(X'), T_2(X')\rangle] \pm c_\sla$ where each of $T_1, T_2$ are lumber. To do so, write $T_1 = S_1(X)\prod_{i=1}^{\alpha}\trunk_{A_i}(X)$ and $T_2 = S_2(X)\prod_{i=1}^{\beta}\trunk_{B_i}(X)$ where $S_1, S_2$ are trees and $\trunk_{A_i},\trunk_{B_i}$ are trunks.

Then, define $T(X) = S_1(X) \circ S_2(X)$ and note that
\[\frac 1n \langle T_1(X), T_2(X)\rangle = \prod_{i=1}^{\alpha}\trunk_{A_i}(X)\prod_{i=1}^{\beta}\trunk_{B_i}(X) \cdot \frac 1n \langle S_1(X),S_2(X)\rangle = \prod_{i=1}^{\alpha + \beta + 1}\trunk_{C_i}(X)\]
where $\{C_1,\ldots,C_\alpha\} = \{A_1,\ldots,A_\alpha\}$, $\{C_{\alpha + 1}, \ldots, C_{\alpha + \beta}\} = \{B_1,\ldots,B_{\beta}\}$, and $C_{\alpha + \beta + 1} = T$. In other words, we may decompose an inner product of two lumber of degree at most $d$ each as a product of at most $2d$ non-empty trunks.

By~\pref{lem:conc-trunk-poly}, we know that for each trunk $\trunk_{C_i}$, with probability at least $1 - O(\frac 1{n\kappa^2})$, $\trunk_{C_i}(X) = \E[\trunk_{C_i}(X)] \pm \kappa$. 
Furthermore, by~\pref{lem:expect-tree-coord} and symmetry, $\E[\trunk_{C_i}(X)] \le (K\deg(C_i))^{\deg(C_i)}$.

From~\pref{lem:expect-trunk-prod}, we also know that
\[
\left|\E\left[\frac 1n \langle T_1(X),T_2(X)\rangle\right] - \prod_{i=1}^{\alpha + \beta + 1}\E[\trunk_{C_i}(X)]\right| \le \frac{(K^22d)^{d}}{n}
\]

Therefore, all we have left to show is that $\frac 1n \langle T_1(X), T_2(X)\rangle = \prod_i \trunk_{C_i}(X)$ concentrates around this expectation as well.

Thus, condition on the event that $\trunk_{T}(X) = \E[\trunk_{T}(X)]\pm \kappa$ for \emph{every} trunk of degree at most $d$: there are at most $2^{O(d\log d)}$ such trunks so we can simply do this via a union bound. Then, conditioned on this event, we have that
\[
\frac 1n \langle T_1(X), T_2(X)\rangle 
= \prod_{i=1}^{\alpha + \beta + 1}\trunk_{C_i}(X)
= \prod_{i=1}^{\alpha + \beta + 1}(\E[\trunk_{C_i}(X)] \pm \kappa) 
\le \prod_{i=1}^{\alpha + \beta + 1} \E[\trunk_{C_i}(X)] \pm \kappa \cdot 2^{2d} \cdot (K2d)^{2d}
\]
where in the last step we use that there are at most $2^{2d}$ subsets of the $\alpha + \beta + 1 \le 2d$ terms, and that the product of the expectations of all terms is bounded by $(K2d)^{2d}$ by H\"{o}lder's inequality and \pref{lem:expect-tree-coord}.
Putting these together, 
\[\frac 1n \langle T_1(X), T_2(X)\rangle = \E\left[\frac 1n \langle T_1(X),T_2(X)\rangle\right] \pm \frac{(K^2 2d)^{d}}{n} \pm \kappa (4Kd)^{2d},\]
completing this proof as the error on the right-hand side is less that $c_{\sla}$ so long as $R$ is chosen larger than $(4Kd)^{2d}$.
	
	\item Next, we wish to show that $\frac 1n\langle T(X), \vAMP\rangle = \E_{(Z,v^\ast)}[\frac 1n\langle T(Z), v^\ast\rangle] \pm c_\sla$ and $\frac 1n \|\vAMP\|_2^2 = 1\pm \frac \eta {48}$. 

If the denoisers are Lipschitz functions, then we may apply Theorem 1 of \cite{JM13} to a well-chosen generalized AMP which produces all tree $T(X)$ in parallel with $\vAMP$ to conclude concentration of $\langle \vAMP, T(X)\rangle$ so long as $c_\sla = \frac{\eta}{s(n)}$ for $\eta = \Omega(1)$ and $s(n)$ a slowly-enough growing function of $n$. 
To get concentration for lumber $T(X) = T^\ast (X)\cdot \trunk_{T_1}(X)\cdots \trunk_{T_k}(X)$, note that $\langle \vAMP, T(X)\rangle =  \langle \vAMP, T^\ast(X)\rangle\cdot \trunk_{T_1}(X)\cdots \trunk_{T_k}(X)$ and since this is a finite product of converging random variables we must achieve the same guarantees.

In the polynomial denoiser case, we prove stronger concentration from scratch.
First, if $\vAMP$ is a polynomial, expand $\vAMP = P(X) = \sum_{T'\in \calL}c_{T'}\cdot T'(X)$.
By arguments identical to the one regarding $\iprod{T_1(X),T_2(X)}$ above:
\begin{align*}
	\frac 1n\langle T(X), \vAMP\rangle 
	&= \sum_{T'\in \calL} c_{T'}\cdot \frac 1n\langle T(X), T'(X)\rangle\\
	&= \E_{(Z, v^\ast)}\left[\frac 1n\langle T(Z), v^\ast\rangle \right]\pm \kappa (4Kd)^{2d} \cdot \sum_{T'\in \calT}|c_{T'}|\\
	 &= \E_{(Z, v^\ast)}\left[\frac 1n\langle T(Z), v^\ast\rangle \right] \pm \kappa (4Kd)^{2d}\cdot O(N_\calL)
\end{align*}
as desired (using that $\sum_{T'\in \calT}|c_{T'}| \le O(N_\calL)$ with the constant in the $O(\cdot)$ term just depending on the constants in the AMP polynomials). 
As long as $R$ is chosen as a sufficiently large constant, this is at most $c_{\sla}$.
As this is just a consequence of the concentration of inner products of lumber, so it must also hold with high probability. 
Similarly,
\begin{align*}\frac 1n \|\vAMP\|_2^2 &= \sum_{T'\in \calL} c_{T'}\cdot \frac 1n\langle T'(X), \vAMP\rangle\\
&= \E_{(Z,v^\ast)}\left[\frac 1n \langle v^\ast, v^\ast\rangle\right] \pm \kappa (4Kd)^{2d}\cdot O(N_\calL)\sum_{T'\in \calT} |c_{T'}|\\
&= 1\pm \kappa (4Kd)^{2d}\cdot O(N_\calL^2)\\
&= 1\pm \frac{\eta}{48}.\end{align*}
using that $\frac 1n \E[\|v^\ast\|_2^2] = 1$ and $R$ was chosen a sufficiently large constant.

\item Next is the bounded maximums property: that is, $\|T(X)\|_{\infty}^4 \le (C_K\cdot d^8\cdot \log n)^{2d}$. 
By \pref{lem:tree-inf-bd} this also holds with high probability.

\item Finally, we are left with showing that $\|X\|_{\op}^2 \le 5$. However, this follows as it is well known that $\|X\|_{\op} \le 2.01$ with probability even larger than $1-1/n$ (\cite{anderson2010introduction}).
\end{itemize}

Therefore, it must be the case that $X$ is reasonable with high probability, as desired.\qedhere
\end{proof}

%% file: appendix-approx-arxiv.tex
\section{Polynomial approximation of AMP denoisers}\label{app:amp-poly}
In this appendix, we show that nice AMP denoisers can be approximated by low-degree polynomials.
Results of this type are known in the literature; we prove a variant that suits our needs.
We first require a useful fact from multivariate approximation theory.

\begin{lemma}[Weighted approximation with Gaussian weights]
\label{lem:weighted-gaussian}
	Fix $\eta > 0$. Suppose $f: \R^t\rightarrow \R$ is a $L$-Lipschitz function and suppose $U\sim \calN(0, \Sigma)$ with $\Sigma \preceq 3tI$.
Then, there exists a polynomial $p$ of total degree at most $O\bigl(4^t\left(\frac{L}{\eta}\right)^8\poly(t,\log\frac{1}{\eta})\bigr)$ such that
	\[\E_U[(f(U) - p(U))^2] \le \eta^2.\]
\end{lemma}

\begin{proof}
	Let $\delta = \eta \cdot 2^{-\frac t4}$. We use~\cite[Theorem 1]{musin2017weighted}, applied to $\Phi(x) = \exp(\frac 1{24t}\|x\|^2)$. Although the result as given does not give quantitative degree bounds, we can substitute $\nu = 48\sqrt{t\log \bigl(\frac {3L}\delta\bigr)}$, $\lambda = O\bigl(\bigl(\frac{L}{\delta}\bigr)^2\bigr)$, and let the degree bound per variable (referenced as $n$ in the paper) be $d = \widetilde{O}\bigl(t\left(\frac{tL}{\delta}\right)^8\bigr)$ (the $\widetilde{O}$ hides $\log \frac 1\delta$ factors and a constant term $K_H$ independent of $t, L, \delta$). With these guarantees, we find a polynomial $p$ such that
	\[\sup_{x\in \R^t}\frac{|f(x) - p(x)|}{\Phi(x)} \le \delta.\]
	To get an expectation bound, note that $\Phi(x) \le \exp(\frac 18x^\top \Sigma^{-1} x)$. Hence, we have that
	\[\E_{U}[(f(U) - p(U))^2] \le \delta^2 \E_U[\exp(\tfrac 14U^\top \Sigma^{-1}U)] = \delta^2 \int_{x\in \R^t} \frac{1}{\sqrt{(2\pi)^t \cdot \det \Sigma}}\exp\left(-\frac 14x^\top \Sigma^{-1}x\right)\,\mathrm dx.\]
	Letting $Z = 2\Sigma$ with $\det Z = 2^t \det\Sigma$ and substituting, we find that
	\[\int_{x\in \R^t} \frac{1}{\sqrt{(2\pi)^t \cdot \det \Sigma}}\exp\left(-\frac 14x^\top \Sigma^{-1}x\right)\,\mathrm dx = 2^{\frac t2}\int_{x\in \R^t}\frac{1}{\sqrt{(2\pi)^t \cdot \det Z}}\exp\left(-\frac 12x^\top Z^{-1}x\right)\,\mathrm dx = 2^{\frac t2}.\]
	From this, we have $\E_U[(f(U) - p(U))^2] \le \delta^2 \cdot 2^{\frac t2} = \eta$ as desired.
\end{proof}

\begin{lemma}[Variant of Stein's Lemma]
\label{lem:stein-extend}
Suppose $X\sim \calN(0, \Sigma)$ with $\Sigma\in \R^{d\times d}$ and $f: \R^d \rightarrow \R$ is weakly differentiable at coordinate $i$. Then,
\[\E_X\biggl[\frac{\partial f}{\partial x_i}\biggr|_{x\rightarrow X}\biggr] = \E[\langle \Sigma^{-1}_i,X\rangle f(X)].\]
\end{lemma}

\begin{proof}
	Let $p_X(x)$ be the probability density of $X$. Then, note that $\frac{\partial p_X}{\partial x_i} = -\langle \Sigma^{-1}_i,x\rangle p_X(x)$. With this in mind, we apply integration by parts on coordinate $i$ (noting that the evaluation of $f(x)p_X(x)$ at the limits of integration is $0$):
	\[\int_{\R^d}\frac{\partial f(x)}{\partial x_i}p_X(x)\,\mathrm dx = \int_{\R^d}\langle \Sigma_i^{-1},x\rangle f(x)p_X(x)\,\mathrm dx = \E[\langle \Sigma^{-1}_i,X\rangle f(X)].\qedhere\]
\end{proof}

\begin{lemma}[Compression of iterates]
\label{lem:iterate-comp}
	Suppose that $\{p^t: \R^t \rightarrow \R\}$ is a sequence of polynomial denoisers whose degree does not depend on $n$, and $\hat{x}^0,\hat{x}^1,\ldots$ are the AMP iterates produced by these denoisers on the input matrix $X$.
Then, for any pseudo-Lipschitz function $\psi: \R^{k + 2} \rightarrow \R$,
	\[\operatorname*{p-lim}_{n\rightarrow\infty}\frac 1n\sum_{i=1}^{n}\psi(\hat{x}_i^0,\hat{x}_i^1,\ldots,\hat{x}_i^k;y_i) = \E[\psi(U^0,U^1,\ldots,U^k;Y)]\]
	for a centered Gaussian Process $U$ with covariance matrix $Q$ given by \[Q_{j + 1,k+1} = \E[p^j(U^0,\ldots,U^j;Y)p^k(U^0,\ldots,U^k;Y)].\]
\end{lemma}

\begin{proof}
	This follows immediately by~\cite[Proposition 2.1]{Mon21}, substituting~\cite[Theorem 4]{bayati2015universality} for the usage of~\cite[Theorem 1]{javanmard2013state}.
\end{proof}

To prove the next lemma, we make use of some special assumptions on $Q$, the covariance matrix of the (not-necessarily polynomial) denoisers $f^t$ when applied to the AMP iterates.

\begin{lemma}[Approximating AMP with polynomials]
\label{lem:approx-poly}
	Fix $\delta > 0$ and $T\in \N$. Suppose that:
	\begin{itemize}
	\item $\{f^t: \R^t \rightarrow \R\}$ is a sequence of $L$-Lipschitz denoiser functions which produce AMP iterates $x^0, x^1,x^2,\ldots$,
	\item For each $i \in [n]$, $\frac{\partial f^t}{\partial x_i}$ is either pseudo-Lipschitz or an indicator,
	\item The covariance matrix $Q^t$ corresponding to $f^1,\ldots,f^t$
	 satisfies $Q\succeq I$ and $\max_{i,j}|Q_{i,j}| \le 2$.
	\end{itemize}
Then, there exists a sequence of polynomial denoisers $\{p^t: \R^t\rightarrow \R\}$ producing AMP iterates $\hat{x}^{t+1} = Xp^t(\hat{x}^0,\hat{x}^1,\ldots,,\hat{x}^t) - \sum_{j=1}^t \hat{b}_{t,j}p^{j-1}(\hat{x}^0,\hat{x}^1,\ldots,\hat{x}^{j-1})$ such that
	\[\operatorname*{p-lim}_{n\rightarrow\infty}\frac 1n\|x^t - \hat{x}^t\|^2 \le \delta^2\]
	for all $t\le T$.
	Furthermore, we can choose such polynomials with $\deg p^t = \widetilde{O}\biggl((256t^{9/2}L(K + L)/\delta)^{2^{T+5-t}}\biggr)$, where the $\sim$ hides polynomial factors in the logarithm of the argument to $\tilde{O}$.
\end{lemma}

\begin{proof}
	We will define the polynomials $p^1,p^2,\ldots$ inductively.
	Note that for $t = 0$, the result immediately holds (as $\hat{x}^0 = x^0$). Suppose we are now at some iterate $t + 1$, having defined polynomials $p^1,p^2,\ldots,p^t$.
As established in \pref{lem:iterate-comp}, the original denoisers $f^1,\ldots,f^t$ and the polynomial denoisers $p^1,\ldots,p^t$ produce two centered Gaussian Processes $U, \hat{U}$ with covariances $Q^t$ and $\hat{Q}^t$ satisfying
	\[Q_{i,j}^t = \E[f^{i-1}(U^0,\ldots,U^{i-1})f^{j-1}(U^0,\ldots,U^{j-1})]\]
	and \[\hat{Q}_{i,j}^t = \E[p^{i-1}(\hat{U}^0,\ldots,\hat{U}^{i-1})p^{j-1}(\hat{U}^0,\ldots,\hat{U}^{j-1})].\]
	We inductively prove four statements dependent on four parameters $g_1(t), g_2(t), g_3(t), g_4(t)$ to be chosen later:
	\begin{enumerate}
		\item $Q^t$ and $\hat{Q}^t$ are close in Frobenius norm: $\|Q^t - \hat{Q}^t\|_F \le g_1(t).$
		\item There exists a polynomial $p^{t}$ of total degree at most $\widetilde{O}\left(4^t\left(\frac L{g_2(t)}\right)^8\right)$ such that
		\[\E_{\hat{U}} [(f^t(\hat{U}^0,\ldots,\hat{U}^t) - p^t(\hat{U}^0,\ldots,\hat{U}^t))^2] \le g_2(t)^2.\]
		\item For all $1\le j\le t$, $|b_{t,j} - \hat{b}_{t,j}| \le g_3(t)$.
		\item $\operatorname*{p-lim}_{n\rightarrow\infty} \frac 1n\|x^{t+1} - \hat{x}^{t+1}\|_2^2 \le g_4(t)^2$.
	\end{enumerate}

	We prove each of these inductive hypotheses as a separate claim.
	\begin{claim}
	\label{clm:frob-bound}
		The first statement is true. That is, $\|Q^t - \hat{Q}^t\|_F \le g_1(t).$
	\end{claim}
	\begin{proof}
	To prove such a Frobenius bound, note that \[\|Q^t - \hat{Q}^t\|_F^2 \le \|Q^{t - 1} - \hat{Q}^{t-1}\|_F^2 + 2\sum_{j=0}^{t-1}(Q_{t,j+1} - \hat{Q}_{t,j+1})^2 \le g_1(t - 1)^2 + 2\sum_{j=0}^{t-1}(Q_{t,j+1} - \hat{Q}_{t,j+1})^2.\]

	Let's consider the contribution of each term in the latter sum separately. We begin by coupling $U$ and $\hat{U}$ from the definition of $Q$ and $\hat{Q}$. In particular, write $U = Q^{\frac 12}g$ and $\hat{U} = \hat{Q}^{\frac 12}g$ for $g\sim \calN(0, I)$. We can do this coupling since $Q_{t,j+1} - \hat{Q}_{t,j+1}$ does not have any cross terms in $U, \hat{U}$. By the Triangle Inequality, Jensen's Inequality, and Cauchy Schwarz, we have that
	\begin{align*}
	|Q_{t,j+1} - \hat{Q}_{t,j+1}| &\le |\E_{g}[f^{t-1}(U)f^j(U) - p^{t-1}(\hat{U})p^j(\hat{U})]|\\
	&= \E_{g}[|f^{t-1}(U)f^j(U) - f^{t-1}(U)p^j(\hat{U}) + f^{t-1}(U)p^j(\hat{U}) - p^{t-1}(\hat{U})p^j(\hat{U})|]\\
	&\le \E_{g}[|f^{t-1}(U)(f^j(U) - p^j(\hat{U}))|] + \E_{U,\hat{U}}[|p^j(\hat{U})(f^{t-1}(U) - p^{t-1}(\hat{U})|]\\
	&\le \sqrt{\E_U[f^{t-1}(U)^2]\E_{g}[(f^j(U) - p^j(\hat{U}))^2]} + \sqrt{\E_{\hat{U}}[p^j(\hat{U})^2]\E_{g}[(f^{t-1}(U) - p^{t-1}(\hat{U}))^2]}.
	\end{align*}
	To simplify this expression, we apply the Almost-Triangle Inequality upon squaring:
	\[(Q_{t,j+1} - \hat{Q}_{t,j+1})^2 \le 2\E_{g}[(f^{t-1}(U) - p^{t-1}(\hat{U}))^2]\left(\E_U[f^{t-1}(U)^2] + \E_{\hat{U}}[p^{t-1}(\hat{U})^2]\right)\]
	Here, we assume that $ \E_{\hat{U}}[p^j(\hat{U})^2]\le \E_U[p^{j}(U)^2]$ and similarly for the other two varieties of terms present.

	To deal with the first term, note that once more by the Almost Triangle Inequality, the duality of nuclear and 1-norm
	 and induction that
	\begin{align*}
		\E_{g}[(f^{t-1}(U) - p^{t-1}(\hat{U}))^2] &\le 2\E_{g}[(f^{t-1}(U) - f^{t-1}(\hat{U}))^2] + 2\E_{\hat{U}}[(f^{t-1}(\hat{U}) - p^{t-1}(\hat{U}))^2]\\
		&\le 2L^2 \E_{U,\hat{U}}[\|U - \hat{U}\|_2^2] + 2g_2(t-1)^2\\
		&=2L^2\E_g[\|(Q^\frac 12 - \hat{Q}^{\frac 12})g\|_2^2] + 2g_2(t-1)^2\\
		&\le 2L^2\|Q^{\frac 12} - \hat{Q}^{\frac 12}\|_{\mathsf{op}}^2\E_g[\|g\|_2^2] + 2g_2(t - 1)^2\\
		&\le 2tL^2\|Q - \hat{Q}\|_{\ast} + 2g_2(t - 1)^2\\
		&\le 2t^{3/2}L^2g_1(t - 1) + 2g_2(t-1)^2.
	\end{align*}

	The second term follows by a direct Lipschitz condition bound: $f^{t-1}(U)^2 \le L^2\|U\|_2^2$. Hence, \[\E_U[f^{t-1}(U)]^2 \le L^2\sum_{i=1}^{t}Q_{i,i} = L^2 \|Q^{t-1}\|_\ast.\]
	For the last term, we have by the Almost-Triangle Inequality that
	\begin{align*}
\E_{\hat{U}}[p^{t-1}(\hat{U})^2] &= \E_{\hat{U}}[(f^{t-1}(\hat{U})^2 + (p^{t-1}(\hat{U}) - f^{t-1}(\hat{U})))^2] \\
&\le 2\E_{\hat{U}}[f^{t-1}(\hat{U})^2] + 2\E_{\hat{U}}[(p^{t-1}(\hat{U}) - f^{t-1}(\hat{U}))^2] \\
&\le 2L^2 \|\hat{Q}^{t-1}\|_\ast+2g_2(t-1)^2.
	\end{align*}

Hence, putting everything together we obtain that
\begin{align*}(Q_{t,j+1} - \hat{Q}_{t,j+1})^2 &\le 2\left(2t^{3/2}L^2g_1(t - 1) + 2g_2(t-1)^2\right)\left(3L^2\|Q^{t-1}\|_\ast + 2g_2(t-1)^2\right)\\
&\le 8(t^{3/2}L^2g_1(t-1)+g_2(t-1)^2)(3t^2L^2 + g_2(t-1)^2)\\
&\le 16t^2L^2(t^{3/2}L^2g_1(t-1)+g_2(t-1)^2).\end{align*}

Therefore,
\[\|Q^t - \hat{Q}^t\|_F^2 \le g_1(t-1)^2+32t^3L^2(t^{3/2}L^2g_1(t-1)+g_2(t-1)^2)\le 64t^3L^2(t^{3/2}L^2g_1(t-1)+g_2(t-1)^2).\]
Choosing $g_1(t)^2$ to be the expression on the right hand side gives us the conclusion.
	\qedhere
	\end{proof}

	Next, we handle the polynomial approximation.
	\begin{claim}
		The second statement is true. That is, there exists a polynomial $p^{t}$ of total degree at most $\widetilde{O}\left(4^t\left(\frac L{g_2(t)}\right)^8\right)$ such that
		\[\E_{\hat{U}} [(f^t(\hat{U}^0,\ldots,\hat{U}^t) - p^t(\hat{U}^0,\ldots,\hat{U}^t))^2] \le g_2(t)^2.\]
	\end{claim}

	\begin{proof}
	Since $\|\hat{Q}^{t} - Q^t\|_F \le g_1(t)$, it follows that $\hat{Q}^t \le (2t + g_1(t))I \le 3tI$.
	From \pref{lem:weighted-gaussian} applied to $f^t$ and $\Sigma = \hat{Q}^{t}$, it follows that we have such a polynomial $p^t$ with total degree bound $\widetilde{O}\left(4^t\left(\frac L{g_2(t)}\right)^8\right)$.
	\end{proof}

	Using our derived polynomial, we show that the Onsager correction term is also close.
	\begin{claim}
		The third statement is true. That is, for all $1\le j\le t$, $|b_{t,j} - \hat{b}_{t,j}| \le g_3(t)$.
	\end{claim}
	\begin{proof}
	Recall that $b_{t,j} = \frac 1n\sum_{i=1}^{n}\frac{\partial f^t(x^0,x^1,\ldots,x^t)}{\partial x^j_i}$ and similarly for $\hat{b}_{t,j}$. Hence, by state evolution and \pref{lem:stein-extend} (using that the partial derivative is either an indicator or pseudo-Lipschitz) it follows that
	\[\operatorname*{p-lim}_{n\rightarrow\infty} b_{t,j} = \E_U \biggl[\frac{\partial f^t(x^0,\ldots,x^t)}{\partial x^j}\biggr|_{x\rightarrow U}\biggr] = E_U[\langle Q^{-1}_j,U\rangle f^t(U)].\]
	Similarly, we have \[\operatorname*{p-lim}_{n\rightarrow\infty} \hat{b}_{t,j} = \E_{\hat{U}}[\langle \hat{Q}^{-1}_j,\hat{U}\rangle p^t(\hat{U})].\]
	Now, similarly to showing closeness of $Q$ and $\hat{Q}$, couple $U = Q^{\frac 12}g, \hat{U} = \hat{Q}^{\frac 12}g$ and write
	\begin{align*}|\operatorname*{p-lim}_{n\rightarrow\infty} b_{t,j} - \hat{b}_{t,j}| &\le \bigl|\E_{g}\bigl[\langle Q^{-1}_j,U\rangle f^t(U) - \langle \hat{Q}^{-1}_j,\hat{U}\rangle p^t(\hat{U})\bigr]\bigr|\\
	&\le \E_{g}\bigl[\bigl|\langle Q^{-1}_j,U\rangle(f^t(U) - p^t(\hat{U}))\bigr|\bigr] + \E_{g}\bigl[\bigl|(\langle Q^{-1}_j,U\rangle - \langle \hat{Q}^{-1}_j,\hat{U}\rangle)p^t(\hat{U})\bigr|\bigr]\\
	&\le \sqrt{\E_{U}[\langle Q^{-1}_j,U\rangle^2]\E_{g}\bigl[(f^t(U) - p^t(\hat{U}))^2\bigr]} + \sqrt{\E_{g}\bigl[(\langle Q^{-1}_j,U\rangle - \langle \hat{Q}^{-1}_j,\hat{U}\rangle)^2\bigr]\E_{\hat{U}}[p^t(\hat{U})^2]}.\end{align*}
		From before, we immediately know that
	\[\E_{g}[(f^t(U) - p^t(\hat{U}))^2] \le 2t^{3/2}L^2g_1(t)+2g_2(t)^2\]
	and \[\E_{\hat{U}}[p^t(\hat{U})^2]\le 2L^2\|\hat{Q}^t\|_\ast+2g_2(t)^2.\]
	So, it suffices to bound the two remaining terms. For the former, note that
	\[\E_U[\langle Q^{-1}_j,U\rangle^2] = \E_U\left[\sum_{i=1}^{t}\sum_{k=1}^{t}Q^{-1}_{ij}Q^{-1}_{jk}U_iU_k\right] = \sum_{i=1}^t\sum_{k=1}^{t}Q^{-1}_{ij}Q^{-1}_{jk}Q_{ik} = \sum_{k=1}^{t}Q^{-1}_{jk}\delta_{j=k} = Q^{-1}_{jj}.\]
	Since $Q\succeq I$, it follows that $Q^{-1}\preceq I$ and hence $Q^{-1}_{jj} \le \sqrt{t}$ (by using the fact that spectral ordering inequalities imply Frobenius norm inequalities on the positive-semidefinite matrices).

	For the latter term, write by the Almost Triangle Inequality that
	\begin{align*}\E_{g}\bigl[(\langle Q^{-1}_j,U\rangle - \langle \hat{Q}^{-1}_j,\hat{U}\rangle)^2\bigr] &\le 2\E_{g}[\langle Q^{-1}_j,U-\hat{U}\rangle^2] + 2\E_{\hat{U}}[\langle Q^{-1}_j - \hat{Q}^{-1}_j,\hat{U}\rangle^2]\\
	&\le 2\|Q^{-1}_j\|_2^2\E_{U,\hat{U}}[\|U-\hat{U}\|_2^2] + 2\|Q^{-1}_j-\hat{Q}^{-1}_j\|_2^2\E_{\hat{U}}[\|\hat{U}\|_2^2].\end{align*}
	By the Frobenius bound, we once more have that $\|Q^{-1}_j\|_2^2 \le t$. Similarly, we have that $\E_{\hat{U}}[\|\hat{U}\|_2^2]=\|\hat{Q}\|_\ast$. By the duality of norms, we find that $\E_{g}[\|U - \hat{U}\|_2^2] \le t^{3/2}g_1(t)$. Hence, the only remaining term is $\|Q^{-1}_j-\hat{Q}^{-1}_j\|_2^2$. Note that
	\[\|Q^{-1}_j-\hat{Q}^{-1}_j\|_2^2\le \|Q^{-1} - \hat{Q}^{-1}\|_F^2 = \|Q^{-1}(\hat{Q} - Q)\hat{Q}^{-1}\|_F^2\le \|Q^{-1}\|_{\mathsf{op}}^2\|\hat{Q}^{-1}\|_{\mathsf{op}}^2\|\hat{Q} - Q\|_F^2\le \frac{1}{(1-g_1(t))^2}\cdot g_1(t)^2 \le 2g_1(t)^2\]
	where we use that $Q\succeq I$ so $\hat{Q}\succeq (1 - g_1(t))I$.

	Unwinding bounds, it follows that
	\[\E_{g}\bigl[(\langle Q^{-1}_j,U\rangle - \langle \hat{Q}^{-1}_j,\hat{U}\rangle)^2\bigr] \le 2t^{5/2}g_1(t) + 4g_1(t)^2\cdot \|\hat Q\|_\ast.\]
	Then, substituting back gives that
	\begin{align*}\operatorname*{p-lim}_{n\rightarrow\infty} |b_{t,j} - \hat{b}_{t,j}| &\le \sqrt{2t^2L^2g_1(t)+2\sqrt{t}g_2(t)^2} + \sqrt{(2t^{5/2}g_1(t) + 4g_1(t)^2\|\hat Q^t\|_\ast)(2L^2\|\hat{Q}^t\|_\ast+2g_2(t)^2)}\\
	&\le 16\sqrt{L^2t^{9/2}g_1(t)}\end{align*}
	by using that $\|\hat Q^t\|_\ast \le 2t^2$ and $g_2(t) \le \frac 12$. This is our function $g_3(t)$.\qedhere
	\end{proof}

	The final subclaim gives the guarantees we are after: closeness of the AMP iterates.
	\begin{claim}
		The final bullet point is true. That is, $\operatorname*{p-lim}_{n\rightarrow\infty}\frac 1{\sqrt n}\|x^{t+1} - \hat{x}^{t+1}\|_2 \le g_4(t)$.
	\end{claim}
	\begin{proof}
	Let $K = \|X\|_{\mathsf{op}}$ (which is $O(1)$ by sub-Gaussian concentration.
	Write
	\begin{align*}
	\|x^{t+1} - \hat{x}^{t+1}\|_2 &= \Biggl\|X\left(f^t(x^0,x^1,\ldots,x^t) - p^t(\hat{x}^0,\hat{x}^1,\ldots,\hat{x}^t)\right) + \sum_{j=1}^{t-1}\left(b_{t,j}f^{j-1}(x)-\hat{b}_{t,j}p^{j-1}(\hat x)\right)\Biggr\|_2	\\
	&\le \|X\|_{\mathsf{op}}\|f^t(x) - p^t(\hat x)\|_2 + \sum_{j=1}^{t-1}\|b_{t,j}f^{j-1}(x) - \hat{b}_{t,j}p^{j-1}(\hat x)\|_2\\
	&\le K\|f^t(x) - p^t(\hat x)\|_2 + \sum_{j=1}^{t-1}\left(|b_{t,j}-\hat{b}_{t,j}|\|p^{j-1}(\hat x)\|_2 + |b_{t,j}|\|f^{j-1}(x) - p^{j-1}(\hat{x})\|_2\right)\\
	&\le K\|f^t(x) - p^t(\hat x)\|_2 + \sum_{j=1}^{t-1}(g_3(t)\cdot \|p^{j-1}(\hat x)\|_2 + L\cdot \|f^{j-1}(x) - p^{j-1}(\hat{x})\|_2)\\
	&\le (K + tL)\|f^t(x) - p^t(\hat x)\|_2 + tg_3(t)\|p^{t-1}(\hat x)\|_2\\
	&\le (K + tL)\|f^t(x) - f^t(\hat x)\|_2 + (K + tL)\|f^t(\hat x) - p^t(\hat x)\|_2 + tg_3(t)\|p^{t-1}(\hat x)\|_2
	\end{align*}
	where in the second to last step we assume that all errors pile up to the last iteration.

	Let's bound the three remaining terms. The first follows directly by induction:
	\begin{align*}
		\frac 1n\|f^t(x) - f^t(\hat{x})\|_2^2 &= \frac 1n\sum_{i=1}^{n}(f^t(x_i) - f^t(\hat{x}_i))^2\\
		&\le \frac {L^2}n\sum_{i=1}^{n}\sum_{j=1}^{t}(x^j_i - \hat{x}^j_i)^2\\
		&= \frac {L^2}n\sum_{j=1}^{t}\|x^j - \hat{x}^j\|_2^2\\
		&\le L^2\sum_{j=1}^{t}g_4(t-1)^2\\
		&\le tL^2g_4(t-1)^2.
	\end{align*}
	The second follows by the second subclaim:
	\begin{align*}
		\operatorname*{p-lim}_{n\rightarrow\infty}\frac 1n\|f^t(\hat x) - p^t(\hat x)\|_2^2 = \E_{\hat U}[(f^t(\hat U) - p^t(\hat U))^2] \le g_2(t)^2.
	\end{align*}
	Finally, the third follows by a sub-argument of the first subclaim:
	\begin{align*}
		\operatorname*{p-lim}_{n\rightarrow\infty}\frac 1n\|p^{t-1}(\hat x)\|_2^2 = \E_{\hat U}[p^{t-1}(\hat U)^2] \le 2L^2 \|\hat{Q}^{t-1}\|_{\ast} + 2g_2(t - 1)^2 \le 9tL^2.
	\end{align*}
	Putting everything together, we find that
	\[\operatorname*{p-lim}_{n\rightarrow\infty}\frac 1{\sqrt n}\|x^{t+1} - \hat{x}^{t+1}\|_2 \le (K + tL)(t^{1/2}Lg_4(t-1) + g_2(t)) + 3t^{3/2}Lg_3(t).\]
	This is our function $g_4(t)$.
	\end{proof}

	Finally, we claim that we can choose increasing functions $g_1,g_2,g_3,g_4$ such that $g_4(T) = \delta$ and $g_2(t)$ is not too quickly growing.

	To review the current expressions regarding $g_1,g_2,g_3,g_4$, we have:
	\begin{align*}
		g_1(t)^2 &= 64t^3L^2(t^{3/2}g_1(t-1)+g_2(t-1)^2)\\
		g_3(t)^2 &= 256L^2t^{9/2}g_1(t)\\
		g_4(t) &= (K + tL)(t^{1/2} Lg_4(t-1) + g_2(t)) + 3t^{3/2}Lg_3(t).
	\end{align*}
	We will choose functions such that $g_1(t),g_4(t-1) \ge g_2(t)$ and $g_4(t-1)\ge g_3(t)$ to simplify the first expression to $g_1(t) = 16t^{9/4}L\cdot \sqrt{g_1(t-1)}$ and the last expression to $g_4(t) = 5t^{1/2}L(K + tL)\cdot g_4(t-1)$.

	To find $g_1(t)$, rewrite this as $g_1(t) = S\sqrt{g_1(t-1)}$ where we want $g_1(T) = \eta$ and we fix $t = T$ in all iterations as an upper bound (to keep $S$ a constant). Stepping back the recursion, this implies that $g_1(t - 1) = \frac{\eta^2}{S^2}$, $g_1(t - 2) = \frac{\eta^4}{S^6}$, and successively back to $g_1(0) = \frac{\eta^{2^{T}}}{S^{2^{T+1} - 2}}$.

	Next, write $g_4(t) = Cg_4(t - 1)$ (by treating $t = T$ as a constant). Then, $g_4(T) = C^Tg_4(0)$. Since we want $g_4(T) = \delta$, we must have $g_4(0) = \frac{\delta}{C^T}$.

	By hypothesis, we required $g_4(t-1) \ge g_3(t)$. Rewriting, this implies
	\[\delta \cdot C^{t + 1 - T} \ge S\sqrt{g_1(t)} = S\sqrt{\frac{\eta^{2^{T - t}}}{S^{2^{T+1-t} - 2}}} \implies \eta^{2^{T - t}}\le \delta^2 \cdot C^{2(t+1-T)} \cdot S^{2^{T + 1 - t}}.\]
	Take the $2^{T - t}$'th root: this yields that $\eta \le \min_{0\le t\le T}S^2 \cdot \left(\delta^2 \cdot C^{2(t+1-T)}\right)^{\frac 1{2^{T - t}}}$. For small enough $\delta > 0$, this minimum is obtained at $t = T$, which gives $\eta = (SC\delta)^2$.

	Then, we may choose $g_2(t) = g_1(t)$: this gives
	\[g_2(t) = \frac{(SC\delta)^{2^{T + 1 - t}}}{S^{2^{T + 1 -t} - 2}} = S^2\cdot C^{2^{T+1-t}}\cdot \delta^{2^{T+1-t}} = 256t^{9/2}L^2\cdot (5t^{1/2}\cdot L (K +tL))^{2^{T+1-t}}\cdot \delta^{2^{T+1-t}}.\]

	Simplifying further yields our final bound of $g_2(t) \le (256t^{9/2}L(K + L)\delta)^{2^{T+2-t}}$ being sufficient to get $g_4(T) \le \delta$.

	Therefore, the total degree of $p^t$ is at most $\widetilde{O}\biggl((256t^{9/2}L(K + L)/\delta)^{2^{T+5-t}}\biggr)$ as desired.
\end{proof}

We can finally prove the approximability of AMP as required by our algorithm.

\restateprop{prop:final-approx}
\begin{proof}
For the first statement, by state evolution, it suffices to show that $\|\hat{x}^{t+1} - x^{t+1}\|_2 \le \tfrac{1}{2}\delta \sqrt{n}$. 
By \pref{lem:approx-poly}, we can achieve this with polynomials $p^t$ of total $\widetilde{O}\biggl((512t^{9/2}L(K + L)/\delta)^{2^{T+5-t}}\biggr)$. 
The only remaining question is to figure out the total compounding degree of the $p^t$ to form $p(X)$.

Let $d(t)$ denote the total degree in $X$ required to form $\hat{x}^{t}$ and $\deg p^t$ the total degree of $p^t$. 
We begin with $d(0) = 0$ and $d(1) = 1$. Then, by iteration it follows that
\[d(t + 1) \le \max\left(1 + \deg p^t\cdot \sum_{j=0}^{t}d(j), \max_{1\le k< t}\left(\deg p^t - 1 + \deg p^k\cdot \sum_{j=1}^{k}d(j)\right)\right) \le \deg p^{t}\left(1 + \sum_{j=1}^{t}d(j)\right).\]
Upper bounding $\deg p^t\le S$ then yields that $d(t + 1) = S(1 + \sum_{j=1}^{t-1}d(j) + d(t)) = (1 + S)d(t)$, and hence $d(t + 1) = (1 +S)^{t-1}\cdot d(2) = 2S(1 + S)^{t-1}$.

Finally, this gives us the degree bound we desire: it is
\[d(T) \le \widetilde{O}\biggl((512T^{9/2}L(K + L)\delta)^{T\cdot 2^{T+5}}\biggr) \le \widetilde{O}\biggl((TL(K + L)/\delta)^{2^{2T}}\biggr)\]
for $T\ge 10$. 
Hence, we are done.

What remains is showing the second statement. 
By the GAMP state evolution (see~\cite[Theorem 1]{JM13}) applied to the pair of iterates $(x^t, \hat{x}^t)$ it follows that almost surely $\frac 1n\|\vAMP - P(X)\|_2^2 \rightarrow C$ for some constant $0 \le C \le \frac{1}{4}\delta^2$. 
Thus, this implies that $\frac 1n \E[\|\vAMP - P(X)\|_2^2] \rightarrow C$ as well. 
Taking $n_0$ large enough for $\frac 1n \E[\|\vAMP - P(X)\|_2^2] \le 2C \le \delta^2$ to hold then proves the claim.
\end{proof}

%% file: main.bbl
\newcommand{\etalchar}[1]{$^{#1}$}
\begin{thebibliography}{KMS{\etalchar{+}}12}

\bibitem[AGZ10]{anderson2010introduction}
Greg~W Anderson, Alice Guionnet, and Ofer Zeitouni.
\newblock {\em An introduction to random matrices}.
\newblock Number 118. Cambridge university press, 2010.

\bibitem[ARV09]{ARV09}
Sanjeev Arora, Satish Rao, and Umesh Vazirani.
\newblock Expander flows, geometric embeddings and graph partitioning.
\newblock {\em Journal of the {ACM} ({JACM})}, 56(2):1--37, 2009.

\bibitem[BBH{\etalchar{+}}21]{BBHLS21}
Matthew~S Brennan, Guy Bresler, Sam Hopkins, Jerry Li, and Tselil Schramm.
\newblock Statistical query algorithms and low degree tests are almost
  equivalent.
\newblock In {\em Conference on Learning Theory}, pages 774--774. {PMLR}, 2021.

\bibitem[BBK{\etalchar{+}}21]{BBKMW21}
Afonso~S Bandeira, Jess Banks, Dmitriy Kunisky, Christopher Moore, and Alex
  Wein.
\newblock Spectral planting and the hardness of refuting cuts, colorability,
  and communities in random graphs.
\newblock In {\em Conference on Learning Theory}, pages 410--473. {PMLR}, 2021.

\bibitem[BKW20]{BKW20}
Afonso~S Bandeira, Dmitriy Kunisky, and Alexander~S Wein.
\newblock Computational hardness of certifying bounds on constrained pca
  problems.
\newblock In {\em 11th Innovations in Theoretical Computer Science Conference
  ({ITCS} 2020)}, volume 151, 2020.

\bibitem[BKW22]{BKW22}
Afonso Bandeira, Dmitriy Kunisky, and Alexander Wein.
\newblock Average-case integrality gap for non-negative principal component
  analysis.
\newblock In {\em Mathematical and Scientific Machine Learning}, pages
  153--171. PMLR, 2022.

\bibitem[BLM15]{bayati2015universality}
Mohsen Bayati, Marc Lelarge, and Andrea Montanari.
\newblock Universality in polytope phase transitions and message passing
  algorithms.
\newblock 2015.

\bibitem[BM11]{BM11}
Mohsen Bayati and Andrea Montanari.
\newblock The dynamics of message passing on dense graphs, with applications to
  compressed sensing.
\newblock {\em {IEEE} Transactions on Information Theory}, 57(2):764--785,
  2011.

\bibitem[BM16]{BM16}
Boaz Barak and Ankur Moitra.
\newblock Noisy tensor completion via the sum-of-squares hierarchy.
\newblock In {\em Conference on Learning Theory}, pages 417--445. PMLR, 2016.

\bibitem[BMR21]{BMR21}
Jess Banks, Sidhanth Mohanty, and Prasad Raghavendra.
\newblock Local statistics, semidefinite programming, and community detection.
\newblock In {\em Proceedings of the 2021 {ACM}-{SIAM} Symposium on Discrete
  Algorithms ({SODA})}, pages 1298--1316. {SIAM}, 2021.

\bibitem[Bol14]{Bolt14}
Erwin Bolthausen.
\newblock An iterative construction of solutions of the {TAP} equations for the
  {S}herrington--{K}irkpatrick model.
\newblock {\em Communications in Mathematical Physics}, 325(1):333--366, 2014.

\bibitem[CM22]{CM22}
Michael Celentano and Andrea Montanari.
\newblock Fundamental barriers to high-dimensional regression with convex
  penalties.
\newblock {\em The Annals of Statistics}, 50(1):170--196, 2022.

\bibitem[CMW20]{CMW20}
Michael Celentano, Andrea Montanari, and Yuchen Wu.
\newblock The estimation error of general first order methods.
\newblock In {\em Conference on Learning Theory}, pages 1078--1141. PMLR, 2020.

\bibitem[CZK14]{CZK14}
Francesco Caltagirone, Lenka Zdeborov{\'a}, and Florent Krzakala.
\newblock On convergence of approximate message passing.
\newblock In {\em 2014 {IEEE} International Symposium on Information Theory},
  pages 1812--1816. {IEEE}, 2014.

\bibitem[DdNS22]{DORS21}
Jingqiu Ding, Tommaso d'Orsi, Rajai Nasser, and David Steurer.
\newblock Robust recovery for stochastic block models.
\newblock In {\em 2021 {IEEE} 62nd Annual Symposium on Foundations of Computer
  Science ({FOCS})}, pages 387--394. {IEEE}, 2022.

\bibitem[DM14]{DM14}
Yash Deshpande and Andrea Montanari.
\newblock Information-theoretically optimal sparse pca.
\newblock In {\em 2014 {IEEE} International Symposium on Information Theory},
  pages 2197--2201. {IEEE}, 2014.

\bibitem[DMM09]{DMM09}
David~L Donoho, Arian Maleki, and Andrea Montanari.
\newblock Message-passing algorithms for compressed sensing.
\newblock {\em Proceedings of the National Academy of Sciences},
  106(45):18914--18919, 2009.

\bibitem[FVRS22]{FVRR22}
Oliver~Y Feng, Ramji Venkataramanan, Cynthia Rush, and Richard~J Samworth.
\newblock A unifying tutorial on approximate message passing.
\newblock {\em Foundations and Trends{\textregistered} in Machine Learning},
  15(4):335--536, 2022.

\bibitem[GJJ{\etalchar{+}}20]{GJJPR20}
Mrinalkanti Ghosh, Fernando~Granha Jeronimo, Chris Jones, Aaron Potechin, and
  Goutham Rajendran.
\newblock Sum-of-squares lower bounds for {S}herrington-{K}irkpatrick via
  planted affine planes.
\newblock In {\em 2020 {IEEE} 61st Annual Symposium on Foundations of Computer
  Science ({FOCS})}, pages 954--965. {IEEE}, 2020.

\bibitem[GSS21]{gotze2021concentration}
Friedrich G{\"o}tze, Holger Sambale, and Arthur Sinulis.
\newblock Concentration inequalities for polynomials in
  $\alpha$-sub-exponential random variables.
\newblock {\em Electronic Journal of Probability}, 26, Jan 2021.

\bibitem[HKP{\etalchar{+}}17]{HKPRSS17}
Samuel~B Hopkins, Pravesh~K Kothari, Aaron Potechin, Prasad Raghavendra, Tselil
  Schramm, and David Steurer.
\newblock The power of sum-of-squares for detecting hidden structures.
\newblock In {\em 2017 {IEEE} 58th Annual Symposium on Foundations of Computer
  Science ({FOCS})}, pages 720--731. {IEEE}, 2017.

\bibitem[HL18]{HL18}
Samuel~B Hopkins and Jerry Li.
\newblock Mixture models, robustness, and sum of squares proofs.
\newblock In {\em Proceedings of the 50th Annual {ACM} {SIGACT} Symposium on
  Theory of Computing}, pages 1021--1034, 2018.

\bibitem[HS17]{HS17}
Samuel~B Hopkins and David Steurer.
\newblock Efficient {B}ayesian estimation from few samples: community detection
  and related problems.
\newblock In {\em 2017 {IEEE} 58th Annual Symposium on Foundations of Computer
  Science ({FOCS})}, pages 379--390. {IEEE}, 2017.

\bibitem[HSS15]{HSS15}
Samuel~B Hopkins, Jonathan Shi, and David Steurer.
\newblock Tensor principal component analysis via sum-of-square proofs.
\newblock In {\em Conference on Learning Theory}, pages 956--1006. {PMLR},
  2015.

\bibitem[JM13a]{JM13}
Adel Javanmard and Andrea Montanari.
\newblock State evolution for general approximate message passing algorithms,
  with applications to spatial coupling.
\newblock {\em Information and Inference: A Journal of the {IMA}},
  2(2):115--144, 2013.

\bibitem[JM13b]{javanmard2013state}
Adel Javanmard and Andrea Montanari.
\newblock State evolution for general approximate message passing algorithms,
  with applications to spatial coupling.
\newblock {\em Information and Inference: A Journal of the IMA}, 2(2):115--144,
  2013.

\bibitem[KB21]{KB21}
Dmitriy Kunisky and Afonso~S Bandeira.
\newblock A tight degree 4 sum-of-squares lower bound for the
  {S}herrington--{K}irkpatrick hamiltonian.
\newblock {\em Mathematical Programming}, 190(1):721--759, 2021.

\bibitem[KMS{\etalchar{+}}12]{KMSSZ12}
Florent Krzakala, Marc M{\'e}zard, Francois Sausset, Yifan Sun, and Lenka
  Zdeborov{\'a}.
\newblock Probabilistic reconstruction in compressed sensing: algorithms, phase
  diagrams, and threshold achieving matrices.
\newblock {\em Journal of Statistical Mechanics: Theory and Experiment},
  2012(08):P08009, 2012.

\bibitem[KSS18]{KSS18}
Pravesh~K Kothari, Jacob Steinhardt, and David Steurer.
\newblock Robust moment estimation and improved clustering via sum of squares.
\newblock In {\em Proceedings of the 50th Annual {ACM} {SIGACT} Symposium on
  Theory of Computing}, pages 1035--1046, 2018.

\bibitem[LM22]{LM22}
Allen Liu and Ankur Moitra.
\newblock Minimax rates for robust community detection.
\newblock {\em arXiv preprint arXiv:2207.11903}, 2022.

\bibitem[LW22]{LW22}
Gen Li and Yuting Wei.
\newblock A non-asymptotic framework for approximate message passing in spiked
  models.
\newblock {\em arXiv preprint arXiv:2208.03313}, 2022.

\bibitem[Mon12]{Mon12}
Andrea Montanari.
\newblock Graphical models concepts in compressed sensing.
\newblock {\em Compressed Sensing: Theory and Applications}, page 394, 2012.

\bibitem[Mon21]{Mon21}
Andrea Montanari.
\newblock Optimization of the {S}herrington--{K}irkpatrick hamiltonian.
\newblock {\em {SIAM} Journal on Computing}, (0):{FOCS}19--1, 2021.

\bibitem[MR15]{MR15}
Andrea Montanari and Emile Richard.
\newblock Non-negative principal component analysis: Message passing algorithms
  and sharp asymptotics.
\newblock {\em {IEEE} Transactions on Information Theory}, 62(3):1458--1484,
  2015.

\bibitem[MRX20]{MRX20}
Sidhanth Mohanty, Prasad Raghavendra, and Jeff Xu.
\newblock Lifting sum-of-squares lower bounds: degree-2 to degree-4.
\newblock In {\em Proceedings of the 52nd Annual {ACM} {SIGACT} Symposium on
  Theory of Computing}, pages 840--853, 2020.

\bibitem[MS16]{MS16}
Andrea Montanari and Subhabrata Sen.
\newblock Semidefinite programs on sparse random graphs and their application
  to community detection.
\newblock In {\em Proceedings of the forty-eighth annual {ACM} symposium on
  Theory of Computing}, pages 814--827, 2016.

\bibitem[MSS16]{MSS16}
Tengyu Ma, Jonathan Shi, and David Steurer.
\newblock Polynomial-time tensor decompositions with sum-of-squares.
\newblock In {\em 2016 {IEEE} 57th Annual Symposium on Foundations of Computer
  Science ({FOCS})}, pages 438--446. {IEEE}, 2016.

\bibitem[Mus17]{musin2017weighted}
I~Kh Musin.
\newblock On weighted polynomial approximation.
\newblock {\em arXiv preprint arXiv:1712.09314}, 2017.

\bibitem[MV21]{MV21}
Andrea Montanari and Ramji Venkataramanan.
\newblock Estimation of low-rank matrices via approximate message passing.
\newblock {\em The Annals of Statistics}, 49(1):321--345, 2021.

\bibitem[MW22]{MW22}
Andrea Montanari and Alexander~S Wein.
\newblock Equivalence of approximate message passing and low-degree polynomials
  in rank-one matrix estimation.
\newblock {\em arXiv preprint arXiv:2212.06996}, 2022.

\bibitem[Par80]{Par80}
Giorgio Parisi.
\newblock A sequence of approximated solutions to the sk model for spin
  glasses.
\newblock {\em Journal of Physics A: Mathematical and General}, 13(4):L115,
  1980.

\bibitem[Rag08]{Rag08}
Prasad Raghavendra.
\newblock Optimal algorithms and inapproximability results for every {CSP}?
\newblock In {\em Proceedings of the fortieth annual {ACM} symposium on Theory
  of computing}, pages 245--254, 2008.

\bibitem[RM14]{RM14}
Emile Richard and Andrea Montanari.
\newblock A statistical model for tensor pca.
\newblock {\em Advances in neural information processing systems}, 27, 2014.

\bibitem[RSFS19]{RSFS19}
Sundeep Rangan, Philip Schniter, Alyson~K Fletcher, and Subrata Sarkar.
\newblock On the convergence of approximate message passing with arbitrary
  matrices.
\newblock {\em {IEEE} Transactions on Information Theory}, 65(9):5339--5351,
  2019.

\bibitem[RSS18]{RSS18}
Prasad Raghavendra, Tselil Schramm, and David Steurer.
\newblock High dimensional estimation via sum-of-squares proofs.
\newblock In {\em Proceedings of the International Congress of Mathematicians:
  {R}io de {J}aneiro 2018}, pages 3389--3423. World Scientific, 2018.

\bibitem[Tal06]{Tal06}
Michel Talagrand.
\newblock The parisi formula.
\newblock {\em Annals of mathematics}, pages 221--263, 2006.

\bibitem[WAM19]{WAM19}
Alexander~S. Wein, Ahmed~El Alaoui, and Cristopher Moore.
\newblock The {K}ikuchi hierarchy and tensor {PCA}.
\newblock In David Zuckerman, editor, {\em 60th {IEEE} Annual Symposium on
  Foundations of Computer Science}, pages 1446--1468. {IEEE} Computer Society,
  2019.

\end{thebibliography}
